\documentclass{elsarticle} 

\usepackage[utf8]{inputenc}
\usepackage[T1]{fontenc}
\usepackage{color}
\usepackage{listings}
\usepackage{amsmath,amssymb,stmaryrd,amsthm}
\usepackage{xspace}
\usepackage{mathtools}
\usepackage{tikz}
\usepackage{wrapfig}
\usepackage{hyperref}
\usepackage[inline]{enumitem}
\usepackage{mathpartir}
\usepackage[noabbrev]{cleveref} % http://tug.ctan.org/tex-archive/macros/latex/contrib/cleveref/cleveref.pdf

\newcommand{\freest}{\textsc{FreeST}\xspace}
\newcommand{\fmusession}{$F^{\mu;}$}

\newcommand{\changed}[1]{#1}
\newcommand{\Keyword}[1]{\mathsf{#1}}
\newcommand{\link}{\Keyword{lin}}
\newcommand{\unk}{\Keyword{un}}

\newcommand{\dualofk}{\Keyword{dualof}}
\newcommand{\skipk}{\Keyword{Skip}}
\newcommand{\ifk}{\Keyword{if}}
\newcommand{\thenk}{\Keyword{then}}
\newcommand{\elsek}{\Keyword{else}}
\newcommand{\letk}{\Keyword{let}}
\newcommand{\ink}{\Keyword{in}}
\newcommand{\newk}{\Keyword{new}}
\newcommand{\sendk}{\Keyword{send}}

\newcommand{\receivek}{\Keyword{receive}}
\newcommand{\selectk}{\Keyword{select}}
\newcommand{\matchk}{\Keyword{match}}
\newcommand{\withk}{\Keyword{with}}
\newcommand{\forkk}{\Keyword{fork}}
\newcommand{\casek}{\Keyword{case}}
\newcommand{\ofk}{\Keyword{of}}

\newcommand{\bytek}{\Keyword{Byte}}
\newcommand{\intk}{\Keyword{Int}}
\newcommand{\boolk}{\Keyword{Bool}}

\newcommand{\unitk}{\unite}

\newcommand{\prekind}{\upsilon}%{\kind} prekinds and kinds coincide when we remove multiplicities
\newcommand{\kind}{\kappa}
\newcommand{\kindstyle}[1]{\textsc{\textbf{#1}}}
\newcommand{\kinds}{\kindstyle s}
\newcommand{\kindt}{\kindstyle t}
\newcommand{\kindm}{\kindstyle m}
\newcommand{\kindk}[2]{{#1}^{#2}}
\newcommand{\kindsl}{\kinds^\lin}
\newcommand{\kindsu}{\kinds^\un}
\newcommand{\kindtl}{\kindt^\lin}
\newcommand{\kindtu}{\kindt^\un}
\newcommand{\kindml}{\kindm^\lin}
\newcommand{\kindmu}{\kindm^\un}
\newcommand{\kindtm}{\kindt^m}

\newcommand{\un}{\unk}
\newcommand{\lin}{\link}

\newcommand{\Label}[1]{\Keyword{#1}}
\newcommand{\fstl}{\Label{fst}}
\newcommand{\sndl}{\Label{snd}}
\newcommand{\pushl}{\Label{Push}}
\newcommand{\popl}{\Label{Pop}}
\newcommand{\donel}{\Label{Done}}
\newcommand{\leafl}{\Label{Leaf}}
\newcommand{\nodel}{\Label{Node}}

\newcommand{\Name}[1]{\Keyword{#1}}

\newcommand{\estackn}{\Name{EStack}}

\newcommand{\recordm}[4]{\{{#1}{#2}{#3}\}_{#4}} % Manual records indexed by sets L of labels, used for internal/external choice types and case processes
\newcommand{\record}[4]{\recordm{#1}{#2}{{#3}_{#1}}{{#1}\in{#4}}} % Records indexed by sets L of labels, used for internal/external choice types and case processes
\newcommand{\recorda}[3]{\{{#1}{#2}{#3}_{#1}\}} % Record abbrev: no l \in L
\newcommand{\recordt}[3]{\record{#1}{\colon}{#2}{#3}} % Records indexed by sets L of labels, used for internal/external choice types and case processes
\newcommand{\variantt}[3]{\langle{#1}\colon{#2}_{#1}\rangle_{{#1}\in{#3}}} % Data types

\newcommand{\variantta}[2]{\langle{#1}\colon{#2}_{#1}\rangle} % Data types abbrev: no l \in L
\newcommand{\skipt}{\skipk}
\newcommand{\semit}[2]{#1;#2}
\newcommand{\sout}[1]{\mathop!#1}
\newcommand{\sint}[1]{\mathop?#1}
\newcommand{\msgt}[1]{\sharp{#1}}
\newcommand{\choicet}[1]{\star{#1}}

\newcommand{\ichoicet}[1]{\oplus{#1}}
\newcommand{\echoicet}[1]{\&{#1}}
\newcommand{\funt}[3][m]{{#2}\rightarrow_{#1}{#3}}
\newcommand{\unfunt}[2]{#1\rightarrow#2} % Examples only; use \funt[\un] otherwise
\newcommand{\pairt}[2]{\{\fstl\colon{#1},\sndl\colon{#2}\}}

\newcommand{\oldpairt}[2]{#1\times #2}
\newcommand{\rect}[3]{\mu{#1}\colon{#2}.{#3}}
\newcommand{\typescheme}[2]{\forall\,{#1}\,.\,{#2}} % was: \Rightarrow
\newcommand{\forallt}[3]{\typescheme{{#1}\colon{#2}}{#3}}

\newcommand{\unitt}[1][m]{()_{#1}}
\newcommand{\intt}{\Keyword{Int}}
\newcommand{\chart}{\Keyword{Char}}
\newcommand{\boolt}{\Keyword{Bool}}
\newcommand{\Endk}{\Keyword{end}}

\newcommand{\unite}[1][m]{()_{#1}}
\newcommand{\recordp}[3]{\record{#1}{\rightarrow}{#2}{#3}} % Record pattern indexed by sets L of labels, used for internal/external choice types and case processes
\newcommand{\recordpa}[2]{\recorda{#1}{\rightarrow}{#2}} % Record abbrev: no l \in L
\newcommand{\recorde}[3]{\record{#1}{=}{#2}{#3}} % Record expression
\newcommand{{\abse}}[4][m]{\lambda_{#1}{#2}\colon{#3}.{#4}}
\newcommand{\appe}[2]{{#1}\,{#2}}
\newcommand{\tabse}[3]{\Lambda{#1}\colon{#2}.{#3}}
\newcommand{\tappe}[2]{#1\,[#2]}
\newcommand{\conde}[3]{\ifk\,#1\,\thenk\,#2\,\elsek \, #3}

\newcommand{\binlete}[4]{\letk\,(#1, #2) = #3\,\ink\,#4}
\newcommand{\unlete}[3]{\letk\,#1 = #2\,\ink\,#3}
\newcommand{\lete}[5]{\letk\,\recorde {#1}{#2}{#3} = #4\,\ink\,#5}

\newcommand{\selecte}[2]{\selectk\,#1\,{#2}}
\newcommand{\matche}[2]{\matchk\,#1\,\withk\,#2}
\newcommand{\forke}[1]{\forkk\,#1}
\newcommand{\casee}[2]{\casek\,#1\,\ofk\,#2}
\newcommand{\injecte}[2]{{#1}\,{#2}}

\newcommand{\newe}[1]{\newk\,{#1}}

\newcommand{\PAR}{\mid}
\newcommand{\PROC}[1]{\langle{#1}\rangle}
\newcommand{\NU}[2]{(\nu{#1}{#2})}

\newcommand{\Empty}{\varepsilon}

\newcommand{\infrule}[3]{\inferrule* [lab=#1]{#2}{#3}}
\newcommand{\axiom}[2]{\infrule {#1}{}{#2}}

\newcommand{\rulename}[2]{{#1}-{#2}\xspace}
\newcommand{\rulenameeq}[1]{\rulename{Q}{#1}} %equivalence Q-
\newcommand{\rulenamelts}[1]{\rulename{L}{#1}} %LTS L-
\newcommand{\rulenameterm}[1]{\rulename{Te}{#1}} %Terminated Te-
\newcommand{\rulenameform}[1]{\rulename{F}{#1}} % context formation F-
\newcommand{\rulenamesplit}[1]{\rulename{S}{#1}} % Split S-
\newcommand{\rulenamecontr}[1]{\rulename{C}{#1}} %Contractivity C-
\newcommand{\rulenamekind}[1]{\rulename{K}{#1}} %Type formation K-
\newcommand{\rulenametype}[1]{\rulename{T}{#1}} %Typing expressions T- 
\newcommand{\rulenameproc}[1]{\rulename{P}{#1}} %Typing processes P-
\newcommand{\rulenameexpred}[1]{\rulename{E}{#1}} %Expression Reduction E-
\newcommand{\rulenameprocred}[1]{\rulename{R}{#1}} %Process Reduction R-
\newcommand{\rulenameltsbpa}[1]{\rulename{BPA}{#1}} %LTS BPA processes LBPA-
\newcommand{\rulenamealgkind}[1]{\rulename{KA}{#1}} %Algorithmic kinding KA- 
\newcommand{\rulenamealgtype}[1]{\rulename{TA}{#1}} %Algorithmic typing TA- 

\newcommand{\rulenametermSkip}{\rulenameterm{Skip}}
\newcommand{\rulenametermSeq}{\rulenameterm{Seq}}
\newcommand{\rulenametermRec}{\rulenameterm{Rec}}

\newcommand{\rulenamecontrSeqOne}{\rulenamecontr{Seq2}}
\newcommand{\rulenamecontrSeqTwo}{\rulenamecontr{Seq1}}
\newcommand{\rulenamecontrRec}{\rulenamecontr{Rec}}
\newcommand{\rulenamecontrPoly}{\rulenamecontr{TAbs}}
\newcommand{\rulenamecontrVar}{\rulenamecontr{Var}}
\newcommand{\rulenamecontrOther}{\rulenamecontr{Other}}

\newcommand{\rulenamekindType}{\rulenamekind{Type}}
\newcommand{\rulenamekindUnit}{\rulenamekind{Unit}}
\newcommand{\rulenamekindSkip}{\rulenamekind{Skip}}
\newcommand{\rulenamekindVar}{\rulenamekind{Var}}
\newcommand{\rulenamekindMsg}{\rulenamekind{Msg}}
\newcommand{\rulenamekindCh}{\rulenamekind{Ch}}
\newcommand{\rulenamekindSeq}{\rulenamekind{Seq}}
\newcommand{\rulenamekindArrow}{\rulenamekind{Arrow}}
\newcommand{\rulenamekindRcd}{\rulenamekind{Record}}
\newcommand{\rulenamekindData}{\rulenamekind{Variant}}
\newcommand{\rulenamekindRec}{\rulenamekind{Rec}}
\newcommand{\rulenamekindPoly}{\rulenamekind{TAbs}}
\newcommand{\rulenamekindSub}{\rulenamekind{Sub}}

\newcommand{\rulenameltsVar}{\rulenamelts{Var}}

\newcommand{\rulenameltsMsg}{\rulenamelts{Msg}}
\newcommand{\rulenameltsCh}{\rulenamelts{Ch}}
\newcommand{\rulenameltsSeqOne}{\rulenamelts{Seq1}}
\newcommand{\rulenameltsSeqTwo}{\rulenamelts{Seq2}}
\newcommand{\rulenameltsRec}{\rulenamelts{Rec}}

\newcommand{\rulenameltsbpaVar}{\rulenameltsbpa{Var}}

\newcommand{\rulenameltsbpaChOne}{\rulenameltsbpa{Ch1}}
\newcommand{\rulenameltsbpaChTwo}{\rulenameltsbpa{Ch2}}
\newcommand{\rulenameltsbpaSeqOne}{\rulenameltsbpa{Seq1}}
\newcommand{\rulenameltsbpaSeqTwo}{\rulenameltsbpa{Seq2}}
\newcommand{\rulenameltsbpaDef}{\rulenameltsbpa{Def}}

\newcommand{\rulenameeqUnit}{\rulenameeq{Unit}}
\newcommand{\rulenameeqST}{\rulenameeq{ST}}
\newcommand{\rulenameeqFix}{\rulenameeq{Fix}}
\newcommand{\rulenameeqVar}{\rulenameeq{Var}}
\newcommand{\rulenameeqArrow}{\rulenameeq{Arrow}}
\newcommand{\rulenameeqRcd}{\rulenameeq{Record}}
\newcommand{\rulenameeqData}{\rulenameeq{Variant}}
\newcommand{\rulenameeqRecL}{\rulenameeq{RecL}}
\newcommand{\rulenameeqRecR}{\rulenameeq{RecR}}
\newcommand{\rulenameeqPoly}{\rulenameeq{TAbs}}

\newcommand\rulenameformEmpty{\rulenameform{Empty}}
\newcommand\rulenameformExt{\rulenameform{Ext}}

\newcommand\rulenamesplitEmpty{\rulenamesplit{Empty}}
\newcommand\rulenamesplitUnr{\rulenamesplit{Unr}}
\newcommand\rulenamesplitLeft{\rulenamesplit{Left}}
\newcommand\rulenamesplitRight{\rulenamesplit{Right}}

\newcommand{\rulenametypeConst}{\rulenametype{Const}}
\newcommand{\rulenametypeVar}{\rulenametype{Var}}
\newcommand{\rulenametypePoly}{\rulenametype{TAbs}}
\newcommand{\rulenametypeTApp}{\rulenametype{TApp}}
\newcommand{\rulenametypeAbs}{\rulenametype{Abs}}
\newcommand{\rulenametypeApp}{\rulenametype{App}}
\newcommand{\rulenametypeRcd}{\rulenametype{Record}}
\newcommand{\rulenametypeLet}{\rulenametype{RcdElim}}
\newcommand{\rulenametypeUnitElim}{\rulenametype{UnitElim}}
\newcommand{\rulenametypeVariant}{\rulenametype{Variant}}
\newcommand{\rulenametypeCase}{\rulenametype{Case}}
\newcommand{\rulenametypeNew}{\rulenametype{New}}
\newcommand{\rulenametypeSel}{\rulenametype{Sel}}
\newcommand{\rulenametypeMatch}{\rulenametype{Match}}
\newcommand{\rulenametypeEq}{\rulenametype{Eq}}

\newcommand{\rulenameprocExp}{\rulenameproc{Exp}}
\newcommand{\rulenameprocPar}{\rulenameproc{Par}}
\newcommand{\rulenameprocNew}{\rulenameproc{New}}

\newcommand{\rulenameexpredApp}{\rulenameexpred{App}}
\newcommand{\rulenameexpredLet}{\rulenameexpred{RcdElim}}
\newcommand{\rulenameexpredUnitElim}{\rulenameexpred{UnitElim}}
\newcommand{\rulenameexpredTApp}{\rulenameexpred{TApp}}
\newcommand{\rulenameexpredCase}{\rulenameexpred{Case}}
\newcommand{\rulenameexpredCtx}{\rulenameexpred{Ctx}}

\newcommand{\rulenameprocredExp}{\rulenameprocred{Exp}}
\newcommand{\rulenameprocredFork}{\rulenameprocred{Fork}}
\newcommand{\rulenameprocredNew}{\rulenameprocred{New}}
\newcommand{\rulenameprocredMsg}{\rulenameprocred{Com}}
\newcommand{\rulenameprocredCh}{\rulenameprocred{Ch}}
\newcommand{\rulenameprocredPar}{\rulenameprocred{Par}}
\newcommand{\rulenameprocredBind}{\rulenameprocred{Bind}}
\newcommand{\rulenameprocredCong}{\rulenameprocred{Cong}}

\newcommand{\rulenamealgkindType}{\rulenamealgkind{Type}}
\newcommand{\rulenamealgkindTypeAgainst}{\rulenamealgkind{Type-Against}}
\newcommand{\rulenamealgkindUnit}{\rulenamealgkind{Unit}}
\newcommand{\rulenamealgkindSkip}{\rulenamealgkind{Skip}}
\newcommand{\rulenamealgkindVar}{\rulenamealgkind{Var}}
\newcommand{\rulenamealgkindMsg}{\rulenamealgkind{Msg}}
\newcommand{\rulenamealgkindCh}{\rulenamealgkind{Ch}}
\newcommand{\rulenamealgkindSeq}{\rulenamealgkind{Seq}}
\newcommand{\rulenamealgkindArrow}{\rulenamealgkind{Arrow}}
\newcommand{\rulenamealgkindRcd}{\rulenamealgkind{Record}}
\newcommand{\rulenamealgkindData}{\rulenamealgkind{Variant}}
\newcommand{\rulenamealgkindRec}{\rulenamealgkind{Rec}}
\newcommand{\rulenamealgkindPoly}{\rulenamealgkind{TAbs}}
\newcommand{\rulenamealgcheckAgainst}{\rulenamealgkind{Against}}

\newcommand{\rulenamealgtypeConst}{\rulenamealgtype{Const}}
\newcommand{\rulenamealgtypeVarLin}{\rulenamealgtype{LinVar}}
\newcommand{\rulenamealgtypeVarUn}{\rulenamealgtype{UnVar}}
\newcommand{\rulenamealgtypeAbsLin}{\rulenamealgtype{LinAbs}}
\newcommand{\rulenamealgtypeAbsUn}{\rulenamealgtype{UnAbs}}
\newcommand{\rulenamealgtypeApp}{\rulenamealgtype{App}}
\newcommand{\rulenamealgtypeTAbs}{\rulenamealgtype{TAbs}}
\newcommand{\rulenamealgtypeTApp}{\rulenamealgtype{TApp}}
\newcommand{\rulenamealgtypeRcd}{\rulenamealgtype{Record}}
\newcommand{\rulenamealgtypeLet}{\rulenamealgtype{Proj}}
\newcommand{\rulenamealgtypeUnitElim}{\rulenamealgtype{UnitElim}}
\newcommand{\rulenamealgtypeVariant}{\rulenamealgtype{Variant}}
\newcommand{\rulenamealgtypeCase}{\rulenamealgtype{Case}}
\newcommand{\rulenamealgtypeNew}{\rulenamealgtype{New}}
\newcommand{\rulenamealgtypeSel}{\rulenamealgtype{Sel}}
\newcommand{\rulenamealgtypeMatch}{\rulenamealgtype{Match}}

\newcommand{\rulenameAlgtypeEq}{\rulenamealgtype{Eq}}

\newcommand{\grmeq}{\; ::= \;\;}
\newcommand{\grmor}{\;\mid\;}

\newcommand{\declrel}[2]{\emph{#1}\hfill\fbox{{#2}}}
\newcommand{\decltworel}[3]{\emph{#1}\hfill\fbox{{#2}}\quad{\fbox{{#3}}}}
\newcommand{\declthreerel}[4]{\emph{#1}\hfill\fbox{#2}\quad\fbox{#3}\quad\fbox{#4}}

\newcommand{\subs}[3]{{#3}[{#1}/{#2}]}
\newcommand{\dual}[1]{\overline{#1}}
\newcommand{\join}{\sqcup}

\newcommand\unravel[1]{\operatorname{unr}(#1)} % shortened from "unravel"
\newcommand\unravelSubs[2][\mathsf{id}]{\operatorname{unr}(#2)} % shortened from "unravel"
\newcommand{\synthetise}{\Rightarrow}
\newcommand{\checkagainst}{\Leftarrow}
\newcommand{\reduces}{\rightarrow}
\newcommand\TypeEquiv{\simeq}
\newcommand\STypeEquiv{\TypeEquiv_s} % was: {\TypeEquiv_\kinds}
\newcommand{\BPAEquiv}[1][\empty]{\approx_{#1}}
\newcommand{\subt}{<:}
\newcommand{\subk}{\subt}
\newcommand{\LTSderives}[1][\lambda]{\stackrel{#1}{\longrightarrow}}
\newcommand{\typeof}{\operatorname{typeof}}

\newcommand{\subterms}{\operatorname{sub}}

\newcommand{\tvarset}{\mathcal A}
\newcommand{\geqs}{\Sigma} % Guarded equations

\newcommand{\Out}{\text{out}}
\newcommand{\In}{\text{in}}

\newcommand{\judgementlabel}[1]{\mathsf{#1}} % labels in judgements
\newcommand{\twocontexts}[2]{{#1} \mid {#2}}
\newcommand{\judgement}[2]{{#1} \: \judgementlabel{#2}}

\newcommand{\judgementctx}[3]{{#1} \vdash \judgement{#2}{#3}}

\newcommand{\judgementrel}[3]{{#1} \; {#2} \; {#3}}

\newcommand{\judgementrelctx}[4]{{#1} \vdash \judgementrel{#2}{#3}{#4}}

\newcommand{\judgementrelbictx}[5]{\judgementrelctx{\twocontexts{#1}{#2}}{#3}{#4}{#5}}

\newcommand{\isContr}[3][\tvarset]{\judgementrelctx{#1}{#2}{.}{#3}}
\newcommand{\isNotContr}[3][\tvarset]{{#1} \nvdash \judgementrel{#2}{.}{#3}}
\newcommand{\isDone}[2][\tvarset]{\judgementctx{\empty}{#2}{\checkmark}}
\newcommand{\isNotDone}[2][\tvarset]{{\empty} \nvdash \judgement{#2}{\checkmark}}
\newcommand{\istype}[3]{\judgementrelctx{#1}{#2}{:}{#3}} % Outer
\newcommand{\isType}[4]{\judgementrelbictx{#1}{#2}{#3}{:}{#4}} % Inner
\newcommand{\isExpr}[4][\Delta]{\judgementrelbictx{#1}{#2}{#3}{:}{#4}}

\newcommand{\algKindOut}[3][\Delta]{\judgementrelctx{#1}{#2}{\synthetise}{#3}} %Outer
\newcommand{\algkindout}[4][\tvarset]{\judgementrelbictx{#1}{#2}{#3}{\synthetise}{#4}} %Inner
\newcommand{\ctxdiff}[4][\Delta]{\judgementrelctx{#1}{#2\div#3}{\synthetise}{#4}}
\newcommand{\algKindIn}[3][\Delta]{\judgementrelctx{#1}{#2}{\checkagainst}{#3}}
\newcommand{\algkindin}[4][\tvarset]{\judgementrelbictx{#1}{#2}{#3}{\checkagainst}{#4}}
\newcommand{\isSubkind}[2]{\judgementrel{#1}{\subk}{#2}}

\newcommand{\algtypeout}[5][\Delta]{\judgementrelbictx{#1}{#2}{#3}{\synthetise}{{#4}\mid{#5}}}
\newcommand{\algtypeoutnorm}[5][\Delta]{\judgementrelbictx{#1}{#2}{#3}{\synthetise\Downarrow}{{#4}\mid{#5}}}
\newcommand{\algtypein}[5][\Delta]{\judgementrelbictx{#1}{#2}{#3:{#4}}{\synthetise}{#5}}
\newcommand{\algequivin}[3][\Delta]{\isequiv[\Empty]{#1}{#2}{#3}{\kindtl}}

\newcommand{\normalisation}[2]{\judgementrel{#1}{\Downarrow}{#2}}
\newcommand\isequiv[5][\Theta]{\judgementrelbictx{#1}{#2}{#3}{\TypeEquiv}{#4}:{#5}}
\newcommand\isequivst[2]{\judgementrel{#1}{\STypeEquiv}{#2}}
\newcommand\isequivbpa[3][\empty]{\judgementrel{#2}{\BPAEquiv[#1]}{#3}}

\newcommand{\ctxequiv}[3][\Delta]{\judgementrelctx{#1}{#2}{\TypeEquiv}{#3}}
\newcommand{\append}[2]{\judgementrel{#1}{\mathrel{+{+}}}{#2}}
\newcommand{\isRed}[2]{\judgementrel{#1}{\reduces}{#2}}

\newcommand{\splitctx}[4][\Delta]{\judgementrelctx{#1}{#2}{=}{#3 \circ #4}}
\newcommand{\isproc}[2]{{#1}\vdash{#2}} % poor spacing
\newcommand{\isCtx}[3][\Delta]{\istype{#1}{#2}{#3}}
\newcommand{\isCUn}[2][\Delta]{\isCtx[#1]{#2}{\kindtu}}
\newtheorem{theorem}{Theorem}
\newtheorem{lemma}[theorem]{Lemma}
\newtheorem{corollary}[theorem]{Corollary}

\newcommand{\transeqs}[1]{\llbracket {#1} \rrbracket}
\newcommand{\transprocs}[1]{\transeqs{#1}}
\newcommand{\metafun}{\operatorname}
\newcommand\dom{\metafun{dom}}
\newcommand\free{\metafun{free}}
\newcommand\subj{\metafun{subj}}
\newcommand\agree{\metafun{agree}}

\newcommand{\ie}{i.e.,\xspace} % comma
\newcommand{\eg}{e.g.\xspace}  % no comma, according to the Handbook for Scholars
\newcommand{\etal}{et al.\xspace} % al abbreviates aliae, no need for am \emph,

\newcommand\Small{\small}

\definecolor{darkviolet}{rgb}{0.5,0,0.4}
\definecolor{darkgreen}{rgb}{0,0.4,0.2}
\definecolor{darkblue}{rgb}{0.1,0.1,0.9}
\definecolor{darkgrey}{rgb}{0.5,0.5,0.5}
\definecolor{lightblue}{rgb}{0.4,0.4,1}

\lstdefinestyle{eclipse}{
  breaklines=true,
  basicstyle=\sffamily\Small,
  emphstyle=\color{red}\bfseries,
  keywordstyle=\color{darkviolet}\bfseries,
  commentstyle=\color{darkgreen},
  stringstyle=\color{darkblue},
  numberstyle=\color{darkgrey},%\lstfontfamily,
  emphstyle=\color{red},
  showstringspaces=false,
}

\begin{document}

\begin{frontmatter}
  \title{Polymorphic Lambda Calculus with Context-Free Session Types}

  \author[1]{Bernardo Almeida\corref{cor1}}
  \ead{bpdalmeida@fc.ul.pt}
  \author[1]{Andreia Mordido}
  \ead{afmordido@fc.ul.pt}
  \author[2]{Peter Thiemann}
  \ead{thiemann@informatik.uni-freiburg.de}
  \author[1]{Vasco T. Vasconcelos}
  \ead{vmvasconcelos@fc.ul.pt}

  \cortext[cor1]{Corresponding author}
  \address[1]{LASIGE, Faculdade de Ciências, Universidade de Lisboa, Portugal}
  % Faculty of Sciences, University of Lisbon, Portugal}
  \address[2]{Faculty of Engineering, University of Freiburg, Germany}

\begin{abstract}
  Session types provide a typing discipline for structured communication on
  bidirectional channels. Context-free session types overcome the restriction to
  tail recursive protocols characteristic of conventional session types.
  This extension enables the serialization and deserialization of tree
  structures in a fully type-safe manner.

  We present the theory underlying the language \freest2, which features
  context-free session types in an extension of System~F with linear types and a
  kinding system to distinguish message types, session types, and channel types.
  The system presents metatheoretical challenges which we address:
  contractivity in the presence of polymorphism, a non-trivial equational theory
  on types, and decidability of type equivalence. We also establish standard
  results on typing preservation, progress, and a characterization of erroneous
  processes.
\end{abstract}

\end{frontmatter}

\section{Introduction}
\label{sec:introduction}

Session types were discovered by Honda as a means to describe structured process
interaction on typed communication
channels~\cite{DBLP:conf/concur/Honda93,DBLP:conf/esop/HondaVK98,DBLP:conf/parle/TakeuchiHK94}.
\changed{Session types allow expressing elaborate protocols on communication
  channels, sharply contrasting with languages such as Concurrent
  ML~\cite{DBLP:conf/pldi/Reppy91} and
  Go~\cite{go,DBLP:journals/pacmpl/GriesemerHKLTTW20} %,heyholetsgo}
  that force channels to carry objects of a common type during the whole
  lifetime of the channel.}

% While connections are homogeneously typed in languages like Concurrent
% ML~\cite{DBLP:conf/pldi/Reppy91} and
% Go~\cite{go,DBLP:journals/pacmpl/GriesemerHKLTTW20}, 
% session types
% facilitate heterogeneously typed exchanges on communication channels.

Session types provide detailed protocol specifications by describing message
exchanges and choice points: if $T$ is a type and $U$ and $V$ session types,
then $!T.U$ and $?T.U$ are session types that describe channels where a message
of type $T$ is sent or received, and where $U$ describes the ensuing protocol.
Choices are usually present in labelled form, so that, in the particular case of
a binary choice, $\oplus\{l\colon U, m\colon V\}$ and
$\&\{l\colon U, m\colon V\}$ describe channels selecting or offering labels~$l$
and~$m$ and continuing as~$U$ or~$V$ according to the choice taken. Type $\Endk$
denotes a channel on which no further interaction is possible. It terminates all
session types. Starting from Honda \etal~\cite{DBLP:conf/esop/HondaVK98}, most
works on session types incorporate recursive types for unbounded behaviour: type
$\mu a.T$ introduces a recursive type with a bound type variable $a$ that can
be used to refer to the whole $\mu$-type.
For example, type%
\begin{equation*}
  \mu a.\oplus\!\{\pushl\colon !\intk.a, \popl\colon ?\intk.a, \donel\colon\Endk\}
\end{equation*}
describes a channel providing for three operations $\pushl$, $\popl$, and
$\donel$. Clients that exercise option $\pushl$ must then send an integer
value, after which the protocol goes back to the beginning. Similarly, the
processes that select option $\popl$ must then be ready to receive an integer
value before going back to the beginning. Finally, by exercising option $\donel$ the
client signals protocol completion.

At the other end of the channel a server awaits. Its type is \emph{dual} to that
of clients, namely
$\mu b.\&\!\{\pushl\colon ?\intk.b, \popl\colon !\intk.b, \donel\colon\Endk\}$.
The server offers the three options and behaves according to the option selected
by the client. The duality $\oplus/\&$ and $!/?$ guarantees that communication
does not go wrong: if one partner selects, the other offers, if one side writes,
the other reads, and conversely.

Session types are well suited to document communication protocols and there is a
whole range of variants and extensions that make them amenable to describe
realistic situations, including those featuring multiple
partners~\cite{DBLP:journals/jacm/HondaYC16}, accounting for object-oriented
programming~\cite{DBLP:journals/corr/abs-1205-5344}, web
programming~\cite{lindley17:_light_funct_session_types} or exceptional
behaviour~\cite{DBLP:journals/pacmpl/FowlerLMD19}.
Despite these advances, session types remain limited in their support for
compositionality.
Protocols $U$ and $V$ can only be combined under (internal or external) choice,
where either $U$ or $V$ are used, but not both. Input and output do not qualify
as protocol composition operators for they merely append a simple communication
(input or output) at the head of a type. A sequential composition operator on
session types, as in $U;V$, greatly increases the flexibility in protocol
composition, opening perspectives for describing far richer protocols.

When implementing the client side of the stack protocol just introduced, the
programmer is faced with the classical difficulty of what action to take when
confronted with a $\popl$ operation on an empty stack.
Session types equipped with sequential composition can accurately describe a
stack protocol where the $\popl$ operation is not available on an empty
stack. Rather than writing a single, long doubly recursive type, we factor the
type in two which we introduce by means of two mutually recursive equations. 
Inspired in Padovani~\cite{DBLP:journals/toplas/Padovani19}, the protocol can be
written in \freest~\cite{freest} as follows:
\begin{lstlisting}
type EStack = +{Push: !Int ; Stack ; EStack, Done: Skip}
type Stack  = +{Push: !Int ; Stack ; Stack,  Pop:  ?Int}
\end{lstlisting}
% \begin{align*}
%   \estackn & = \oplus\{\pushl\colon !\intk;\stackn;\estackn, \donel\colon\skipk\}
%   \\
%   \stackn & = \oplus\{\pushl\colon !\intk;\stackn;\stackn, \popl\colon ?\intk\}
% \end{align*}
%
$\estackn$ describes an empty stack on which $\pushl$ and $\donel$ operations
are available. A $\pushl$ must be followed by a non-empty stack followed by an
empty stack again. According to the type, stack sessions can only terminate when
the stack is empty. A non-empty stack offers $\pushl$ and $\popl$ operations.
The former must be followed by two non-empty stacks, the first accounts for the
element just pushed, the second for the state the stack was at prior to
$\pushl$. The $\popl$ operation returns the value at the top of the stack.

The syntactic changes with respect to conventional session types are mild: input and
output become types of their own, $!T$ and $?T$, irrespective of the
continuation, if any. Prefixes give way to sequential composition. In the
process, we rename type $\Endk$ to $\skipk$ for the new type behaves differently
from its conventional counterpart. Choices remain as they were. Then, what once
was $!\intk.\Endk$ is now simply $!\intk$ and what one once wrote as
$!\intk.?\boolk.\Endk$ can now be written as $!\intk;?\boolk$.

Session types arose in the scope of process calculi, the $\pi$-calculus in
particular~\cite{DBLP:journals/iandc/MilnerPW92a}. They however fit nicely with
strongly typed functional languages, as shown by Gay
\etal~\cite{DBLP:journals/jfp/GayV10,DBLP:journals/eatcs/Vasconcelos11,DBLP:conf/concur/VasconcelosRG04},
and this is the trend we follow in this paper.
Suppose that we are to program a client that must interact with a stack via
a channel end \lstinline|c| of type $\estackn$. The code can be written in
\freest as follows,
\begin{lstlisting}[numbers=right]
  let c = select Push c in let     c  = send  5 c in
  let c = select Pop  c in let (x, c) = receive c in
  let c = select Pop  c in let (y, c) = receive c in
\end{lstlisting}
where \lstinline|select| chooses a branch on a given channel and returns a
channel where to continue the interaction, \lstinline|send| sends a value on a
given channel and again returns the channel where to continue the interaction,
and \lstinline|receive| returns a pair composed of the value read from the
channel and the continuation channel. We thus see that all session operations
return a continuation channel; we identify them all by the same identifier
\lstinline|c| for simplicity.
The code does not contain any syntax for unfolding the recursive
session type. It is customary for session type systems to use
equi-recursion which equates a recursive type with its unfolding.
Running the above code through our interpreter~\cite{freest} on a context that
assigns type \lstinline|EStack| to identifier \lstinline|c|, we get
\begin{verbatim}
stack.fst:3:11: error:
    Branch Pop not present in external choice type EStack
\end{verbatim}

Type \lstinline|EStack| (or \lstinline|Stack| for that matter) cannot be written
with conventional session types. Almeida
\etal~\cite{DBLP:conf/tacas/AlmeidaMV20} translate context-free session types to
grammars for the purpose of deciding type equivalence.
% Considering $\oplus\pushl$, $\oplus\donel$, $\oplus\popl$, $!\intk$ and $?\intk$
% as terminal symbols, the grammar associated to the stack types is clearly not
% regular.
% %
% \begin{align*}
%   \estackn & \rightarrow \oplus\pushl \;\; !\intk \;\; \stackn \;\; \estackn
%   &
%   \estackn & \rightarrow \oplus\donel
%   \\
%   \stackn & \rightarrow \oplus\pushl \;\; !\intk \;\; \stackn \;\; \stackn
%   &
%   \stackn & \rightarrow \oplus\popl \;\; ?\intk
% \end{align*}
%
In contrast, the language described by traditional session types is
regular. More precisely, taking infinite executions into account, each
traditional session type can be related to the union of a regular language and a
$\omega$-regular language that describe the finite and infinite sequences of
communication actions admitted by the type.

% Getting back to the stack example
We now look at the code for the server side. We need two mutually recursive
functions, one to handle the empty stack, the other for the non empty case. The
code closely follows the types:
\begin{lstlisting}[numbers=right]
eStack c = match c with
  { Push c -> let (x, c) = receive c in eStack (stack x c)
  , Done  c -> c
  }
stack x c = match c with
  { Push c -> let (y, c) = receive c in stack x (stack y c)
  , Pop  c -> send x c
  }
\end{lstlisting}
The \lstinline|match| expression branches according to the label selected by the
client. Next we focus on the types for the two functions. They both expect and
return channels. Function \lstinline|eStack| consumes an \lstinline|EStack|
channel to completion, so that its type must be \lstinline|EStack -> Skip| when
called for the first time. Similarly, the type of function \lstinline|stack| is
\lstinline|Int -> Stack -> Skip| when called for the first time. But then the
type of the call \lstinline|(stack x c)| in line 2 is \lstinline|Skip| and thus
its value cannot be used as a parameter to function \lstinline|eStack| that
expects an \lstinline|EStack|.
Consequently programming in \freest requires support for polymorphism, and in
particular for \emph{polymorphic recursion}. Both functions must be polymorphic in
the continuation, so that they may accommodate the top-level and the recursive
calls equally. The types for the functions are as follows.
\begin{lstlisting}
eStack : foralla:SL . EStack;a -> a
 stack : foralla:SL . Int -> Stack;a -> a
\end{lstlisting}

The polymorphic types above abstract over a \emph{linear session} type. \freest
distinguishes three base kinds: functional (\lstinline|T|), session
(\lstinline|S|) and message (\lstinline|M|) in order to control type formation
in the presence of polymorphism. For example, only \lstinline|S| types can be
used in sequential composition, and only \lstinline|M| types can be used in
message exchanges (\lstinline|send| or \lstinline|receive|). To obtain a kind,
we complement base kinds with their multiplicity: \lstinline|L| for linear and
\lstinline|U| for unrestricted. For example, \lstinline|Skip : SU| and
\lstinline|U;V : SL| if both \lstinline|U| and \lstinline|V| are session types.
Kinding comes equipped with a partial order that allows both \lstinline|M| and
\lstinline|S| types to be viewed as \lstinline|T|, and \lstinline|U| as
\lstinline|L|, so that \lstinline|Skip : TL| if needed. We often omit the top
kind, that is \lstinline|TL|, in examples.

% \freest features equi-recursive types, with allows us to type $Z$, a call-by-value
% fixed point combinator.
% %
% \begin{lstlisting}
% Lambdaa . lambdaf:(a->a->(a->a)) .
%   (lambdax:(mub.b->a->a) . f(lambdaz:a . x x z))
%   (lambdax:(mub.b->a->a) . f(lambdaz:a . x x z))
% \end{lstlisting}
% fixZ : foralla . ((a->a)->(a->a))->(a->a)
% fixZ = Lambdaa . lambdaf . (lambdax:(mub.b->a->a) . f(lambdaz:a . x x z))
%                  (lambdax:(mub.b->a->a) . f(lambdaz:a . x x z))

% The examples \cref{sec:motivation} highlight the need for the following language
% construction:
% %
% \begin{itemize}
% \item Sequential composition for session types, $T;U$,
% \item Linear types,
% \item Universally quantified type variables, $a$ 
% \todo{does this add anything to the next one? -am}
% \item Type-level abstraction $\forallt a \kind T$ (but not type level
%   application),
% \item Recursive types of fixed type functions, $\rect a \kind T$.
% \end{itemize}

\paragraph{Contributions}

\begin{itemize}
\item We introduce context-free session types that extend the expressiveness of
  regular session types, allowing type-safe implementation of algorithms beyond
  the reach of regular sessions.
\item We propose a first-order kinding system that distinguishes messages from
  sessions from functional types as well as linear from unrestricted types,
  and a notion of contractiveness for recursive types that takes into account
  types equivalent to \lstinline|Skip| and ensures that substitution is a total
  function.
\item We formalise the core of \freest as a functional language \fmusession, a
  polymorphic linear language with equi-recursive types, records and variants,
  term-level type abstraction and application, channel and thread creation, and
  session communication primitives governed by context-free session types; we
  show soundness, absence of runtime errors of the language as well as progress
  for single-threaded programs.
\item We show that type equivalence is decidable in the presence of
  context-free sessions; we present a bidirectional type checking algorithm and prove it correct
  with respect to the declarative system.
\end{itemize}

All these ideas are embodied in a freely available interpreter~\cite{freest}.
This paper polishes and expands Thiemann and
Vasconcelos~\cite{DBLP:conf/icfp/ThiemannV16}. The two main novelties are
impredicative (System F) polymorphism instead of predicative
(Damas-Milner~\cite{DBLP:conf/popl/DamasM82}) and algorithmic type checking.
% and an interpreter for \freest 2.0~\cite{freest}.
%
The original paper includes an embedding of a functional language for regular
session types~\cite{DBLP:journals/jfp/GayV10}; the embedding is still valid in
the language of this paper even if we decided not to discuss it.

\paragraph{Outline}

The rest of the paper is organised as follows. \Cref{sec:motivation} discusses
the overall design of \freest and explains the requirements to the metatheory
with examples. \Cref{sec:types} introduces the type language and the notions of
type equivalence and duality. \Cref{sec:processes} introduces the 
process language and proves type preservation, absence of runtime errors and
progress. \Cref{sec:type-equiv-decidable} shows that type equivalence is
decidable. \Cref{sec:algorithmic-typing} presents algorithmic type checking
and proves its correctness with respect to the declarative system.
\Cref{sec:related} discusses related work and \cref{sec:conclusion} concludes
the paper and points to future work.

%%% Local Variables:
%%% mode: latex
%%% TeX-master: "main"
%%% End:

\section{\fmusession\ in Action}
\label{sec:motivation}

Our language for polymorphic context-free session types is called \fmusession.
It is based on $F^\mu$, System F with equi-recursive types. On top of this
system we add multi-threading and communication based on context-free session
types. Extending polymorphism to session types requires an appropriate kind
structure. We freely use {\freest} syntax in our examples, deferring the
introduction of the formal syntax of {\fmusession} to \cref{sec:types}.

\subsection{Polymorphic Session Types}
\label{sec:polym-sess-types}

Many session type systems model their communication
operations with type-indexed families of constants. For instance,
Gay and Vasconcelos~\cite[Fig.~9]{DBLP:journals/jfp/GayV10} specify
informal \emph{typing schemas for constants}, \eg
${\sendk} \colon T
\to !T.S \to S$ where $S$ may be instantiated with any session type
and $T$ with any arbitrary type. This design enables the use of
the $\sendk$ operation at different types, but it is somewhat
restrictive.  Formally, we would have to make the indices explicit as
in $\sendk_{S,T}$, but more importantly this approach does not allow
us to abstracting over operations that send (or receive) data. 
For example, suppose we want to write a function that sends all data encrypted:
\begin{lstlisting}
sendEncrypted : foralla b . (a -> Int) -> a -> (!Int;b) -> b
sendEncrypted encrypt x c = send (encrypt x) c
\end{lstlisting}
It takes an encryption function that encodes a value of type
\lstinline{a} into an integer, a value of type \lstinline{a}, and a
channel on which we can send an integer. It returns the continuation
of the channel.
This function cannot be written with a polymorphic type in most
previous session-type systems
(exceptions include lightweight functional session 
types~\cite{lindley17:_light_funct_session_types}
and nested session types~\cite{DBLP:conf/esop/DasDMP21,das2021subtyping}).
In {\fmusession}, we define the sending and receiving operations as
polymorphic constants. In a first approximation, we might consider the
following (incomplete) types:
\begin{align*}
  \sendk & : \forall a.\,\unfunt a {\forall b.\,\unfunt{\semit{\sout a} b}{b}}
  \\
  \receivek & : \forall a.\forall b. \, \unfunt{\semit{\sint a} b}{\oldpairt a b}
\end{align*}

The $\sendk$ operation takes a value of type $a$, a channel of type
$\semit{\sout a} b$ (send $a$ and then continue as $b$), and returns a channel
of type $b$. Similarly, $\receivek$ takes a channel of type $\semit{\sint a}b$
(receive $a$ and continue as $b$) and returns a pair of the received value of
type $a$ and the remaining channel of type $b$. This typing suffers from
several shortcomings. First, the type $b$ must be a session type in both cases.
Instantiating $b$ with a record, a variant or a function type would result
in an ill formed type. Second, to keep typing decidable, our system restricts
the type $a$ to any type that does \emph{not} contain a session. Third, channel
references must be handled in a linear fashion, which is also not reflected in
the proposed signature, yet.

Our solution is to classify types with suitable kinds.
We identify three base kinds, together with their linearity variants:
$  \kindmu$ and $\kindml$ stand for types that can be exchanged on channels;
$\kindsu$ and $\kindsl$ refer to session types;
$  \kindtu$ and $\kindtl$ stand for arbitrary types.

Adding kinds to the types of  $\sendk$ and $\receivek$ we get:
% \vv{Check these
%   types against those in \cref{fig:types-constants}. Also, the
%   $\forall b$ appears later in the type for $\sendk$ and we can take
%   advantage of this.}\pt{Taking the type in \cref{fig:types-constants}
%   requires to change the definition of values. Moreover, the types are
%   isomorphic so I'd favor the one with prenex quantification.}
%
\begin{align*}
  \sendk & : \forallt a \kindml {\funt[\unk] a {\forallt b\kindsl {{\funt[\link]{\semit{\sout a}b}{b}}}}}
  \\
  \receivek & : \forallt a \kindml {\forallt b \kindsl {\funt[\unk]{\semit{\sint a} b}{\oldpairt a b}}}
\end{align*}

What types are classified as $\kindm$ is limited by type equivalence alone. For
regular session types, $\kindm$ coincides with $\kindt$. For context free
session types we restrict $\kindm$ to base types $\unitt[\un]$ and
$\unitt[\lin]$ %$\changed{\intk, \boolk, \unitk[]}, \dots$,
but this could be easily extended to other base types, records and variants, for
example.

We annotate the last function arrow in the type of $\sendk$ with $\link$ to
cater for the possibility that the first argument $x:a$ is in fact
linear. In that case, the closure obtained by partial application
$\sendk\ x$ has type ${\funt[\link]{\semit{\sout a}b}{b}}$ to indicate that it
must be used exactly once (i.e., linearly). The remaining function
arrows are unrestricted as indicated by the subscript $\to_{\unk}$.

There is no equally pleasing way to enable abstraction over the
operations to offer and accept branching in a session type. Hence,
these operations are still hardwired as expressions with parametric
typing rules.

We conclude with the implementation of encrypted sending and receiving
in \freest. Syntactically, \lstinline{forall} introduces a universal
type, \lstinline{ML} stands for $\kindml$, \lstinline{SL} for
$\kindsl$, and we use the Haskell notation for pair types.
\begin{lstlisting}
sendEncrypted :
    foralla:ML . forallb:SL . (a -> Int) -> a -> (!Int; b) -> b
sendEncrypted encrypt x c = send (encrypt x) c

recvEncrypted :
    foralla:ML . forallb:SL . (Int -> a) -> (?Int; b) -> (a, b)
recvEncrypted decrypt c =
    let (i, c) = receive c in
    (decrypt i, c)
\end{lstlisting}
As an example of the use of these abstractions, we implement the addition server
with encryption. It takes an encoding function, a decoding function, and an
encoded channel to produce \lstinline{Skip}, i.e., a depleted channel. As
customary in System F, we must provide explicit type arguments (in square
brackets) to the polymorphic operations for encrypted sending and receiving. The
typechecker has special support to infer type arguments for $\sendk$ and
$\receivek$. We omit the corresponding client implementation, which is
straightforward.
\begin{lstlisting}
server : (Int->Int) -> (Int->Int) -> (?Int;?Int;!Int) -> Skip
server enc dec c =
    let (x, c) = recvEncrypted [Int, ?Int;!Int] dec c in
    let (y, c) = recvEncrypted [Int, !Int] dec c in
    sendEncrypted [Int, Skip] enc (x + y) c
\end{lstlisting}

% This example is about (functional) parametric polymorphism, not about session types
%
% Explicit polymorphism enables some flexibility which is
% hard to achieve in an ML-like language. For example, we can exploit
% the type isomorphism
% \begin{align*}
%   \forallt a \kindml {\forallt b\kindsl{\funt[\unk] a
%       {{\funt[\link]{\semit{\sout a}b}{b}}}}}
%   & \cong
%   \forallt a \kindml {{\funt[\unk] a {\forallt b\kindsl{\funt[\link]{\semit{\sout a}b}{b}}}}}
% \end{align*}
% embodied in function \lstinline{send1} to ``partially apply'' the type
% parameters as in the following code fragment, where \lstinline|e1;e2|
% abbreviates \lstinline|let _ = e1 in e2|:
% \begin{lstlisting}
% consumeC : !Int -> Skip -- omitted
% consumeD : !Int;?Bool -> Skip -- omitted

% send1 : foralla:ML . a -> forallb:SL . !a;b -> b
% send1 = Lambdaa:ML => lambdax:a -> Lambdab:SL => lambdac:!a;b -> send x c

% f : Bool -> !Int -> !Int;?Bool -> Skip
% f cond c d =
%     let x = send1 [Int] 5 in  -- x : forallb:SL . !Int;b -> b
%     if cond
%     then x [Skip] c            ; consumeD d
%     else receive (x [?Bool] d) ; consumeC c
% \end{lstlisting}
% Function \lstinline{f} definitely sends integer 5, but makes an
% internal choice on \lstinline{cond} in choosing a
% continuation. In an ML-like language, the type of \lstinline{x} would
% not be polymorphic because the expression \lstinline{send1 5} is not a
% syntactic value in ML. Hence, the two different uses of \lstinline{x}
% would give rise to a type clash of the two continuation types.

\subsection{Parameterization over Sub-Protocols}
\label{sec:param-over-sub}

As a more advanced example of polymorphism we consider
parameterization over fragments of a protocol. This facility is not
available in the standard, tail-recursive formulations of session
types. It relies crucially on the availability of the composition
operator ``;'' that forms the sequential composition of two
protocols. Using composition we can write parametric wrapper functions
to implement
protocols such as \lstinline_!Auth; &{Ok: s; ?Acct, Denied: Skip}_, where
the parameter \lstinline{s} is the session type of the actual payload
protocol of a client.
For example, the client wrapper might send an authentication token of type
\lstinline{Auth}. If the authentication is accepted, then the actual
payload protocol \lstinline{s} proceeds and accounting information of
type \lstinline{Acct} is attached. Otherwise, execution of the payload protocol
is skipped.

An implementation of the client wrapper would have to parameterise
over the client function implementing the payload protocol. The proper
type for this wrapper involves higher-rank polymorphism to
parameterise over the type of the continuation. In this example, we
hardcode the password for authentication as a constant in the code.
\begin{lstlisting}
password : Auth
password = ...

clientWrapper :
    foralla b:MU .
    foralls cont:SL .
    (foralld:SL . a -> s;d -> (b, d)) ->
    a -> b ->
    !Auth; &{ Ok: s;?Acct, Denied: Skip}; cont ->
    (b, cont)
clientWrapper proto init def c =
    let c = send password c in
    match c with {
        Ok c -> let (ret, c) = proto [?Int; cont] init c in
                let (acc, c) = receive c in
                (ret, c),
        Denied c -> (def, c)
    }
\end{lstlisting}
The first parameter is a function that processes the payload
protocol. It is polymorphic in its continuation parameter \lstinline{d} so that it
could be invoked in different contexts.  It takes an initial parameter
of type \lstinline{a}, transforms a channel of type \lstinline{s;d},
and returns a result of type \lstinline{b} paired with the
continuation \lstinline{d}.
Then there are parameters for the initial parameter of the payload and
to provide a default value of type \lstinline{b}. Finally, the channel
with the augmented protocol is transformed into a pair of the result
of the payload and the continuation channel.

The client wrapper can be used with multiple instantiations for the
protocol. As an example, we consider a recursive protocol to send a
list of integer items.
\begin{lstlisting}
type ListC:SL = +{Item: !Int; ListC, Stop: Skip}
\end{lstlisting}
We write a client for this protocol where the initial value and the output are
integers. The infix operator
\lstinline|& : foralla:TL . forallb:TL . a -> (a -> b) -> b| represents reverse
function application; we use it to hide the continuation channel in the code.
\begin{lstlisting}
clientPayload : foralld:SL . Int -> ListC;d -> (Int, d)
clientPayload init c =
    if init == 0 then
        (0, select Stop c)
    else
        select Item c &
        send init &
        clientPayload [d] (init - 1)
\end{lstlisting}
Putting everything together, we choose arbitrary initialisation and
default parameters (111 and 999) to  obtain
\begin{lstlisting}
clientFull :
    !Int;&{Ok: ListC;?Int, Denied: Skip} -> (Int, Skip)
clientFull =
    clientWrapper [Int,Int,ListC,Skip] clientPayload 111 999
\end{lstlisting}
It is crucial that \lstinline|clientPayload| is polymorphic in the continuation.

\subsection{Tree-structured Transmission}
\label{sec:tree-struct-transm}

The next example fully exploits the freedom gained with context-free
session types. We develop code for the type-safe serialisation and
deserialisation of a binary tree. Such trees may be defined by the
following algebraic datatype definition.
\begin{lstlisting}
data Tree = Leaf | Node Tree Int Tree
\end{lstlisting}
To transmit such a tree faithfully requires a context-free structure
as embodied by the following protocol.
\begin{lstlisting}
type TreeChannel:SL = +{
  Leaf: Skip,
  Node: TreeChannel ; !Int ; TreeChannel
 }
\end{lstlisting}
A \lstinline{TreeChannel} either selects a \lstinline{Leaf} and then nothing, or
it selects a \lstinline{Node}, sends the left subtree, followed by the root and
then the right subtree. While the tail-recursive occurrence of
\lstinline{TreeChannel} (that for the right subtree) is acceptable in a regular
session type, the first occurrence of \lstinline{TreeChannel} (that for the left
subtree) is not. In \freest it is straightforward to implement a function to
write a \lstinline{Tree} on a \lstinline{TreeChannel}.
\begin{lstlisting}[numbers=right]
write : foralla:SL . Tree -> TreeChannel;a -> a
write t c =
  case t of {
    Leaf ->
      select Leaf c,
    Node t1 x t2 ->
      select Node c &
      write [TreeChannel;a] t1 &
      send x &
      write [a] t2
  }
\end{lstlisting}
Both recursive invocations of \lstinline{write} take the type of  their
continuation session as a parameter. The continuation of the first
\lstinline{write} consists of the \lstinline{TreeChannel} generated by
the second \lstinline{write} and the enclosing continuation
\lstinline{a}. The second \lstinline{write} occurs in tail position
and its continuation corresponds with the enclosing continuation \lstinline{a}.

It is similarly straightforward to read a \lstinline{Tree} from a
session of type \lstinline{dualof TreeChannel}.
% At present the \lstinline{dualof} operator only applies to monomorphic types.
Instead of reifying a
tree, we show a function that computes the sum of all values in a tree
directly from the serialised format on the channel.
\begin{lstlisting}
treeSum : foralla:SL . dualof TreeChannel; a -> (Int, a)
treeSum c =
  match c with {
    Leaf c ->
      (0, c),
    Node c ->
      let (sl, c) = treeSum [dualof TreeChannel;a] c in
      let (x, c) = receive c in
      let (sr, c) = treeSum [a] c in
      (sl + x + sr, c)
  }
\end{lstlisting}

The examples of \lstinline{write} and \lstinline{treeSum} demonstrate
the requirement for polymorphism when handling proper context-free
sessions. The recursive invocations of 
\lstinline{write} and \lstinline{treeSum} require different instantiations of the session
continuation \lstinline{TreeChannel;a}
(\lstinline{dualof TreeChannel;a}, respectively) and \lstinline{a}, which
are evidence for polymorphic recursion.

A closer look at the typing of \lstinline{write} reveals further requirements on
the type structure of {\fmusession}, which features non-trivial identities
beyond the ones imposed by equi-recursion (see \cref{lemma:laws}). The initial
type of \lstinline{c} in line~2 is \lstinline{TreeChannel;a}, which is a
recursive type followed by \lstinline{a}. Applying \lstinline{select Leaf} to
\lstinline{c} (in line~5) requires unrolling the recursive type in
\lstinline{TreeChannel; a} to
\lstinline|+{Leaf: Skip, Node: TreeChannel;!Int;TreeChannel}; a|, but the
top-level type is still a sequence. Hence, we must allow choice to distribute
over sequence as in
\lstinline|c : +{Leaf: Skip;a, Node: TreeChannel;!Int;TreeChannel;a}|. With this
argument type, the inferred type of \lstinline{select Leaf c} is
\lstinline{Skip;a}, which is still different from the expected type
\lstinline{a}. As \lstinline{Skip} has no effect, we consider it a left and
right unit of the sequencing operator so that \lstinline{Skip;a} is equivalent
to \lstinline{a}, which finally fits the expected type. Typing the
\lstinline{Node} branch of the code yields no further insights.
With the same approach one can easily write type-safe code for serialising and
deserialising JSON data. 

%%% Local Variables:
%%% mode: latex
%%% TeX-master: "main"
%%% End:

\section{Types}
\label{sec:types}

This section introduces the notion of types, type equivalence and session type
duality.

\subsection{Types and Type Formation}

%%% BASE SETS AND SYNTAX

\begin{figure}[t!]
  \begin{align*}
    \text{Multiplicity} &&
      m \grmeq & \un \grmor \lin
    \\
    \text{Basic kinds} &&
      \prekind \grmeq & \kindm \grmor \kinds \grmor \kindt
    \\
    \text{Kinds} &&
      \kind \grmeq& \prekind^m
    % \grmor \kabs\kind\kind  % TOK _ Type Operators and Kinding
    \\
    \text{Polarity} &&
      \msgt{} \grmeq& {}! \grmor {}? 
    \\
    \text{View} &&
      \star \grmeq& \oplus \grmor {}\& 
    \\
    \text{Types} &&
      T \grmeq& \skipt \grmor \msgt T \grmor \choicet{\recordt \ell T L} \grmor T;T
    \\
    &&& \unitt \grmor  \funt TT \grmor \recordt \ell T L \grmor \variantt \ell T L
    \\
    &&& \forallt a \kind T \grmor \rect a \kind T \grmor a
    \\
    % \text{Type variables} &&
    %   \tvarset \grmeq & \Empty \grmor \tvarset, a  
    % \\
    \text{Kinding contexts} &&
      \Delta \grmeq & \Empty \grmor \Delta, a\colon \kind    
  \end{align*}
  \caption{Syntax of kinds, type and kinding contexts}
  \label{fig:types}
\end{figure}

%%% Local Variables:
%%% mode: latex
%%% TeX-master: "main"
%%% End:

We rely on a couple of base sets: \emph{type variables}, denoted by $a,b,c$, and
\emph{labels}, denoted by $k,\ell$.
The syntax of kinds and types is in \cref{fig:types}.
We identify three basic kinds, together with their
multiplicity variants.
\begin{align*}
  \kindmu, \kindml &\quad \text{for types that can be exchanged on channels}
  \\
  \kindsu, \kindsl &\quad \text{for session types}
  \\
  \kindtu, \kindtl &\quad \text{for arbitrary types}
\end{align*}
Multiplicities describe the number of times a value can be used: exactly once
for $\lin$ and zero or more times for $\un$.

Session types include the terminated type $\skipt$, output messages $\sout T$
and input messages $\sint T$, internal labelled choices
$\ichoicet{\recordt \ell T L}$ and external choices
$\echoicet{\recordt \ell T L}$, and the sequential composition $T;U$.

Functional types include \changed{the unrestricted and linear units,
  $\unitk[\un]$ and $\unitk[\lin]$}, unrestricted and linear functions,
$\funt[\un]TU$ and $\funt[\lin]TU$, records $\recordt \ell T L$, and variants
$\variantt \ell T L$. The unit \changed{types are representative of all base
  types including boolean and integer used in examples. The linear multiplicity,
  $\lin$, assigned to basic types enables to specify linear assets to be
  exchanged on channels (e.g., tokens).} The multiplicity in the function type
constrains the number of times a function can be used: exactly once for
$\rightarrow_\lin$ case and zero or more times for $\rightarrow_\un$. A function
capturing in its body a free linear value must itself be linear. Functions
receiving linear values may still be reused (hence $\un$) if they contain no
free linear variables. Unrestricted functions are abbreviated in examples to
$\unfunt TU$.
% and linear functions to $\linfunt
% TU$. % and universal types $\forallt a \kind T$.
\changed{The pair type $T\times U$ used in examples abbreviates the record
  $\pairt TU$.} Recursive types $\rect a \kind T$ are both session and
functional. Type variable $a$ denotes a type variable in a recursive type or a
polymorphic variable in a universal type.

% \paragraph{The $\alpha$-identification Convention}

% %%% alpha-mu-tick CONVENTION

We say that a \emph{change of bound variables} in $T$ is the
replacement of a part $\rect{a}{\kind}{U}$ of $T$ by $\rect{b}{\kind}{U'}$
where $b$ does not occur (at all) in $U$ and where $U'$ is obtained from $U$ by
replacing all free occurrences of $a$ by $b$.
We say that two types are $\alpha$-equivalent if one can be obtained
from the other by a series of changes of bound variables.
This paper adopts the \emph{$\alpha$-identification convention} whereby
processes, terms and types are always identified up to $\alpha$-equivalence.

%%% THE TERMINATED PREDICATE

\paragraph{Terminated Types}

\begin{figure}[t!]
  \declrel{Type $T$ is terminated}{$\isDone T$}
  \begin{gather*}
    \axiom{\rulenametermSkip}{\isDone \skipt}
    \qquad
    \infrule{\rulenametermSeq}{\isDone{T} \\ \isDone{U}}{\isDone{\semit TU}}
    \qquad
    \infrule{\rulenametermRec}{\isDone[\tvarset,a]{T}}{\isDone{\rect{a}{\kinds^m}{T}}}
%    \qquad
%    \infrule{\rulenametermVar}{a\in\tvarset}{\isDone a}
  \end{gather*}
  \caption{The is-terminated predicate}
  \label{fig:terminated}
\end{figure}
%%% Local Variables:
%%% mode: latex
%%% TeX-master: "main"
%%% End:

Type termination in \cref{fig:terminated} applies to session types alone.
Intuitively a type is terminated if it does not exhibit a communication action.
Terminated types are composed of $\skipt$, sequential composition and recursion.
We say that type $T$ is \emph{terminated} when judgement $\isDone{T}$ holds.

%%% CONTRACTIVITY

\paragraph{Contractive Types}

\begin{figure}[t]
  \declrel{Type $T$ is contractive on type variable $a$}{$\isContr aT$}
  \begin{gather*}
    \infrule{\rulenamecontrSeqTwo}
    	{\isDone{T} \\ \isContr aU}
    	{\isContr a {(\semit TU)}}
    \quad
  	\infrule{\rulenamecontrSeqOne}
  		{\isNotDone{T} \\ \isContr aT}
  		{\isContr a {(\semit TU)}}
    \quad
    \infrule{\rulenamecontrPoly}
    {\isContr aT}
    {\isContr{a}{(\forallt b\kind T)}}
    \\
    \infrule{\rulenamecontrRec}
    	{\isContr aT}
    	{\isContr{a}{(\rect b\kind T)}}
    \infrule{\rulenamecontrVar}
%    	{b \neq a \\ b \not \in \tvarset }
    	{b \notin \tvarset,a }
    	{\isContr ab}
    \infrule{\rulenamecontrOther}
    % {T \neq \skipt}
    {T =
      \skipt,
      \msgt U,
      \choicet{\recorda \ell : U},
      \unitt,
      \funt UV,
      \recorda \ell : U,
      \variantta \ell U}
     {\isContr aT}
  \end{gather*}

  \caption{Contractivity}
  \label{fig:contractivity}
\end{figure}

%%% Local Variables: 
%%% mode: latex
%%% TeX-master: "main"
%%% End: 

Contractivity is adapted from the standard ``no subterms of the form
$\rect a \kind {\rect {a_1}{\kind_1}{\dots \rect{a_n}{\kind_n}a}}$'' to take into
account the nature of the semi-colon operator as the neutral for sequential
composition. % We would not like to deem as a type a piece of syntax such as
% $\rect a{\kindsl}{a;!\intt}$.
Following Harper~\cite{DBLP:books/cu/Ha2016}, \emph{contractivity on a given
  type variable} is captured by judgement $\isContr aT$, read as ``type $T$ is
contractive on type variable $a$ over a set of polymorphic variables $\tvarset$''. The
rules are in \cref{fig:contractivity}. 
The set of type variables $\tvarset$ is expected to contain the polymorphic
variables in scope \changed{at the time of the call to the 
contractivity judgement (see~\cref{fig:kinding}, rule \rulenamekindRec).}

Contractivity is defined with respect to a given type variable. For example,
$\rect b \kindsl {\semit ba}$ is contractive on type variable $a$ (under the
empty set of variables) even if $\semit ba$ is not contractive on $b$.
Central to the definition of contractivity is \rulenamecontrVar\ rule. The
contractivity of $\semit ba$ on $a$ depends on $b$ not being a polymorphic
variable: when $b$ is a recursion variable we have
$\isContr[]{a}{\semit ba}$, otherwise, when $b$ is a polymorphic variable, we
have $\isNotContr[b]{a}{\semit ba}$.
Given that polymorphic variables may be replaced by arbitrary types, if $b$
turns out to be replaced by $\skipt$, we are left with $\semit {\skipt}a$ which,
by rule \rulenamecontrSeqOne, is not contractive on $a$, because
$\isNotContr[\emptyset]{a}{a}$.

The set of polymorphic type variables in the judgement allows for the
contractive predicate to be closed under substitution. We note, however, that
the same problem does not arise when a polymorphic type occurs in the body of a
$\mu$-type (e.g., $\rect a \kind {\forallt b {\kind'} T}$). In this case, the
corresponding polymorphic variable is not added to the context, as evidenced
by the rule \rulenamecontrPoly. Contractivity is a concept that concerns the
behavior of $\kinds$-kinded types and a polymorphic type is of kind $\kindt$
(see~\cref{fig:kinding}, rule \rulenamekindPoly). Would $b$ be added to the
context in the \rulenamecontrPoly rule, type
$\rect a \kindtu {\forallt b \kind b}$ would not be well-formed, without any
apparent reason since the body of the $\mu$-type has kind~$\kindt$.

%%%% SUBKINDING

\paragraph{Subkinding}

\begin{figure}[t]
  \declthreerel{Subkinding}{$m \subk m$}{$v \subk v$}{$\kind \subk \kind$}
    \begin{gather*}
    \axiom{}{\un \subk \lin}
    \qquad\qquad
    \axiom{}{\kindm \subk \kindt}
    \quad
    \axiom{}{\kinds \subk \kindt}
    \quad
    \infrule{}{\prekind_1 \subk \prekind_2\\ m_1 \subk m_2}{\prekind_1^{m_1} \subk \prekind_2^{m_2}}
  \end{gather*}
  \caption{Subkinding}
  \label{fig:subkinding}
\end{figure}

%%% Local Variables: 
%%% mode: latex
%%% TeX-master: "main"
%%% End: 

\begin{wrapfigure}{r}{0.18\textwidth}
  \begin{tikzpicture}[scale=.67]
    \node (TL) at (0,1) {$\kindtl$};
    \node (TU) at (0,0) {$\kindtu$};
    \node (SL) at (1,0) {$\kindsl$};
    \node (SU) at (1,-1) {$\kindsu$};
    \node (ML) at (-1,0) {$\kindml$};
    \node (MU) at (-1,-1) {$\kindmu$};
    \draw (TL) -- (TU) -- (SU) -- (SL) -- (TL);
    \draw (MU) -- (ML) -- (TL);
    \draw (MU) -- (TU);
  \end{tikzpicture}
\end{wrapfigure}
Given the intended meaning for $\un$ and $\lin$ multiplicities it should be
clear that an unrestricted type can be used where a linear type is sought. On a
similar vein, a value that can be used in a message ($\kindmu$ or $\kindml$) can
be used where an arbitrary type ($\kindtl$) is required, and similarly for a
value of a session type. This gives rise to a notion of subkinding
$\kind \subk \kind'$ governed by the rules in \cref{fig:subkinding} and depicted
by the partial order at the right.

%%%% TYPE FORMATION

\paragraph{Type Formation}

\begin{figure}[t]
  \declrel{Type formation}{$\istype \Delta T \kind$}
  \begin{gather*}
    \infrule{\rulenamekindType}{\isType {\dom(\Delta)} \Delta T \kind }
    {\istype \Delta T \kind}
  \end{gather*}  
  \declrel{Type formation (inner)}{$\isType \tvarset \Delta T \kind$}
  \begin{gather*}
    % Session types
    \axiom{\rulenamekindSkip}{\isType{\tvarset}{\Delta}{\skipk}{\kindsu}}
    \quad
    \infrule{\rulenamekindMsg}
    	{\isType{\tvarset}{\Delta} T \kindml}
    	{\isType{\tvarset}{\Delta} {\msgt T}\kindsl}
    \quad
    \infrule{\rulenamekindCh}
    	{\isType{\tvarset}{\Delta}{T_\ell}\kindsl \\ \changed{(\forall\ell\in L)}}
    	{\isType{\tvarset}{\Delta}{\choicet{\recordt \ell T L}}{\kindsl}}
    \\
    \infrule{\rulenamekindSeq}
    	{\isType{\tvarset}{\Delta} T \kinds^m
      		\\
      	\isType{\tvarset}{\Delta} U \kinds^m}
      	{\isType{\tvarset}{\Delta}{(T;U)} \kinds^m}
    \quad
    % Functional types
    \axiom{\rulenamekindUnit}{\isType{\tvarset}{\Delta}{\unitt}{\changed{\kindm^m}}}
    \\
    \infrule{\rulenamekindArrow}
    	{\isType{\tvarset}{\Delta}{T}{\kind_1}
      	\\
      	\isType{\tvarset}{\Delta}{U}{\kind_2}}
      	{\isType{\tvarset}{\Delta} {\funt{T}{U}}{\kindt^m}}
    \quad
    \infrule{\rulenamekindRcd}
    	{\isType{\tvarset}{\Delta} {T_\ell}{\kindt^{m}} \\ \changed{(\forall\ell\in L)}}
    	{\isType{\tvarset}{\Delta}{\recordt \ell T L}{\kindt^m}}
    % \infrule{\rulenamekindPair}
	% 	{\isType{\tvarset}{\Delta}{T}{\kindt^{m}}
    %   	\\
    %   	\isType{\tvarset}{\Delta}{U}{\kindt^{m}}}
    %   	{\isType{\tvarset}{\Delta} {\pairt{T}{U}}{\kindt^{m}}}
    \\
    \infrule{\rulenamekindData}
    	{\isType{\tvarset}{\Delta} {T_\ell}{\kindt^{m}} \\ \changed{(\forall\ell\in L)}}
    	{\isType{\tvarset}{\Delta}{\variantt \ell T L}{\kindt^m}}
    \quad
    \infrule{\rulenamekindPoly}
    	{\isType{\tvarset,a}{\Delta,a\colon\kind} T {\kindt^m}}
    	{\isType{\tvarset}{\Delta}{\forallt a \kind T}{\kindt^m}}
    \\
    % Session or functional
    \infrule{\rulenamekindRec}
     	{\isContr aT
      	 \\
      	 \isType{\tvarset}{\Delta, a\colon\kind} T \kind}
        {\isType{\tvarset}{\Delta} {\rect a\kind T} \kind}
    \quad 
    \infrule{\rulenamekindVar}{a\colon\kind\in\Delta}{\isType{\tvarset}{\Delta} a \kind}
    \quad
	\infrule{\rulenamekindSub}
    	{\isType{\tvarset}{\Delta} T \kind
      	\\
      	\kind \subk \kind'}
      	{\isType{\tvarset}{\Delta} T \kind'}
  \end{gather*}
  \caption{Type formation}
  \label{fig:kinding}
\end{figure}

%%% Local Variables:
%%% mode: latex
%%% TeX-master: "main"
%%% End:

Equipped with the notions of contractivity and subkinding we may address type
formation. Type formation is captured by judgement $\istype{\Delta} T\kind$ stating
that type $T$ has kind $\kind$ under kinding context $\Delta$. The rules are in
\cref{fig:kinding}.
The term language defined in \cref{sec:processes} may talk about types
containing free polymorphic variables but not free recursion variables. The
latter are collected in kinding context $\Delta$. During type formation new type
variables (both polymorphic and recursion) must be moved to the context. We have
seen how important it is to be able to distinguish the two sorts of type
variables. Therefore, judgement $\istype\Delta T\kind$ for (external) type
formation considers all variables in $\Delta$ to be polymorphic. The rules for
inner type formation, judgement $\isType{\tvarset}{\Delta}{T}{\kind}$, are then
able to distinguish polymorphic from recursion variables: contrast rule
\rulenamekindPoly against \rulenamekindRec. They both move the bind
$a\colon\kind$ to context $\Delta$, but only the former rule moves $a$ to
$\tvarset$. This setup allows us not to impose a syntactic distinction of the
two sorts of type variables, important for programming languages.

The rules classify types into session types ($\kindsu,\kindsl$), message types
($\kindmu,\kindml$) and arbitrary ($\kindtu,\kindtl$) types.
Rule \rulenamekindArrow gives kind $\kindt^m$ to an arrow type $\funt{T}{U}$
regardless of the kinds for types $T$ and $U$. This decouples the number of
times a function can be used (governed by $m$) from the number of times the
argument (and the result) can be used~\cite{DBLP:conf/tldi/MazurakZZ10}. For
simplicity, rule \rulenamekindPoly restricts polymorphism to functional types.
The only $\kindsu$ types are those made of $\skipk$, such as $\skipk$ itself and
$\semit{\skipt}{\skipt}$ (that is, terminated types); all other session types
are linear. Only base types can be transmitted in messages (represented by kind
$\kindml$). Recursive types are required to be contractive (rule
\rulenamekindRec and \cref{fig:contractivity}). 

Rule \rulenamekindSub incorporates subkinding (\cref{fig:subkinding}) into type 
formation.
Using \rulenamekindSub, any type is of kind $\kindtl$ and any session type is of
kind $\kindsl$ (including $\skipk$), and so, in a sense, all types may be viewed
as linear. In the prose we often refer to linear types to mean those types whose
only kind is $\prekind^\lin$, that is, types that cannot be classified as
unrestricted.

Continuing with the examples in the discussion of type termination and
contractivity, a simple derivation concludes that
$\istype{}{\rect{a}{\kindsl}{(\sout\unitt ; \rect{b}{\kindsl}{\semit ab}})}{\kindsl}$.
%(even if $b;a$ is not contractive on $b$). 
On the other hand, a polymorphic type that repeatedly performs some
protocol~$a$, namely $\forallt{a}{\kindsl}{\rect{b}{\kindsl}{\semit ab}}$ is not
well formed, because $\isNotContr[a]{b}{\semit ab}$. However the type of
polymorphic streams as seen from the consumer side,
$\forallt{a}{\kindmu}{\rect{b}{\kindsl}{\sint \semit ab}}$, is well formed.

% \paragraph{Why not explicitely kinded pairs?}

The reader may wonder why we explicitly annotate function types $\funt TU$ with
a multiplicity mark $m$, but not record or variant types. The key difference is
that the linearity of a record or variant type is completely determined by that
of its components (a record with a linear component must be linear). This is not
the case with function types where the multiplicity of the type is independent
of those of its components (an unrestricted function may accept and/or return
linear values). Rather, the multiplicity of the function type depends on those
of the values captured in the function closure whose types are not apparent in
the type. Hence, annotating record/variant types is redundant (and may
even lead to inconsistencies), whereas annotating function types is unavoidable.

%%% BINDINGS AND SUBSTITUTION

\subsection{Bindings and Substitution}

The binding occurrences for type variables $a$ are types $\rect a\kind T$ and
$\forallt a\kind T$. The set $\free(T)$ of free type variables in type $T$ is
defined accordingly, and so is the \emph{capture avoiding substitution} of a
type variable~$a$ by a type~$T$ in a type~$U$, denoted by $\subs TaU$.
Type formation may be subject to weakening and strengthening on non free type
variables. \changed{The following lemmas show that we can discard and insert 
variables into the context (under some assumptions), which turns out to be paramount 
to verify that substitution preserves the good properties of types 
(\cref{lem:type-substitution}).}

\begin{lemma}[Type strengthening]
  \label{lem:type-strengthening}
  Let $a \not \in \free(T)$.
  \begin{enumerate}
  \item\label{it:ts-contr} If $\isContr[\tvarset,a] bT$, then
    $\isContr[\tvarset] bT$.
  \item\label{it:ts-kindenv} If
    $\isType{\tvarset}{\Delta,a\colon \kind'} T\kind$, then
    $\isType{\tvarset}{\Delta} T\kind$.
  \item\label{it:ts} If
    $\isType{\tvarset,a}{\Delta, a\colon \kind'} T\kind$, then
    $\isType{\tvarset}{\Delta} T\kind$.
  \end{enumerate}
\end{lemma}
\begin{proof}
  By rule induction on the hypothesis.
  We sketch one illustrative case for \cref{it:ts}: rule \rulenamekindRec with
  $T= \rect b\kind U$. Given
  the $\alpha$-identification convention, we assume that $a \neq b$. The
  premises of rule \rulenamekindRec are $\isContr[\tvarset,a] bU$ and
  $\isType{\tvarset,a}{\Delta, a\colon\kind', b\colon\kind} U \kind$. Since
  $a \not \in \free(\rect b\kind U)$ then $a \not\in\free(U)$ and so, by
  induction, we obtain $\isType{\tvarset}{\Delta, b\colon\kind} U \kind$.
  \Cref{it:ts-contr} gives $\isContr[\tvarset]bU$. The result follows from rule
  \rulenamekindRec.
\end{proof}

\begin{lemma}[Type weakening]\ 
  \label{lem:type-weakening}
  \begin{enumerate}
  \item If $\isContr{b}{T}$, then $\isContr[\tvarset,a]{b}{T}$.
  \item\label{it:te-type} If $\isType{\tvarset}{\Delta} T\kind$, then
    $\isType{\tvarset}{\Delta,a\colon \kind} T\kind$ and
    $\isType{\tvarset,a}{\Delta,a\colon \kind} T\kind$.
  \end{enumerate}
\end{lemma}
\begin{proof}
  By straightforward rule induction on the first hypothesis.
\end{proof}

\changed{\Cref{lem:type-substitution} shows that substitution} preserves
termination, contractivity and type formation, \changed{fundamental properties
  to guarantee the proper behaviour of contexts when subject to substitutions,
  as we elaborate in~\cref{sec:processes}.}

\begin{lemma}[Type substitution]
  \label{lem:type-substitution}
  Suppose that $\isType{\tvarset}{\Delta} U \kind$.
  \begin{enumerate}
  \item\label{it:subs-done} If $\isType{\tvarset,a}{\Delta,a\colon\kind} T\kind'$
    and $\isDone T$, then $\isDone {\subs UaT}$.
  \item\label{it:subs-contr} If $\isType{\tvarset,a}{\Delta,a\colon\kind, b \colon \kind'} T\kind''$
    and \changed{$\isContr[\tvarset,a] bT$} and $b\notin\tvarset\cup\free(U)$, then
    $\isContr b {\subs UaT}$.
  \item\label{it:subs-type} If $\isType{\tvarset,a}{\Delta,a\colon\kind} T\kind'$,
    then $\isType{\tvarset}{\Delta} {\subs UaT} \kind'$.
  \end{enumerate}
\end{lemma}
\begin{proof}
  \Cref{it:subs-done}. By rule induction on $\isDone T$. 
  
  \Cref{it:subs-contr}. By rule induction on \changed{$\isContr[\tvarset,a] bT$}, using 
  \cref{it:subs-done}. All cases follow by induction hypothesis
  except case
  for rule \rulenamecontrVar, where we have  $T=c$ with $c \neq a,b$;
  the result follows from the hypothesis \changed{$\isContr[\tvarset,a] bc$}.
  % Rule \rulenamecontrSeqOne with $T=V;W$. Premises are (1)
  % $\isType{\tvarset}{\Delta,a\colon\kind} V\kind'$ and (2)
  % $\isType{\tvarset}{\Delta,a\colon\kind} W\kind'$.
  % %
  % If $\isContr{b}{V}$ and (1) then, by induction, we obtain
  % $\isContr{b}{\subs UaV}$ and conclude by applying rule \rulenamecontrSeqOne.
  % % 
  % If $\isDone V$, then $V$ does not have
  % free variables, hence $\isDone{\subs UaV}$. Similarly to the previous
  % subcase, by induction, we have $\isContr{b}{\subs UaW}$ and we conclude that
  % $\isContr b {\subs UaT}$ using rule \rulenamecontrSeqTwo. Case
  % $T=\rect{c}{\kind''}{S}$ follows by weakening (\cref{lem:weakening}),
  % induction hypothesis and the rule \rulenamecontrRec.

  \Cref{it:subs-type}. By rule induction on
  $\isType{\tvarset,a}{\Delta,a\colon\kind} T\kind'$. 
  For the case \rulenamekindRec we have $T=\rect{b}{\kind''}{V}$. The premises to the rule
  are $\isContr[\tvarset,a] bV$ and
  $\isType{\tvarset,a}{\Delta,a\colon\kind, b\colon \kind''} V\kind'$. 
  % Weakening
  % (\cref{lem:type-weakening}) on the second hypothesis followed by 
  \changed{Induction on the second premise gives}
  $\isType{\tvarset}{\Delta,b\colon \kind''} {\subs UaV}\kind'$.
  \Cref{it:subs-contr} gives $\isContr b{\subs UaV}$. Rule \rulenamekindRec
  gives $\isType{\tvarset}{\Delta}{\rect b{\kind''}{\subs UaV}}{\kind'}$.
  Conclude with the definition of substitution.
  For the case \rulenamekindVar with $T=b\neq a$, we have $a\notin\free(b)$. The result
  follows from hypothesis $\isType{\tvarset,a}{\Delta,a\colon\kind} b\kind'$ and
  strengthening (\cref{lem:type-strengthening}).
  For \rulenamekindVar with $T=a$ we have $k=k'$. The result follows from
  the hypothesis $\isType{\tvarset}{\Delta} U \kind$.
\end{proof}

% Notice that there is no substitution lemma for types in the absence of the
% $\alpha\checkmark$-identification convention. There are types $T$ and $U$ for
% which $\subs{U}{\alpha}{T}$ is not a type. Take for $T$ the type
% $\forallt{a}{\kindsl}{\rect{b}{\kindsl}{a;b}}$ of a stream of elements of
% session type $a$. We can show that $\isType{\tvarset}{}{T}{\kindsl}$. Clearly
% $\isType{\tvarset}{\empty}{\skipt}{\kindsl}$, hence the substitution $\subs{\skipt}{a}{T}$
% is defined and equal to $\rect{b}{\kindsl}{\skipt;b}$. But this is not well
% formed because $\skipt;b$ is not contractive on type variable $b$. At this
% point, noting that $\rect{b}{\kindsl}{\skipt;b}$ is terminated, we make use of
% the $\checkmark$-convention and use instead type $\skipt$. Then,
% $\subs{\skipt}{a}{T}= \skipt$ and $\isType{\tvarset}{\empty}{\subs \skipt a T}{\kindsl}$
% as expected.

%% TYPE EQUIVALENCE
%% the next lemma does not seem to be used in the current agreement for type equivalence
The following lemma is used in agreement for type equivalence
(\cref{lem:agreement-type-equiv}). It should hold for all kinds, but we only
need it for non session types.
\begin{lemma}[Inversion of substitution]
  \label{lem:inversion-subs}
  If $\isType{\tvarset}{\Delta}{\subs T a U} \kind$ and $a\notin\free(T)$ and
  $\kind\neq\kinds^m$, then $\isType{\tvarset}{\Delta,a\colon\kind}{U}{\kind}$.
\end{lemma}
\begin{proof}
  By rule induction on the hypothesis. We sketch a few cases.
  
  Case \rulenamekindVar with $U=b\neq a$. The premise is $b\colon\kind \in \Delta$;
  rule \rulenamekindVar gives $\isType{\tvarset}{\Delta}{b}{\kind}$ and
  weakening (\cref{lem:type-weakening}) yields the result.
  
  Case \rulenamekindVar with $U=a$, conclude with rule
  \rulenamekindVar.
  
  Case \rulenamekindPoly. We have $U=\forallt{b}{\kind'}{U'}$ and
  $\subs TaU = \forallt{b}{\kind'}{(\subs Ta{U'})}$ and $\kind\neq \kinds^m$. The
  premise of the rule is
  $\isType{\tvarset,b}{\Delta,b\colon\kind'}{\subs Ta{U'}}{\kind}$. Induction
  gives $\isType{\tvarset,b}{\Delta,b\colon\kind',a\colon\kind}{U'}{\kind}$ and
  rule \rulenamekindPoly gives the result.
\end{proof}

\subsection{Type Equivalence}
\label{sec:type-equivalence}

Type equivalence for context-free session types is based on bisimulation. In
contrast, the same notion for regular types is a standard equi-recursive
definition~\cite{DBLP:books/daglib/0005958}.

\paragraph{Session Type Bisimilarity}

\begin{figure}
  \begin{align*}
    \text{Labels} && \lambda \grmeq &
     \msgt T \grmor \star \ell \grmor a %\grmor \forall a
  \end{align*}
  \declrel{Labelled transition system}{$T \LTSderives T$}
  \begin{gather*}
  	\axiom {\rulenameltsVar}
  		{a \LTSderives[a] \skipt}	
  	\qquad
  	% \axiom {\rulenameltsCoVar}
  	% 	{\dual a \LTSderives[\dual a] \skipt}	
  	% \qquad
  	\axiom {\rulenameltsMsg}
  		{\msgt T \LTSderives[\msgt T] \skipt}
  	\qquad
  	\infrule {\rulenameltsCh}
  		{k\in L}
  		{\choicet{\recordt \ell T L} \LTSderives[\star k] T_k}
  	\\
  	\infrule {\rulenameltsSeqOne}
  		{T \LTSderives T'}
  		{\semit T U \LTSderives \semit {T'}U}	
  	\qquad
  	\infrule {\rulenameltsSeqTwo}
  		{\isDone[\empty]{T} \\ U \LTSderives U'}
  		{\semit TU \LTSderives U'}	
  	\qquad
  	\infrule {\rulenameltsRec}
  		{\subs{\rect a \kind T}aT \LTSderives U}
  		{\rect a \kind T\LTSderives U}
  \end{gather*}
  \caption{Labelled transition system}
  \label{fig:lts}

\end{figure}

%\begin{figure}
%  \declrel{Type equivalence for session types}{$\isequivst TT$}
%  % 
%  \begin{gather*}
%    \axiom{ \isequivst{\msgt B}{\msgt B}}
%    \qquad
%    \axiom{\isequivst{\alpha}{\alpha}}
%    \qquad
%    \frac{(\forall \ell \in L) \enspace\isequivst{T_\ell}{U_\ell}}{\isequivst{\choicet{\ell\colon T_\ell}_{\ell\in L}}{\choicet{\ell\colon U_\ell}_{\ell\in L}}}%TODO: star
%    \\
%    \frac{\isequivst{T}{T'}}{\isequivst{\semit T U}{\semit {T'}U}}
%    \qquad
%    \frac{\isDone{T} \quad \isequivst{U}{U'}}{\isequivst{\semit TU}{U'} }
%    \qquad
%    \frac{\isequivst{\subs{\rect a \kind T}aT}{U}  }{\isequivst{\rect a \kind T}{U}}
%  \end{gather*}
%  \caption{Alternative type equivalence}
%  \label{fig:lts}
%\end{figure}

%%% Local Variables:
%%% mode: latex
%%% TeX-master: "main"
%%% End:

The labelled transition system for session types is in~\cref{fig:lts}.
The labels we consider are $\sout T$ and $\sint T$ for output and input messages
of type $T$, $\oplus\ell$ and $\&\ell$ for internal and external choices on
label $\ell$, and $a$ for transition on a \emph{polymorphic} variable.
%
% The rules should be self explanatory. 
There are two rules for sequential
composition $T;U$ depending on $T$ being terminated or not.

A bisimulation is defined in the usual way from the labelled transition
system~\cite{sangiorgi2014introduction}.
We say that a type relation $\mathcal R$ is a \emph{bisimulation} if for $T$ and
$U$ whenever $T\mathcal RU$, for all~$\lambda$ we have:
\begin{itemize}
\item for each $T'$ with $T \LTSderives T'$, there is $U'$ such that $U
  \LTSderives U'$ and $T'\mathcal RU'$, and
\item for each $U'$ with $U \LTSderives U'$, there is $T'$ such that $T
  \LTSderives T'$ and $T'\mathcal RU'$.
\end{itemize}
We say that two types are bisimilar, written $\isequivst T U$, if there
is a bisimulation~$\mathcal R$ with $T\mathcal RU$.
Below we identify some laws of session bisimilarity.
\begin{itemize}
\item $\skipt$ is left and right neutral for sequential composition; sequential
  composition is associative; choice right-distributes over sequential
  composition;
\item $\mu$-bindings for variables not free in the body can be eliminated;
  % two consecutive bindings can be contracted; 
  unfold preserves bisimilarity.
\end{itemize}

\changed{
\begin{lemma}[Properties of type bisimilarity]\
  \label{lemma:laws}
  % Let $\istype{\Delta} {T_i} \kindsl$ for $i=1,2,3$, and $\istype{\Delta} {T_\ell} \kindsl$ 
  % for $\ell \in L$.
  %
  % Let $\istype{\Delta} {T_i} \kindsl$ with $\isequivst{T_i}{T_j}$ for
  % $i,j=1,2,3$, and $\istype{\Delta} {T_\ell,U_\ell} \kindsl$ with $\isequivst{T_\ell}{U_\ell}$
  % for $\ell \in L$.
  %
  \begin{enumerate}
    \item\label{it:skip}
      $\isequivst{\semit \skipt T}{\isequivst {\semit {T}\skipt}T}$
    \item\label{it:semi}
      $\isequivst{\semit {(\semit {T_1}T_2)} T_3}{\semit {T_1}{(\semit {T_2}T_3)}}$
    \item\label{it:ch}
      $\isequivst
        {\semit {\choicet{\recordt \ell T L}}{T}}
        {\choicet{\recordm \ell \colon {\semit {T_\ell}T}{\ell \in L}}}$
    \item\label{it:mu1}
      $\isequivst{\rect a \kind T}{T}$, if $a \notin \free (T)$
    \item $\isequivst{\rect a \kind T}{\subs{\rect a \kind T}{a}{T}}$
%    \item\label{it:mu2}
%      $\isequivst{\rect a \kind {(\rect b \kind T)}}{\rect a \kind {\subs
%          ab{T}}} \STypeEquiv \rect b \kind {\subs ba{T}}$
  \end{enumerate}
\end{lemma}
\begin{proof}
  Each law is proved by exhibiting a suitable bisimulation. We show a couple of
  examples.
  For~\cref{it:ch} we use the relation $\mathcal R$ that contains the identity
  relation as well as all pairs of the form
  $(\semit {\choicet{\recordt \ell T L}}{T}, \choicet{\recordm \ell \colon
    {\semit {T_\ell} T}{\ell \in L}})$.
  For~\cref{it:mu1}, considering $\mathcal R$ to be a bisimulation that
  contains the identity relation, then the bisimulation we seek is
  $\mathcal R \cup \{(\rect a \kind T,T)\}$.
  %
  % For~\cref{it:mu2}, if $\mathcal R$ is a bisimulation that contains the
  % identity relation, % the bisimulation for $\rect a \kind {(\rect b \kind T)}$
  % % and ${\rect a \kind {\subs ab{T}}}$ is given by
  % we consider
  % $\mathcal R\cup\{(\rect a \kind {(\rect b \kind T)}, \allowbreak\rect a \kind
  % {\subs ab{T}}), (\rect b \kind {\subs {\rect a \kind {(\rect b \kind
  %       T)}}a{T}}, \rect a \kind {\subs ab{T}})\}$.
\end{proof}
}

\paragraph{Type Equivalence}

\begin{figure}[t!]
  \declrel{Type equivalence}{$\isequiv[\changed{\Theta}] \Delta TT\kind$}
  \begin{gather*}
    % Session Types
  	\infrule{\rulenameeqST}
    {
      \istype{\Delta} {T} \kinds^m
      \\
      \istype{\Delta} {U} \kinds^m
      \\
      \isequivst TU}
    {\isequiv\Delta{T}{U}{\kinds^m}}
    \quad
    % Functional Types
  	\infrule{\rulenameeqUnit}
    {}
    {\isequiv\Delta \unitt \unitt{\changed{\kindm^m}}}
    \\
    \infrule {\rulenameeqArrow}
    {\isequiv\Delta{T}{T'}{\kind_1} \\ \isequiv\Delta{U}{U'}{\kind_2}}
    {\isequiv\Delta{\funt {T}{U}}{\funt {T'}{U'}}{\kindt^m}}
    \;\;
    \infrule {\rulenameeqVar}
    {
      \kind \neq \kinds^m
      \quad
      a\colon\kind\in\Delta
    }
    {\isequiv\Delta aa{\kind}}
    \\
    \infrule {\rulenameeqRcd}
    {\isequiv\Delta{T_\ell}{U_\ell}{\kindt^m} \\ \changed{(\forall\ell\in L)}}
    {\isequiv\Delta{\recordt \ell T L}{\recordt \ell
        {U}L}{\kindt^m}}
    \quad
    \infrule {\rulenameeqData}
    {\isequiv\Delta{T_\ell}{U_\ell}{\kindt^m} \\ \changed{(\forall\ell\in L)}}
    {\isequiv\Delta{\variantt \ell T L}{\variantt \ell
        {U}L}{\kindt^m}}
    \\
    % Universal types
    \infrule {\rulenameeqPoly}
    {\isequiv{\Delta, a\colon\kind} T U {\kindt^m}}
    {\isequiv\Delta{\forallt a \kind T}{\forallt a \kind U}{\kindt^m}}
    % Note: once we have type operators, replace \kind by * or \prekind^m	
    % Recursive Types
    \;\;
    % We have been here before
\changed{    \infrule{\rulenameeqFix}{
      (\Delta,T,U,\kind) \in \Theta
      \quad
      \changed{\istype{\Delta}{T}{\kind}}
      \quad
      \changed{\istype{\Delta}{U}{\kind}}
    }{
      \isequiv \Delta TU\kind
    }
}    \\
    \infrule {\rulenameeqRecL}
    {
      \kind \neq \kinds^m
      \\
      \changed{(\Delta',\rect a\kind T,U,\kind') \notin \Theta}
      \\
      \isContr[\dom(\Delta)] aT
      \\
      \isequiv[\Theta,\changed{(\Delta,\rect a\kind T,U,\kind)}]\Delta{\subs{\rect a\kind T}{a}{T}}{U}{\kind}
    }{
      \isequiv\Delta{\rect a\kind T}{U}{\kind}
    }
    \\
    \infrule {\rulenameeqRecR}
    {
      T \text{ not } \mu
      \\
      \kind \neq \kinds^m
      \\
      \changed{(\Delta', T, \rect a\kind U,\kind') \notin \Theta}
      \\
      \isContr[\dom(\Delta)] aU
      \\
      \isequiv[\Theta,\changed{(\Delta,T, \rect a\kind U,\kind)}] \Delta T {\subs{\rect a\kind U}{a}{U}}{\kind}
    }
    {\isequiv\Delta T {\rect a\kind U}{\kind}}
  \end{gather*}
  \caption{Type equivalence}
  \label{fig:type-equivalence}
\end{figure}

%%% Local Variables:
%%% mode: latex
%%% TeX-master: "main"
%%% End:

\Cref{fig:type-equivalence} introduces the type equivalence rules for arbitrary
types. The rules use judgements of the form \changed{$\isequiv\Delta TU\kind$,
  where $\Theta$ is a list of kinded pairs of types (quadruples
  $(\Delta,T,U,\kind)$ such that $\istype \Delta {T,U} \kind$)}.
% the judgement $\isequiv\Delta TU\kind$.
%
% The definition is coinductive, meaning that type equivalence is the largest
% relation that is backward closed for the rules in the figure.
%
% In order to proceed by induction we keep track of the pairs of types visited
% ($\typeequation TU\kind$) in an equation context $\Xi$.
%
Rule \rulenameeqST incorporates session type bisimilarity into type equivalence;
in this case the two types are required to be session types.
% For agreement purposes this and other axioms require the equation context
% to be well formed.
The rules for functional types (\rulenameeqUnit, \rulenameeqArrow,
\rulenameeqRcd, \rulenameeqData, \rulenameeqPoly and \rulenameeqVar) are the
traditional congruence rules.
\changed{Rule \rulenameeqVar is expected to be applied to polymorphic variables only,
that is, neither session variables (hence the premise
$\kind \neq \kinds^m$) nor recursion variables (hence $a\colon\kind\in\Delta$).}
Those for equi-recursion \changed{(\rulenameeqFix, \rulenameeqRecL and
  \rulenameeqRecR) keep track of a set $\Theta$ containing the pairs of types
  visited so far, in which one of the components is recursive. This is
  essentially the Amadio-Cardelli system~\cite{DBLP:journals/toplas/AmadioC93},
  brought to session types (in the context of subtyping) by Gay and
  Hole~\cite{DBLP:journals/acta/GayH05}. The algorithmic reading of the rules is
  known to be
  exponential~\cite{DBLP:conf/tacas/LangeY16,DBLP:books/daglib/0005958}. Pierce
  proposes a polynomial variant that can be
  used in concrete implementations~\cite{DBLP:books/daglib/0005958}.}

\begin{lemma}[Agreement for type equivalence]
  \label{lem:agreement-type-equiv}
  % If $\istype{\Delta}{T}{\kind}$ and $\isequiv [\Empty] \Delta TU{\kind'}$, then
  % $\istype{\Delta}{U}{\kind}$.
  \changed{
  \begin{sloppypar}
    If $\isequiv \Delta TU\kind$, then
    $\istype{\Delta}{T}{\kind}$.
  \end{sloppypar}}
\end{lemma}
\begin{proof}
  % We show a stronger result, namely that if:
  % \begin{itemize}
  % %\item \changed{$\istype{\Delta}{T}{\kind}$ (by rule \rulenamekindType this means that $\isType {\dom(\Delta)} \Delta T \kind$)} and
  % \item $\Theta$ is well-formed (that is, $(\Delta, T,U,\kind)\in\Theta$
  %   implies $\isType \tvarset \Delta T \kind$) and
  % \item  $\isequiv \Delta TU{\kind'}$,
  % \end{itemize}
  % then $\istype{\Delta}{\changed{T},U}{\kind}$.
  % %
  The proof is by rule induction on the hypothesis.

  % All cases (including \rulenameeqFix) go smoothly except \rulenameeqRecL and
  % \rulenameeqRecR.
  Case \rulenameeqST. In this case $\kind = \kinds^m$ and the premises state that $\istype{\Delta}{T}{\kinds^m}$.

  Case \rulenameeqUnit. In this case, $T = \unitt{}$, $\kind = \kindm^m$ and we have $\istype{\Delta}{T}{\kindm^m}$.
  
  Case \rulenameeqArrow, \rulenameeqRcd, \rulenameeqData.
  % Follow by induction hypothesis.
  For \rulenameeqArrow, we know that $T$ is of the form $\funt {T_1}{T_2}$, $\kind = \kindt^m$
  and the premises are 
  $\isequiv\Delta{T_1}{U_1}{\kind_1}$ and $\isequiv\Delta{T_2}{U_2}{\kind_2}$, 
  where $U = \funt {U_1}{U_2}$.
  By induction hypothesis we have 
  $\istype{\Delta}{T_1}{\kind_1}$ and $\istype{\Delta}{T_2}{\kind_2}$.
  Applying \rulenamekindType and \rulenamekindArrow, we conclude that $\istype{\Delta}{\funt {T_1}{T_2}}{\kindt^m}$.
  The other cases are similar.

  Case \rulenameeqPoly. In this case $T$ is of the form $\forallt{a}{\kind}{T'}$
  and $U$ is $\forallt{a}{\kind}{U'}$. The
  premise is $\isequiv{\Delta, a\colon\kind} {T'} {U'} {\kindt^m}$. 
  By induction hypothesis, we have $\istype{\Delta, a\colon\kind}{T'}{\kindt^m}$,
  \ie $\isType {\dom(\Delta), a} {\Delta, a\colon\kind} {T'} {\kindt^m}$.
  Applying \rulenamekindPoly we conclude that $\isType {\dom(\Delta)} {\Delta} T {\kindt^m}$,
  which means that $\istype{\Delta}{T}{\kindt^m}$.

  Case \rulenameeqVar. Here, $T=U=a$ and $a\colon\kind\in\Delta$. Using 
  \rulenamekindVar we get $\istype{\Delta}{a}{\kind}$.

  Case \rulenameeqFix. The premises state that $\istype{\Delta}{T}{\kind}$.

  Case \rulenameeqRecL. In this case $T$ is of the form $\rect a\kind {T'}$.
  Premises include $\kind \neq \kinds^m$ and $\isContr[\dom{(\Delta)}] a{T'}$
  and $\isequiv[\Theta,(\Delta,T,U,\kind)]\Delta{\subs{T}{a}{T'}}{U}{\kind}$.
  By induction hypothesis, we know that
  $\istype{\Delta}{\subs{T}{a}{T'}}{\kind}$.
  Applying \rulenamekindType gives $\isType{\dom\Delta}{\Delta}{\subs{T}{a}{T'}}{\kind}$.
  By~\cref{lem:inversion-subs}, we have $\isType{\dom\Delta}{\Delta,a\colon \kind }{T'}{\kind}$.
  Using \rulenamekindRec, we conclude that $\istype{\Delta}{T}{\kind}$.
    % By~\cref{lem:inversion-subs}, we have
    % $\istype{\Delta,a\colon \kind }{T'}{\kind}$. Applying \rulenamekindType and
    % \rulenamekindRec, we conclude that $\istype{\Delta}{T}{\kind}$.

  Case \rulenameeqRecR. 
  Premises include $\isequiv[\Theta,(\Delta,T,U,\kind)]\Delta{T}{\subs{U}{a}{U'}}{\kind}$, for $U=\rect a\kind{U'}$.
  By induction hypothesis we conclude that $\istype{\Delta}{T}{\kind}$.
\end{proof}

\begin{lemma}[Type equivalence]
  \label{lem:type-equivalence}
  The relation $\isequiv\Delta TU\kind$ is reflexive, transitive, and symmetric.
\end{lemma}
\begin{proof}
  \changed{By rule induction.} % on the formation of type $T$.}
\end{proof}

%%% DUALITY

\subsection{Session Type Duality}

\begin{figure}[t]
  \decltworel{Duality on polarities and views}{$\dual \sharp = \sharp$}{$\dual \star=\star$}
  \begin{align*}
    \dual {!} = {?} 
    &&
     \dual {?} = {!} 
    &&
      \dual {\oplus} = \&
    &&
      \dual {\&} = \oplus
  \end{align*}

  \declrel{The duality function on session types}{$\dual T=T$}
  \begin{gather*}
    \dual \skipk = \skipk
    \qquad
    \dual {\msgt T} = \dual\sharp T
    \qquad
    \dual{\choicet{\recordt \ell TL}} = \dual \star \recordt \ell {\dual T} L
    \\
    \dual{T;U} = \dual T; \dual U
    \qquad
    \dual a = a
    \qquad
    \dual{\rect{a}{\kind}{T}} = \rect{a}{\kind}{\dual T}
  \end{gather*}
  \caption{The duality function on session types}
  \label{fig:duality}
\end{figure}

%%% Local Variables:
%%% mode: latex
%%% TeX-master: "main"
%%% End:

% THE RELATION VERSION

  % %
  % \decltworel{Duality on polarities and views}{$\sharp\bot\sharp$}{$\star\bot\star$}
  % %
  % \begin{align*}
  %   {!} \bot {?} 
  %   &&
  %      {?} \bot {!} 
  %   &&
  %      {\oplus} \bot \&
  %   &&
  %      {\&} \bot \oplus
  % \end{align*}
  % % 
  % \declrel{Duality on types}{$\isdual TT{\kinds^m}$}
  % % 
  % \begin{gather*}
  %   \axiom{\rulenamedualSkip}{\isdual{\skipk}{\skipk}{\kindsu}}
  %   \qquad
  %   \infrule{\rulenamedualMsg}{\isType T \kindml}{\isdual{\msgt T}{\msgdualt T}{\kindsl}}
  %   \\
  %   \infrule{\rulenamedualCh}
  %   {\isdual{T_\ell}{U_\ell}{\kinds^m}}
  % 	{\isdual{\choicet{\recordt \ell TL}}{\choicedualt{\recordt \ell UL}}{\kindsl}}
  %   \quad
  %   \infrule{\rulenamedualSeq}
  % 	{\isdual{T_1}{U_1}{\kinds^m}
  %     \\
  %     \isdual{T_2}{U_2}{\kinds^n}}
  %   {\isdual{\semit{T_1}{T_2}}{\semit{U_1}{U_2}}{\kindsl}}
  %   \\
  %   \axiom{\rulenamedualVar}
  % 	{\isdual[\Delta, a\colon{\kinds^m}] aa{\kinds^m}}
  %   \qquad
  %   \infrule{\rulenamedualRec}
  % 	{\isContr aT
  %     \\
  %     \isdual[\Delta,a\colon{\kinds^m}] TU{\kinds^m}}
  % 	{\isdual{\rect a{\kinds^m} T}{\rect a{\kinds^m} U}{\kinds^m}}
  % \end{gather*}

Duality is a notion central to session types. It allows ``switching'' the point
of view from one end of a channel (say, the client side) to the other end (the
server side). The rules for duality are in \cref{fig:duality}.
The simplified formulation for recursion variables and recursive types is
justified by the fact that the types we consider are first order (channels
cannot be conveyed in messages), thus avoiding a complication known to arise in
the presence of
recursion~\cite{DBLP:conf/concur/BernardiH14,DBLP:journals/corr/BernardiH13,DBLP:journals/corr/abs-2004-01322}.
The properties below state that duality is idempotent, preserves termination and
contractivity, and that every session has a dual.

\begin{lemma}[Properties of duality]\
 \label{lem:dual-preserves-contractivity}
 \begin{enumerate}
 \item If $\dual T=U$, then $\dual U=T$
  \item\label{it:dual-done} If $\isDone T$ and $\dual T=U$, then $\isDone U$. 
 \item\label{it:dual-contr} If $\isContr aT$ and $\dual T=U$, then
   $\isContr aU$.
 \item\label{it:has-dual} If $\istype{}{T}{\kinds^m}$, then $\dual T=U$ and
   $\istype{}{U}{\kinds^m}$ for some $U$.
 \end{enumerate}
\end{lemma}
\begin{proof}
  By straightforward rule induction on the first hypothesis, in all cases.
\end{proof}

%%% Local Variables:
%%% mode: latex
%%% TeX-master: "main"
%%% End:

% LocalWords:  prekinds iff contractive

\section{Expressions and Processes, Statics and Dynamics}
\label{sec:processes}

This section introduces the terms, typing and operational semantics of \fmusession.

\subsection{Expressions and Processes}

%%% PROCESSES

\begin{figure}[t]
  \begin{align*}
    \text{Constants}&&
    c \grmeq\:&
    \sendk
    \grmor \receivek
    \grmor \forke
    \grmor \unite
    \\
    \text{Values}&&
    v \grmeq\:&
    x
    \grmor c
    \grmor \abse xTe 
    \grmor \tabse a\kind v
    \grmor \recorde  \ell v L
    \grmor \injecte \ell v
    \\
    &&
    \grmor & \selecte \ell {}
    \grmor \tappe \sendk T
    \grmor \appe{\tappe\sendk T} v
    \grmor \tappe {\appe{\tappe\sendk T} v} U % \appe{\tappe{\tappe\sendk T} U} v
    \\&&
    \grmor& \tappe\receivek T \grmor \tappe{\tappe\receivek T} U
    \\
    \text{Expressions} &&
    e \grmeq\:&
    v 
    \grmor \appe e e
    \grmor \recorde  \ell e L
    \grmor\lete \ell xLee
    \\
    &&\grmor& 
    \changed{\unlete{\unite}{e}{e}}
    \grmor \injecte \ell e 
    \grmor \casee e {\recordp \ell e L}
    \\
    &&\grmor&
    \tappe eT \grmor
    \newe T \grmor
    \matche{e}{\recordp \ell e L}
    \\
    \text{Processes}
    && p \grmeq\:&
    \PROC e
    \grmor p \PAR p
    \grmor \NU xxp
    \\
    \text{Typing contexts} &&
    \Gamma \grmeq & \Empty \grmor \Gamma, x\colon T 
  \end{align*}
  \caption{The syntax of values, expressions, processes and typing contexts}
  \label{fig:processes}
\end{figure}

%%% Local Variables:
%%% mode: latex
%%% TeX-master: "main"
%%% End:

\Cref{fig:processes} introduces the syntax of values, expressions and processes.
A value is either a variable (which may stand for a channel end), a constant, an
expression abstraction, a type abstraction, a record whose fields are values, a
variant injection of a value, a partial (type) application of the constants that
denote the operations on channels.
Notice that expression $\tappe{\forkk}{T}$ is not a value because
 $\appe{\tappe{\forkk}{T}}{e}$ is evaluated before its body $e$
becomes a value, see rule \rulenameprocredFork in \cref{fig:reduction}.
The body of a type abstraction is restricted to a syntactic
value as customary. This value restriction guarantees the sound interplay of
call-by-value evaluation and polymorphism in the presence of effects like
channel-based communication~\cite{DBLP:journals/lisp/Wright95}.

An \emph{expression} $e$ is a value, an application, a record, a record
elimination by pattern matching (this is necessary because records may contain
linear fields), \changed{a unit elimination,} an injection in a variant type, a
variant elimination, a type application, a channel creation, or matching on a
tag received over a channel.
\changed{Expression $\letk~(x,y) = e_1~\ink~e_2$ used in examples abbreviates
  $\letk~\{\fstl\colon x,\sndl\colon y\} = e_1~\ink~e_2$.}

A \emph{process} $p$ is either an expression lifted to the process level, a
parallel composition of two processes, or a channel restriction binding the two
ends of a channel in a subsidiary process.

\paragraph{Context Formation and Context Split}

%%%% CONTEXT MULTIPLICITY and CONTEXT SPLIT

% \input{fig-context-formation} % included in context split
\begin{figure}[t!]
  \declrel{Context formation}{$\isCtx \Gamma \kind$}
  \begin{gather*}
    \axiom{\rulenameformEmpty}{\isCtx\Empty\kind}
    \qquad
    \infrule{\rulenameformExt}{
      \isCtx \Gamma \kind \\ \istype\Delta{T}{\kind}
    }{
      \isCtx{\Gamma, x\colon T}\kind
    }
  \end{gather*}
  \declrel{Context split}{$\splitctx \Gamma \Gamma \Gamma$}
  \begin{gather*}
    \axiom{\rulenamesplitEmpty}{\splitctx \Empty \Empty \Empty}
    \qquad
    \infrule{\rulenamesplitUnr}
    	{\splitctx \Gamma {\Gamma_1} {\Gamma_2}\\ \istype{\Delta} T{\kindtu}}
    	{\splitctx {\Gamma, x\colon T}
            	   {(\Gamma_1,x\colon T)}{(\Gamma_2,x\colon T)}}
    \\
    \infrule{\rulenamesplitLeft}
    	{\splitctx \Gamma {\Gamma_1}{\Gamma_2} \\ \istype{\Delta} T{\kindtl}}
    	{\splitctx {\Gamma, x\colon T} {(\Gamma_1,x\colon T)}{\Gamma_2}}
    \qquad
    \infrule{\rulenamesplitRight}
    	{\splitctx \Gamma {\Gamma_1}{\Gamma_2}\\ \istype{\Delta} T {\kindtl}}
    	{\splitctx {\Gamma, x\colon T} {\Gamma_1}{(\Gamma_2,x\colon T)}}
  \end{gather*}
  \caption{Context formation and context split}
  \label{fig:context-split}
  \label{fig:context-formation}
\end{figure}

%%% Local Variables:
%%% mode: latex
%%% TeX-master: "main"
%%% End:

Type formation (\cref{fig:kinding}) is lifted pointwise to contexts in judgement
$\isCtx\Gamma\kind$ (see \cref{fig:context-formation}).
Context split (also in \cref{fig:context-split}) governs the distribution of
linear bindings. If $\splitctx \Gamma {\Gamma_1}{ \Gamma_2}$, then the
unrestricted bindings of $\Gamma$ are both in $\Gamma_1$ and $\Gamma_2$ whereas
linear bindings are either in $\Gamma_1$ or $\Gamma_2$ but not in both (two
rules). Splitting does not generate new bindings. For convenience, we write
$\Gamma_1 \circ \Gamma_2$ in the conclusion of an inference rule instead of
$\Gamma$ with the additional premise $\splitctx\Gamma{\Gamma_1}{\Gamma_2}$. We
need a few technical lemmas about splittings.

\begin{lemma}[Agreement for context split]
  \label{lem:agreement-context-split}
  \label{lemma:context-split-preserves-kinding}
  Let $\splitctx \Gamma {\Gamma_1}{\Gamma_2}$. Then,
  $\isCtx \Gamma \kind$ iff $\isCtx{\Gamma_1} \kind$ and
  $\isCtx {\Gamma_2} \kind$.
\end{lemma}
\begin{proof}
  By rule induction on the context split hypothesis.
\end{proof}

% We introduce basic properties of context split as defined in
% \cref{fig:context-split}.
Given a well-formed context $\Gamma$, we distinguish % $\dom{(\Gamma)}$ denote the set of variables $x$
% such that $x \colon T \in \Gamma$ and let 
% $\mathcal L(\Gamma)$ and
% $\mathcal U(\Gamma)$ to refer to 
the linear and unrestricted portions of
$\Gamma$ and denote them by $\mathcal L(\Gamma)$ and $\mathcal U(\Gamma)$,
respectively~\cite{walker:substructural-type-systems}. %(see \cref{lemma:context-split-preserves-kinding}).

\begin{lemma}[Properties of context split]\label{lem:cSplitProps}
  Let $\splitctx \Gamma {\Gamma_1}{\Gamma_2}$
  \begin{enumerate}
  \item\label{item:split-equalUn}$\mathcal U(\Gamma) = \mathcal U{(\Gamma_1)} = \mathcal U{(\Gamma_2)}$
  \item\label{item:split-varNotBoth} If $x\colon T\in\Gamma$ and $\istype{\Delta}{T}{\kindtl}$, then either
    $x \in \dom{(\Gamma_1)}$ and $x \not\in\dom{(\Gamma_2)}$, or 
    $x \not\in\dom{(\Gamma_1)}$ and $x \in\dom{(\Gamma_2)}$.
  \item\label{item:split-fst} If $\splitctx{\Gamma_1}{\Gamma_{11}}{\Gamma_{12}}$, then there exists
    exactly one $\Gamma'$ such that $\splitctx\Gamma{\Gamma_{11}}{\Gamma'}$ and
    $\splitctx{\Gamma'}{\Gamma_{12}}{\Gamma_2}$.
  % \item\label{item:split-snd}  If $\splitctx{\Gamma_2}{\Gamma_{21}}{\Gamma_{22}}$, then there exists
  %   exactly one $\Gamma'$ such that $\splitctx{\Gamma}{\Gamma'}{\Gamma_{22}}$
  %   and $\splitctx{\Gamma'}{\Gamma_1}{\Gamma_{21}}$.\todo{Used where?}
  \end{enumerate}
\end{lemma}
\begin{proof}
  \Cref{item:split-equalUn,item:split-varNotBoth}. By rule induction on the
  context split hypothesis.
  \Cref{item:split-fst}. By rule induction on the context split for $\Gamma_1$.
\end{proof}

\begin{lemma}[Substitution and context split]
  \label{lem:substitution-context-split}
  If $\splitctx[\Delta, a:\kind] \Gamma {\Gamma_1}{\Gamma_2}$ and
  $\istype\Delta U \kind$, then
  $\subs Ua\Gamma = \subs Ua{\Gamma_1} \circ \subs Ua{\Gamma_2}$.
\end{lemma}
\begin{proof}
  % \Cref{it:subs-ctx}. By rule induction on the hypothesis, using 
  % \cref{lem:type-substitution}, \cref{it:subs-type}.
  % %
  % \Cref{it:subs-split}.
  By rule induction on the context split hypothesis.
\end{proof}

%%% TYPING EXPRESSIONS

\paragraph{Expression and Process Formation}

\begin{figure}[t!]
  \declrel{Types for constants}{$\typeof(c) = T$}
  \begin{align*}
    % \typeof(\newk) = &\
    %   \forallt a \kindsl {\pairt a {\dualoft a}}
    % \\
    \typeof(\sendk) = &\
      \forallt{a}{\kindml}{
        \funt[\un]{a}{\forallt{b}{\kindsl}{
             {\funt[\lin]{\semit {\sout a}b}{b}}}
      }}
    \\
    \typeof(\receivek) = &\
      \forallt{a}{\kindml}{\forallt{b}{\kindsl}{
        \funt[\un] {\semit{\sint a}{b}} {\pairt ab}
      }}
    \\
    \typeof(\forkk) = &\
      \forallt{a}{\kindtu}{\funt[\un]{a}{\unitt[\un]}}
    \\
    \typeof(\unite) = &\ \unitt
  \end{align*}
  % No other constant is of type $()$.
  \caption{Types for constants}
  \label{fig:types-constants}
\end{figure}

%%% Local Variables:
%%% mode: latex
%%% TeX-master: "main"
%%% End:

Types for constants are in \cref{fig:types-constants}. Thanks to polymorphism,
we can give types to most of the session operations. Previous work
\cite{DBLP:journals/jfp/GayV10,DBLP:conf/icfp/ThiemannV16} includes specific
typing rules for these primitives. Inspecting types (and kinds), \changed{we see
  that sending and receiving is restricted to $\kindm$-values.}
In both cases, the continuation is a linear session (of kind $\kindsl$).
\changed{The $\sendk$ and $\receivek$ functions are unrestricted because they do
  not close over linear values.
 % hence can be used as
 %  many times as needed.
  However, a partially applied $\sendk$ must be linear for it captures a linear
  value (that of the value to be sent, of type $a\colon\kindml$).}
The type of $\forkk$ indicates that a new process may return any value of an
unrestricted type, a value that is discarded by the operational semantics.

At a first glance, the use of linear polymorphic variables for the values
exchanged in messages ($a\colon\kindml$ in the types for $\sendk$ and
$\receivek$) may seem restrictive, but the prescription of $\kindml$ to $a$ is
just a convenience for checking type application. A $\sendk$ function, partially
applied to an unrestricted value may still be used multiple times. By
$\eta$-expanding the partial application, we obtain a thunk that can be used
repeatedly. For example, function
$\abse[\un] \_ {\unitt[\un]} {\tappe \sendk \intt~5}$ is a partially applied
$\sendk$ of an unrestricted type:
$\funt[\un]{\unitk[\un]}{\forallt{b}{\kindsl}{ {\funt[\lin]{\semit {\sout
          a}b}{b}}}}$.

The remaining session operations retain their status as language primitives with
specific typing rules. Expression $\newe T$ creates a new channel of session
type $T$ and returns a pair of channel ends of types $T$ and $\dual T$.
Expression $\selectk\ k$ selects option $k$ on a internal choice type.
Expression $\matchk$ receives a label and branches on it.

\begin{figure}[t!]
  \declrel{Expression formation}{$\isExpr \Gamma eT$}
  \begin{gather*}
    % 1 _ FUNCTIONAL PART
    % Constants
    \infrule{\rulenametypeConst}
    {\isCUn \Gamma}
    {\isExpr \Gamma c{\typeof(c)}}
    \quad
    % Variables
    \infrule{\rulenametypeVar}
    {\isCUn \Gamma
      \\
      \istype{\Delta} T\kind}
    {\isExpr{\Gamma,x\colon T} xT}
    \\
    % Term abstraction and elimination
    \infrule{\rulenametypeAbs}
    {\isCtx \Gamma {\kindk \kindt m}
      % When we introduce the type ops this vkind must be v^m
      % \istype{\Delta}{T_1}{\prekind^{m'}}
      \\\!\!\!\!\!
      \isExpr {\Gamma,x\colon T_1} e T_2}
    {\isExpr \Gamma {\abse x {T_1} e}{\funt{T_1}{T_2}}}
    \;\;\;
    \infrule{\rulenametypeApp}
    {\isExpr {\Gamma_1} {e_1}{\funt{T_1}{T_2}}
      \\\!\!\!\!\!
      \isExpr{\Gamma_2}{e_2}{T_1}}
    {\isExpr{\Gamma_1\circ\Gamma_2}{\appe{e_1}{e_2}}{T_2}}    
    \\
    % Split in the premise
    % \infrule{\rulenametypeApp}
    % {\splitctx \Gamma {\Gamma_1}{\Gamma_2}\\
    %   \isExpr {\Gamma_1} {e_1}{\funt{T_1}{T_2}} \\
    %   \isExpr{\Gamma_2}{e_2}{T_1}}
    % {\isExpr{\Gamma}{\appe{e_1}{e_2}}{T_2}}    
    % \\
    % Type abstraction and elimination
    \infrule{\rulenametypePoly}
    {
      \isExpr[\Delta,a\colon\kind]{\Gamma}{v}{T}
      \\
      a \notin \free(\Gamma)
    }
    {\isExpr{\Gamma}{\tabse a\kind v}{\forallt a \kind T}}
    \quad
    \infrule{\rulenametypeTApp}
    {
      \isExpr{\Gamma}{e}{\forallt{a}{\kind}{U}}
      \\ 
      \istype{\Delta}{T}{\kind}
    }
    {\isExpr{\Gamma}{\tappe{e}{T}}{\subs Ta U}}
    \\
    % Product Intro / Elim
    \infrule{\rulenametypeRcd}
    {
      \isExpr{\Gamma_\ell}{e_\ell}{T_\ell}\\
      (\forall\ell\in L)
    }
    {\isExpr{\circ\Gamma_\ell}{\recorde \ell eL}{\recordt \ell TL}}
    \\
    \infrule{\rulenametypeLet}
    {
      \isExpr{\Gamma_1}{e_1}{\recordt \ell TL}
      \\
      \isExpr{\Gamma_2, (x_\ell\colon T_\ell)_{\ell\in L}}{e_2} U}
    {\isExpr{\Gamma_1\circ\Gamma_2}{\lete \ell xL{e_1}{e_2}} U}
    \\
    \changed{
    \infrule{\rulenametypeUnitElim}
    {
      \isExpr{\Gamma_1}{e_1}{\unitk}
      \\
      \isExpr{\Gamma_2}{e_2} T}
    {\isExpr{\Gamma_1\circ\Gamma_2}{\unlete {\unite}{e_1}{e_2}} T}}
    \\    
    % Variant Intro / Elim
    \infrule{\rulenametypeVariant}
    {
      \isExpr \Gamma {e}{T_k}
      \\
      \istype{\Delta}{T_\ell}{\kindtl}
      \\            
      k \in L
      \\
      \changed{(\forall\ell\in L)}
    }
    {\isExpr \Gamma {\injecte ke}{\variantt \ell T L}}
    \\
    \infrule{\rulenametypeCase}
    {% \splitctx \Gamma {\Gamma_1}{\Gamma_2}
     %  \\
      \isExpr{\Gamma_1}{e}{\variantt\ell T L}
      \\
      \isExpr{\Gamma_2}{e_\ell}{\funt[\changed{\lin}] {T_\ell} T}
      \\
      \changed{(\forall\ell\in L)}
       }
    {\isExpr{\Gamma_1\circ\Gamma_2}{\casee e {\recordp \ell e L}}{T}}
    \\
    % % All branches are lin
    % \infrule{\rulenametypeCase}
    % {\splitctx \Gamma {\Gamma_1}{\Gamma_2}
    %   \\
    %   \isExpr{\Gamma_1}{e}{\variantt\ell T L}
    %   \\
    %   \isExpr{\Gamma_2}{e_\ell}{\funt[\lin] {T_\ell} T}}
    % {\isExpr{\Gamma}{\casee e {\recordp \ell e L}}{T}}
    % \\
    % As in Cai et al, taking advantage of records
    % \infrule{\rulenametypeCase}
    % {\splitctx \Gamma {\Gamma_1}{\Gamma_2}
    %   \\
    %   \isExpr{\Gamma_1}{e_1}{\variantt \ell T L}
    %   \\
    %   \isExpr{\Gamma_2}{e_2}{\varianttt \ell {\funt[\un] {T_\ell} T} L}
    % }
    % {\isExpr{\Gamma}{\casee {e_1}{e_2}}{T}}
    % \\
    % 2 _ SESSIONS
    % Select
    \infrule{\rulenametypeSel}
    {
      \isCUn \Gamma
      \\
      \istype{\Delta}{T_\ell}{\kindsl}
      \\
      k \in L
      \\
      \changed{(\forall\ell\in L)}
    }
    {\isExpr \Gamma{\selecte k {}}{\funt{\ichoicet{\recordt \ell T L}}{T_k}}}
    \\
    % Match
    \infrule{\rulenametypeMatch}
    {
      \isExpr{\Gamma_1}{e}{\echoicet{\recordt \ell T L}}\\
      \isExpr{\Gamma_2}{e_\ell}{\funt[\changed{\lin}] {T_\ell} T}\\
      \changed{(\forall\ell\in L)}
    }
    {\isExpr{\Gamma_1\circ\Gamma_2}{\matche e {\recordp \ell e L}}{T}}
    \\
    % New
    \infrule{\rulenametypeNew}
    {\isCUn \Gamma\\
      \istype{\Empty} T \kindsl\\
    }
    {\isExpr \Gamma {\newe T} {\pairt T{\dual T}}}
    \quad
    \infrule{\rulenametypeEq}
    {
      \isExpr \Gamma e {T_1}
      \\
      \isequiv[\Empty]\Delta{T_1}{T_2}{\kind}
    }
    {\isExpr \Gamma e {T_2}}
  \end{gather*}
  \caption{Expression formation}
  \label{fig:typing-exps}
\end{figure}

%%% Local Variables:
%%% mode: latex
%%% TeX-master: "main"
%%% End:

Expression formation is defined by rules in \cref{fig:typing-exps}. Constants are
typed in {\rulenametypeConst} according to \cref{fig:types-constants} under the
condition that the context is unrestricted (\ie all linear values have been
consumed). The type rule for variables, \rulenametypeVar, insists that the
unused part of the context must be unrestricted. In rule \rulenametypeAbs,
context formation under multiplicity $m$ enforces that unrestricted abstractions
do not capture linear values. Application {\rulenametypeApp}, kinded type
abstraction {\rulenametypePoly}, and type application {\rulenametypeTApp} are
all standard. Rule \rulenametypeApp rule splits context in two, one part types
the function, the other the argument. The introduction rule for records
{\rulenametypeRcd} uses the notation $\circ\Gamma_\ell$ for iterated context
split. The corresponding elimination rule {\rulenametypeLet} is standard. We
cannot offer the dot notation to select single fields from records, as such
projection might discard linear values. \changed{The elimination of $\unitk$ is
  given by rule \rulenametypeUnitElim.} The rule {\rulenametypeVariant} and
{\rulenametypeCase} are standard introduction and elimination rules for
variants.

% \begin{lemma}[Agreement for expression formation] % Merged in processes
%   If $\isExpr{\Gamma}{e}{T}$, then $\isCtx{\Gamma}{\kind}$ and
%   $\istype{\Delta}{T}{\kind'}$, for some kinds $\kind,\kind'$.
% \end{lemma}
% %
% \begin{proof}
%   By rule induction on the hypothesis, using strengthening
%   (\cref{lem:strengthening}) in the rule for type abstraction, substitution
%   (\cref{lem:substitution-context-split}) in the rule for type substitution,
%   agreement for type equivalence (\cref{lem:agreement-type-equiv}) in the
%   corresponding rule, and agreement for context split
%   (\cref{lem:agreement-context-split}) in the rules involving context split.
% \end{proof}

%%% TYPING PROCESSES

\begin{figure}[t!]
   \declrel{Process formation}{$\isproc\Gamma p$}
  \begin{gather*}
  	\infrule{\rulenameprocExp}
    	{
          \isExpr[\Empty]\Gamma{e}{T}\\
          \istype{\empty}T \kindtu
    	}
    	{\isproc {\Gamma} {\PROC e}}
    \qquad
    \infrule{\rulenameprocPar}
     {
      % \splitctx[\empty] \Gamma{\Gamma_1}{\Gamma_2}
      % \\ 
      \isproc{\Gamma_1}{p_1}
      \\ 
      \isproc{\Gamma_2}{p_2}
     }
     {\isproc{\Gamma_1 \circ \Gamma_2}{p_1 \mid p_2}}
    \\
    \infrule{\rulenameprocNew}
     {
      \isproc{\Gamma, x\colon T, y \colon {\dual T}}{p}\\
      \istype{\empty} T \kindsl\\
     }
     {\isproc{\Gamma}{\NU xyp}}
  \end{gather*}
  \caption{Process formation}
  \label{fig:typing-procs}
\end{figure}

%%% Local Variables:
%%% mode: latex
%%% TeX-master: "main"
%%% End:

Process formation is defined in \cref{fig:typing-procs}. Any
expression of unrestricted type can be made into a  process
({\rulenameprocExp}). Processes running in parallel split the resources
among them ({\rulenameprocPar}). Restriction introduces two channel ends
of session type, which are dual to one another ({\rulenameprocNew}).

%% DYNAMICS

\paragraph{Reduction}

\begin{figure}[t!]
  \declrel{Structural congruence}{$p \equiv q$}
  \begin{gather*}
    \axiom{}{p \equiv \PROC {\changed{\unite[\un]}} \PAR p}
    \qquad
    \axiom{}{p \PAR q \equiv q \PAR p}
    \qquad
    \axiom{}{(p \PAR  q) \PAR r \equiv p \PAR (q  \PAR r)}
    \\
    \axiom{}{(\NU xy p) \PAR q \equiv \NU xy (p \PAR q)}
    \qquad
    \axiom{}{\NU xy \PROC{\changed{\unite[\un]}} \equiv \PROC{\changed{\unite[\un]}}}
    \qquad
    \axiom{}{\NU xyp \equiv \NU yxp}
    \\
    \axiom{}{\NU wx \NU yzp \equiv \NU yz \NU wxp}
  \end{gather*}
  Rules for congruence and equivalence omitted.
  \caption{Structural congruence}
  \label{fig:scong}
\end{figure}
%%% Local Variables:
%%% mode: latex
%%% TeX-master: "main"
%%% End:

We consider processes up to structural congruence as defined in
\cref{fig:scong}. The rules are standard~\cite{DBLP:journals/iandc/Vasconcelos12}.
\Cref{fig:reduction} specifies the call-by-value, type-agnostic, reduction rules for expressions and
processes. Expression reduction $\isRed ee$ is standard for
call-by-value lambda calculus. Evaluation contexts $E$ specify a
left-to-right call-by-value reduction strategy. Thanks to the value
restriction, we need not evaluate under type abstractions.
Record components are evaluated
from left to right as they appear in the program text.

Process reduction $\isRed pp$ includes standard rules for congruence
({\rulenameprocredCong}), parallel execution ({\rulenameprocredPar}), and
reduction under restriction ({\rulenameprocredBind}). An expression process can
perform an internal reduction step ({\rulenameprocredExp}). Forking creates a
new process from the argument ({\rulenameprocredFork}). The operations on
channels are reasonably standard for session type calculi. Creation of a new
channel, \rulenameprocredNew, ignores the type argument, wraps the process in a
restriction, and returns the channel ends in a pair. Rule {\rulenameprocredMsg}
matches a send with a receive operation across threads. It transfers the value
ignoring the type arguments. Rule {\rulenameprocredCh} matches a select with
a suitable match operation and dispatches according to the label $k$ selected by
the internal choice process.

\begin{figure}[t!]
  \textit{Evaluation contexts}
  \begin{align*}
    && E \grmeq&
                   []
                   \grmor Ee
                   \grmor vE
                   \grmor \injecte \ell E
                   \grmor \tappe ET
                   \grmor \lete \ell xLEe 
    \\
    &&\grmor&
              \unlete {\unite}Ee
              \grmor \{ \ell_1= v_1, \dots, \ell_i=E, \dots, \ell_n=e_n\}
    \\
    &&\grmor& \casee E {\recordp \ell e L}
              \grmor \matche{E}{\recordp \ell e L}
  \end{align*}
  \declrel{Expression reduction}{$\isRed ee$}
  \begin{gather*}
    % FUNCTIONAL
    \axiom{\rulenameexpredApp}{(\abse x{\_}e)v \reduces \subs vxe}
    \qquad
    \axiom{\rulenameexpredLet}{
      \lete \ell xL{\recorde \ell vL}{e} \reduces \subs
      {v_\ell}{x_\ell}e_{\ell\in L}}
    \\
    \axiom{\rulenameexpredUnitElim}{\unlete {\unite}{\unite}e \reduces e}
    \qquad
    \axiom{\rulenameexpredTApp}{(\tappe{\tabse a \_ v)}{T} \reduces \subs Tav}
    \\
    \infrule{\rulenameexpredCase}
    {k\in L}
    {\casee{\injecte k v}{\recordp \ell e L}  \reduces e_kv}
    \qquad
    \infrule{\rulenameexpredCtx}
    {e \reduces e'}
    {E[e] \reduces E[e']}
  \end{gather*}
  \declrel{Process reduction}{$\isRed pp$}
  \begin{gather*}
    \infrule{\rulenameprocredExp}
    {e \reduces e'} 
    {\PROC{e} \reduces \PROC{e'}}
    \quad
    \axiom{\rulenameprocredFork}{\PROC{E[\appe{\tappe\forkk\_} e]} \reduces
      \PROC{E[\unite[\un]]} \PAR \PROC e}
    \quad
    \axiom{\rulenameprocredNew}{\PROC{E[\newe \_]} \reduces\NU xy\PROC{E[(x,y)]}}
    \\
    \axiom{\rulenameprocredMsg}{
      \NU xy(\PROC{E_1[\sendk[\_]v[\_]x]}
      \PAR \PROC{E_2[\receivek[\_][\_]y})
      % \NU xy(\PROC{E_1[\appe{\tappe{\appe{\tappe\sendk\_}v}\_}x]}
      % \PAR  \PROC{E_2[\appe{\tappe{\tappe\receivek\_}\_} y]})
      \reduces\NU xy(\PROC{E_1[x]} \PAR \PROC{E_2[(v,y)]})}
    \\
    \infrule{\rulenameprocredCh}
    {k\in L}
    {\NU xy(\PROC{E_1[\selecte k x]} \PAR \PROC{E_2[\matche
        y{\recordp \ell e L}]}) \reduces
      \hspace{10em}\\\hspace{23em}
      \NU xy(\PROC{E_1[x]} \PAR \PROC{E_2[e_ky]})
    }
    \\
    \infrule{\rulenameprocredPar}{p \reduces p'}{p\PAR q \reduces p'\PAR q}
    \qquad
    \infrule{\rulenameprocredBind}{p \reduces p'}{\NU xyp \reduces \NU xyp'}
    \qquad
    \infrule{\rulenameprocredCong}{p \equiv q\\ q \reduces q'}{p \reduces q'}
  \end{gather*}
  % not needed because the evaluation contexts do not bind any variables
  % Context $E_1$ (resp.~$E_2$, resp.~$E$) does not bind~$x$ (resp.~$y$,
  % resp.~$x$ and~$y$). 
  % \\
  % Dual $\NU xy$-rules for $\sendk/\receivek$ and $\selectk/\casek$
  % omitted. % Use structural congruence
  \caption{Reduction}
  \label{fig:reduction}
\end{figure}

%%% Local Variables:
%%% mode: latex
%%% TeX-master: "main"
%%% End:

\subsection{Type Safety}

Before proving the typing preservation for congruence
(\cref{lemma:congruence-preserves}), we need to establish some auxiliary results.

\begin{lemma}[Weakening]
  \label{lemma:unrestricted-weakening}
  \label{lem:weakening-typeEquiv}
  \label{lem:ctx-weakening}
  \label{lemma:weakening-ctx-split}
  \
  \begin{enumerate}
  \item If $\isequiv\Delta{T_1}{T_2}\kind$, then
    $\isequiv{\Delta, a:\kind'}{T_1}{T_2}\kind$.
  \item\label{it:weakening-ctx} If $\isCtx\Gamma\kind$, then $\isCtx[\Delta, a:\kind']\Gamma\kind$.
  \item\label{it:weakening-ctx-split} If
    \sloppy{$\istype{\Delta}{T}{\kindtl}$, then $\splitctx\Gamma{\Gamma_1}{\Gamma_2}$ if
    and only if
    $\splitctx{(\Gamma, x\colon T)}{(\Gamma_1, x\colon T)}{\Gamma_2}$.}
  \item\label{it:isExpr} If $\isExpr\Gamma e T$, then
    $\isExpr[\Delta, a:\kind]\Gamma eT$.
  \item\label{it:unweak} If $\isExpr\Gamma e T$ and
    $\istype{\Delta} U {\kindtu}$, then $\isExpr{\Gamma, x:U}e T$.
  \item If $\splitctx\Gamma{\Gamma_1}{\Gamma_2}$, then
    $\splitctx[\Delta, a:\kind]\Gamma{\Gamma_1}{\Gamma_2}$.
  \item If $\isproc \Gamma p$ and $\istype\Delta T\kindtu$, then $\isproc{\Gamma,x\colon T} p$. 
  \end{enumerate}
\end{lemma}
\begin{proof}
  \Cref{it:weakening-ctx-split}. Immediate from rule \rulenamesplitLeft.
  All the remaining items follow by rule induction on the hypothesis. In the
  case of~\cref{it:isExpr}, rules \rulenametypeConst, \rulenametypeVar,
  \rulenametypeAbs, \rulenametypeNew, \rulenametypeSel, and \rulenametypeEq
  use~\cref{it:weakening-ctx}.
\end{proof}

\begin{lemma}[Strengthening]\
  \label{lem:strengthening}
  \begin{enumerate}
  \item If $\isExpr {\Gamma,x\colon T} eU$ and $x\notin\free(T)$, then
    $\istype{\Delta}{T}{\kindtu}$ and $\isExpr\Gamma eU$.
  \item If $\isproc{\Gamma,x\colon T}{p}$ and $x\notin\free(p)$, then
  $\istype{\empty}{T}{\kindtu}$ and $\isproc{\Gamma}{p}$.

  \end{enumerate}
  % \begin{enumerate}
  % \item If $x\colon T\in\Gamma$, then $\istype{\empty}{T}{\kindtu}$.
  % \item If $\Gamma = \Gamma', x\colon T$, then $\isproc{\Gamma'}{p}$.
  % \end{enumerate}
\end{lemma}
\begin{proof}
  By rule induction on the first hypothesis.
\end{proof}

\begin{lemma}[Agreement for process formation]\
  \begin{enumerate}
  \item\label{item:agree-exps}
    If $\isExpr{\Gamma}{e}{T}$, then $\isCtx{\Gamma}{\kind}$ and
    $\istype{\Delta}{T}{\kind'}$. % , for some kinds $\kind,\kind'$.
  \item\label{item:agree-procs} If $\isproc \Gamma p$, then
    $\isCtx[\empty]{\Gamma} \kind$. % , for some $\kind$.
  \end{enumerate}
\end{lemma}
\begin{proof}
  \begin{sloppypar}
    \Cref{item:agree-exps}. By rule induction on the hypothesis, using
    strengthening (\cref{lem:strengthening}) in the rule for type abstraction,
    substitution (\cref{lem:substitution-context-split}) in the rule for type
    substitution, agreement for type equivalence
    (\cref{lem:agreement-type-equiv}) in the corresponding rule, and agreement
    for context split (\cref{lem:agreement-context-split}) in the rules
    involving context split.
  \end{sloppypar}
  \Cref{item:agree-procs}. By rule induction on the hypothesis,
  using~\cref{item:agree-exps} on \rulenameprocExp and the agreement for context
  split (\cref{lem:agreement-context-split}) on the rule for parallel
  composition.
\end{proof}

\begin{theorem}[Congruence preserves typing]
  \label{lemma:congruence-preserves}
  If $p \equiv q$ and $\isproc\Gamma p$, then $\isproc\Gamma q$.
\end{theorem}
\begin{proof}
  % By rule induction on $p \equiv q$. The axioms in \cref{fig:scong} are
  % immediate \changed{using~\cref{lemma:weakening-ctx-split}}. 
  % Reflexivity and transitivity are trivial. The congruence rules
  % follow directly by induction. \todo{This section opens by
  %   saying ``before the proof of this theorem we need some technical results'';
  %   then we do not refer to them.}
  The proof is by analysis of the derivation for each member of each axiom
  in~\cref{fig:scong}. We use the properties of context split
  (\cref{lem:cSplitProps}), weakening (\cref{lemma:unrestricted-weakening}),
  strengthening (\cref{lem:strengthening}), and check both directions of each
  axiom.

  We detail the case for scope restriction. In the forward direction of the
  axiom we have to show that if $\isproc{\Gamma}{\NU xyp \PAR q}$ then
  $\isproc{\Gamma}{\NU xy (p \PAR q)}$. We start by building the only derivation
  for $\isproc{\Gamma}{\NU xyp \PAR q}$ and obtain:
  $\splitctx[\empty] \Gamma{\Gamma_1}{\Gamma_2}$ and
  $\isproc{\Gamma_1,x\colon T, y\colon\dual T}p$ and $\istype{\empty} T \kindsl$
  and $\isproc{\Gamma_2}{q}$. To build the derivation for the conclusion, we
  distinguish two subcases based on the linearity of $T$. If $T$ is linear, then
  we have
  $\Gamma_1, x\colon T, y\colon\dual T \circ \Gamma_2 = (\Gamma_1\circ\Gamma_2),
  x\colon T, y\colon\dual T$. If $T$ is unrestricted, we use weakening and
  obtain $\isproc{\Gamma_2,x\colon T,y\colon\dual T}q$ and thus,
  $(\Gamma_1\circ\Gamma_2),x\colon T,y\colon\dual T = (\Gamma_1,x\colon
  T,y\colon\dual T)\circ(\Gamma_2,x\colon T,y\colon\dual T)$. Conclude with
  \rulenameprocPar and \rulenameprocNew.

  \sloppy In the reverse direction we show that if
  $\isproc{\Gamma}{\NU xy (p \PAR q)}$ then $\isproc{\Gamma}{\NU xyp \PAR q}$.
  We start by building the only derivation for
  $\isproc{\Gamma}{\NU xy (p \PAR q)}$ and obtain
  $\Gamma,x\colon T,y\colon\dual T = \Gamma_1\circ\Gamma_2$ and
  $\isproc{\Gamma_1}p$, $\isproc{\Gamma_2}q$ and $\istype\empty T\kindsl$.
  Again, we distinguish two subcases. If $T$ is linear, the properties of the
  context split (\cref{lem:cSplitProps}) state that $T$ is either in $\Gamma_1$
  or in $\Gamma_2$, but not in both. Given that $x\not\in\free(q)$,
  strengthening gives us that $x\colon T \not\in\Gamma_2$,
  $x\colon T \in\Gamma_1$, and thus
  $\Gamma_1 = \Gamma_1',x\colon T,y\colon\dual T$. If $T$ is unrestricted, then
  $\Gamma,x\colon T,y\colon\dual T = (\Gamma_1,x\colon T,y\colon\dual
  T)\circ(\Gamma_2,x\colon T,y\colon\dual T)$. Applying strengthening on
  $\isproc{\Gamma_2,x\colon T,y\colon\dual T}q$ yields $\isproc{\Gamma_2}q$. In both cases, conclude with
  \rulenameprocNew and \rulenameprocPar.
\end{proof}

%In the reverse direction, to show that if Γ ⊢ (νxy)(P | Q )
% then Γ ⊢ (νxy)P | Q , we start by building the only derivation
% forΓ⊢(νxy)(P|Q)toconcludethatΓ,x:T,y:T=Γ1◦Γ2,Γ1⊢P
% andΓ2⊢Q.Tobuildaderivationfortheconclusionwe distinguish two cases. If T is
% linear, then x : T is either in Γ1 or in Γ2 , but not in both (properties of
% context split). Given that x∈/fv(Q), strengthening gives us that x:T ∈/Γ2, hence
% x:T ∈Γ1, and similarly for y and T. Hence Γ1 =Γ1′,x:T,y:T. If, on the other hand
% T is unrestricted, we know that Γ1 =Γ1′,x:T,y:T and Γ2 =Γ2′,x:T,y:T, and we
% apply strengthening to obtain Γ2′ ⊢ Q . In either case we conclude the proof
% with rule [T-Res] and [T-Par].

\begin{lemma}[Normalised type derivation]
  \label{lem:normalized-type-derivation}
  Suppose that $\isExpr\Gamma e T$ is derivable.  Then there is
  a derivation of the same judgment that ends with exactly one use of
  the rule {\rulenametypeEq}. 
\end{lemma}
\begin{proof}
  By induction on the number of uses of {\rulenametypeEq} that conclude the derivation of
  $\isExpr\Gamma e T$.
  If the derivation does not end with {\rulenametypeEq}, then we may add one
  because type equivalence is reflexive.
  If the derivation ends with more than one application of
  {\rulenametypeEq}, then we may combine them to one rule application
  because type equivalence is transitive.
  % , i.e., $\typeequiv {T''}{T'}$
  % and $\typeequiv{T'} T$ implies $\typeequiv {T''} T$.
\end{proof}

\begin{lemma}[Inversion for expression formation]
  \label{lem:inversion}
  Let $\isExpr\Gamma e T$.
  \begin{enumerate}
    
  % \item T-Const. No need?
    
  \item If $e=x$, then $\Gamma = \Gamma',x\colon U$ and
    $\isCtx{\Gamma'}{\kindtu}$ and
    $\isequiv[\Empty]{\Delta}{U}{T}{\kind}$.

  \item\label{itm:inv-abs} If $e = \abse x U e$, then
    $\isCtx\Gamma {\kindt^m}$ and $\isExpr{\Gamma, x\colon U} e {V}$ and
    $\isequiv[\Empty] \Delta {\funt UV} T \kindt^m$.

  \item\label{itm:inv-app} If $e = \appe{e_1}{e_2}$, then
    $\splitctx{\Gamma}{\Gamma_1}{\Gamma_2}$ and
    $\isExpr{\Gamma_1}{e_1}{\funt U{V}}$ and
    $\isExpr{\Gamma_2}{e_2}U$ and
    $\isequiv[\Empty] \Delta V T \kind$.

  \item If $e = \tabse{a}{\kind}{v}$, then
    $\isExpr[\Delta, a\colon\kind]\Gamma vU$ and
    $\isequiv[\Empty] \Delta {\forallt a \kind U} T \kindt^m$.

  \item If $e = \tappe eU$, then $\isExpr\Gamma{e}{\forallt a \kind {V}}$ and
    $\istype \Delta U \kind$ and $\isequiv[\Empty] \Delta {V[U/a]} T\kind$.
    
  \item If $e = \recorde \ell eL$, then $\splitctx{\Gamma}{}{\Gamma_\ell}$ and
    $\isExpr{\Gamma_\ell}{e_\ell}{T_\ell}$ and
    $\isequiv[\Empty]\Delta {\recordt \ell TL} T \kindt^m$.
    
  \item If $e = \lete \ell xL{e_1}{e_2}$, then
    $\splitctx{\Gamma}{\Gamma_1}{\Gamma_2}$ and
    $\isExpr{\Gamma_1} {e_1} {\recordt \ell TL}$ and
    $\isExpr{{\Gamma_2}, (x_\ell\colon T_\ell)_{\ell\in L}} {e_2}{U}$ and
    $\isequiv[\Empty]\Delta UT\kind$.

  \item If $e = \unlete\unite {e_1}{e_2}$, then
    $\splitctx{\Gamma}{\Gamma_1}{\Gamma_2}$ and $\isExpr{\Gamma_1}{e_1}\unitk$
    and $\isExpr{\Gamma_2}{e_2} U$ and $\isequiv[\Empty]\Delta UT\kind$.

  \item If $e = \injecte k e$, then exist $\{T_\ell\}_{\ell\in L}$ such that
    $\isExpr\Gamma e {T_k}$ for $k\in L$ and
    $\istype{\Delta}{T_\ell}{\kindtl}$ and
    $\isequiv[\Empty]\Delta {\variantt \ell{T}L} T \kindt^m$.

  \item \sloppy{If $e = \casee{e}{\recordp \ell e L}$, then
    $\splitctx{\Gamma}{\Gamma_1}{\Gamma_2}$ and
    $\isExpr{\Gamma_1} e {\variantt \ell TL}$ and
    $\isExpr{\Gamma_2}{e_\ell}{\funt[\lin]{T_\ell}{U}}$ and
    $\isequiv[\Empty]\Delta UT\kind$.}

  % \item T-New. No need?

  \item If $e = \selecte k{}$, then exist $\{T_\ell\}_{\ell\in L}$ such that 
    $\istype{\Delta}{T_\ell}{\kindsl}$ and
    $\isequiv [\Empty]\Delta {\funt{\ichoicet{\recordt \ell T L}}{T_k}} T
    \kindt^m$ for $k \in L$ and $\isCUn \Gamma$.

  % \item T-Match. No need?
  \end{enumerate}
\end{lemma}
\begin{proof}
  All cases are similar; we detail
  %
  % \Cref{itm:inv-abs}: \todo{Use agreement to get the kind of $V$.}
  %
  \cref{itm:inv-app}.
  By \cref{lem:normalized-type-derivation}, there is a derivation that ends with
  a single use of \rulenametypeEq. Its premises read
  $\isExpr\Gamma{\appe{e_1}{e_2}}{V}$ and $\isequiv[\Empty]\Delta VT\kind$.
  The premises to rule \rulenametypeApp on $\isExpr\Gamma{\appe{e_1}{e_2}}{V}$
  include $\splitctx{\Gamma}{\Gamma_1}{\Gamma_2}$ and
  $\isExpr{\Gamma_1}{e_1}{\funt U {V}}$ and $\isExpr{\Gamma_2}{e_2}U$ which
  proves the claim.
\end{proof}

The next lemma is used to establish type preservation for the communication
reductions. Since such rules move values from one process to another, the values
need to take a part of its typing context with them.

\begin{lemma}[Context typing]
  \label{lem:context-typing}
  Let $\splitctx[\Delta]\Gamma{\Gamma_1}{\Gamma_2}$. Then,
  $\isExpr[\Delta]\Gamma{E[e]}T$ if and only if $\isExpr[\Delta]{\Gamma_1}{e}U$
  and $\isExpr[\Delta]{\Gamma_2, x\colon U}{E[x]}T$.
\end{lemma}
\begin{proof}
  The forward implication is by induction on the structure of evaluation context
  $E$. We sketch one case: $Ee'$, in which case $(Ee')[e] = (E[e]) e'$.
  Inversion of the typing relation (\cref{lem:inversion}) yields
  (1) $\isequiv[\Empty] \Delta W T \kind$ and 
  (2) $\isExpr{\Gamma_1}{E[e]}{\funt VW}$ and
  (3) $\isExpr{\Gamma_2}{e'}V$.
  From (2) and induction we get
  (4) $\isExpr{\Gamma_{11},x\colon U}{E[x]}{\funt VW}$ and
  (5) $\isExpr{\Gamma_{12}}{e}{U}$ and
  (6) $\splitctx[\Delta]{\Gamma_1}{\Gamma_{11}}{\Gamma_{12}}$.
  From the hypothesis $\splitctx[\Delta]\Gamma{\Gamma_1}{\Gamma_2}$ and (6),
  properties of context split (\cref{lem:cSplitProps}) guarantee that
  $\Gamma_{11}\circ\Gamma_2$ is defined. Then, from (3), (4) and rule
  \rulenametypeApp we have
  (7) $\isExpr{\Gamma_{11}\circ\Gamma_2,x\colon U}{(E[x])e'}{W}$.
  Given that $(E[x])e' = (Ee')[x]$, from (1), (7) and rule \rulenametypeEq we
  have $\isExpr{\Gamma_{11}\circ\Gamma_2,x\colon U}{(Ee')[x]}{T}$, which
  concludes the case together with item (5).

  The reverse direction is a simple application of the appropriate typing rule.
  For example, for context $Ee'$ we use rule \rulenametypeApp.
\end{proof}

\begin{lemma}[Substitution preserves equivalence]
  \label{lem:type-equivalence-substitution}
  \begin{sloppypar}
    If $\isequiv{\Delta, a:\kind}{T_1}{T_2}{\kind'}$ and
    $\istype\Delta U \kind$, then
    $\isequiv \Delta{\subs Ua{T_1}}{\subs Ua{T_2}}{\kind'}$.
  \end{sloppypar}
\end{lemma}
\begin{proof}
  By rule induction on the first hypothesis. The key step is the case for rule
  {\rulenameeqVar} when the type variable considered is $a$. Here we need to
  appeal to reflexivity of type equivalence to establish
  $\isequiv\Delta UU\kind$ (\cref{lem:type-equivalence}).
\end{proof}

\begin{lemma}[Type substitution]
  \label{lem:type-substitution-exps}
  If $\isExpr[\Delta,a\colon\kind]{\Gamma}{e}{T}$ and $\istype \Delta U \kind$,
  then $\isExpr{\subs Ua\Gamma}{\subs Uae}{\subs UaT}$.
\end{lemma}
\begin{proof}
  The proof is by rule induction on the first hypothesis. We rely on
  \cref{lem:type-substitution} to handle premises on type and typing context
  formation. \Cref{lem:substitution-context-split} handles the interaction with
  context splitting. \Cref{lem:type-equivalence-substitution} handles the case
  for {\rulenametypeEq}. We detail one case: rule \rulenametypeVar with $e=x$.
  The premises to the rule are $\isCUn[\Delta,a\colon\kind]{\Gamma}$ and
  $\istype{\Delta,a\colon\kind}{T}{\kind'}$
  Noticing that $\subs Uax=x$, the result follows by applying type
  substitution (\cref{lem:type-substitution}) to the second hypothesis and to
  the premises to the rule.
\end{proof}

\begin{lemma}[Value substitution]
  \label{lem:value-substitution}
  If $\isExpr{\Gamma_1, x:T}eU$ and $\splitctx\Gamma{ \Gamma_1}{ \Gamma_2}$ and
  $\isExpr{\Gamma_2}vT$, then $\isExpr\Gamma{e[v/x]}U$.
\end{lemma}
\begin{proof}
  The proof is by rule induction on the first hypothesis. In some cases
  (variables, application and abstraction), we have to distinguish whether $T$
  (or the type of the abstracted variable) is linear or not. Unrestricted
  weakening (\cref{lemma:unrestricted-weakening}) is required in the case for
  abstraction.
\end{proof}

\begin{theorem}[Typing preservation for expression reduction]
  \label{thm:preservation-exps}
  If $\isRed{e_1}{e_2}$ and $\isExpr{\Gamma}{e_1}{T}$, then
  $\isExpr{\Gamma}{e_2}{T}$.
\end{theorem}
\begin{proof}
  By rule induction on $\isRed{e_1}{e_2}$.

  Case \rulenameexpredTApp. We apply inversion (\cref{lem:inversion}) to the
  typing judgement for the redex $\isExpr{\Gamma}{\tappe{(\tabse{a}{\kind}{v})}{T'}}{T}$ 
  and obtain 
  \begin{gather}
    \label{eq:tpe1}
    \isExpr{\Gamma}{(\tabse{a}{\kind}{v})}{\forallt{a}{\kind}{U}}
    \\
    \label{eq:tpe2}
    \istype{\Delta}{T'}{\kind}
    \\
    \label{eq:tpe2a}
    \isequiv[\Empty]\Delta T {\subs {T'}aU} \kind
  \end{gather}
  Inversion of \cref{eq:tpe1} (\cref{lem:inversion}) yields:
  \begin{gather}
    \label{eq:tpe3}
    \isExpr[\Delta,a\colon\kind]{\Gamma}{v}{U'}
    \\
    \isequiv[\Empty]{\Delta}{\forallt a \kind {U'}}{\forallt a \kind U}{\kind'}
  \end{gather}
  % Since type equivalence is a congruence
  By \rulenameeqPoly
  \begin{gather}
    \label{eq:tpe4}
    \isequiv[\Empty]{\Delta, a:\kind}{ {U'}}{ U}{\kind'}
  \end{gather}
  Applying rule \rulenametypeEq to \cref{eq:tpe3,eq:tpe4} yields
  \begin{gather}
    \label{eq:tpe5}
    \isExpr[\Delta,a\colon\kind]{\Gamma}{v}{U}
  \end{gather}
  Applying the substitution lemma (\cref{lem:type-substitution-exps}) to
  \cref{eq:tpe2,eq:tpe5} yields
  \begin{gather}
    \isExpr\Gamma{\subs {T'} av}{\subs{T'}aU}
  \end{gather}
  noticing that
  $a\notin\free(\Gamma)$, hence $\subs {T'}a\Gamma = \Gamma$.
  Conclude by applying \rulenametypeEq using \cref{eq:tpe2a}.

  Case \rulenameexpredApp. We apply inversion
  (\cref{lem:inversion}) to the typing judgement for the redex
  $\isExpr\Gamma{\appe{(\abse x U e)}v}T$ and obtain
  \begin{gather}\label{eq:21}
    \isExpr\Gamma{\abse x U e}{\funt U{U'}} \\\label{eq:23}
    \isExpr\Gamma {v} U \\\label{eq:24}
    \isequiv[\Empty]\Delta T{U'} \kind
  \end{gather}
  By inversion (\cref{lem:inversion}) applied to \eqref{eq:21} we
  obtain
  \begin{gather}\label{eq:22}
    \isExpr{\Gamma, x:U} e {U''} \\\label{eq:25}
    \isequiv[\Empty]\Delta {U'}{U''}{\kind}
  \end{gather}
  Applying substitution (\cref{lem:value-substitution}) to
  \eqref{eq:22} and \eqref{eq:23} we have
  \begin{gather}
    \isExpr\Gamma {\subs vxe} {U''}
  \end{gather}
  Using transitivity on \eqref{eq:24} and \eqref{eq:25}, we apply
  {\rulenametypeEq} to obtain
  \begin{gather}
    \isExpr\Gamma {\subs vxe} {T}
  \end{gather}

  Case \rulenameexpredLet. We apply inversion
  (\cref{lem:inversion}) to the typing judgement for the redex
  $\isExpr\Gamma{\lete \ell xL{\recorde \ell vL}{e}}T$ and obtain
  \begin{gather}\label{eq:26}
    \splitctx{\Gamma}{\Gamma_1}{\Gamma_2}\\
    \isExpr{\Gamma_1}{\recorde \ell vL}{\recordt \ell T L}
    \\\label{eq:27}
    \isExpr{\Gamma_2, (x_\ell:T_\ell)_{\ell\in L}} e {T'}
    \\
    \isequiv[\Empty]\Delta T{T'}\kind
  \end{gather}
  Inversion (\cref{lem:inversion}) and \rulenameeqRcd applied to \eqref{eq:27} yields
  \begin{gather}
    \splitctx{\Gamma_1}{}{\Gamma_\ell}
    \\
    \isExpr{\Gamma_\ell}{v_\ell}{T_\ell'}
    \\
    \isequiv[\Empty]\Delta {T_\ell}{T_\ell'}\kind
  \end{gather}
  By {\rulenametypeEq} we have $\isExpr{\Gamma_\ell}{v_\ell}{T_\ell}$, for
  each $\ell$, by substitution (\cref{lem:value-substitution}) on~\eqref{eq:27}
  we
  have $\isExpr{\Gamma_2\circ \Gamma_1}{\subs {v_\ell}{x_\ell}e_{\ell\in L}}{T'}$, and
  by  {\rulenametypeEq} its type is also $T$. 

  Case \rulenameexpredCase. We apply inversion
  (\cref{lem:inversion}) to the typing judgement for the redex
  \begin{gather}
    \isExpr\Gamma {\casee{\injecte k v}{\recordp \ell e L}} T
  \end{gather}
  (where $k\in L$) and obtain some $T'$ such that
  \begin{gather}\label{eq:30}
    \splitctx{\Gamma}{\Gamma_1}{\Gamma_2}\\
    \isExpr{\Gamma_1}{\injecte k v}{\variantt \ell T L}
    \\
    \isExpr{\Gamma_2}{e_\ell}{\funt {T_\ell}{T'}}
    \\
    \isequiv[\Empty]\Delta T {T'}\kind
  \end{gather}
  By inversion (\cref{lem:inversion}) on \eqref{eq:30}, exist $\variantt\ell T L$ such that
  \begin{gather}
    \isExpr{\Gamma_1} v {T'}
    \\
    \isequiv[\Empty]\Delta{T_k}{T'}\kind
  \end{gather}
  Putting things back together: By \rulenametypeEq, $\isExpr{\Gamma_1} v{T_k}$. Hence, we can apply
  $e_k$ using {\rulenametypeApp}
  \begin{gather}
    \isExpr\Gamma{e_kv}{T'}
  \end{gather}
  and conclude with    {\rulenametypeEq} to obtain type $T$.

  Case \rulenameexpredCtx. From context typing (\cref{lem:context-typing})
  we know that $\isExpr[\Delta]{\Gamma_1}{e_1}U$ and
  $\isExpr[\Delta]{\Gamma_2, x\colon U}{E[x]}T$, with
  $\splitctx[\Delta]\Gamma{\Gamma_1}{\Gamma_2}$. By induction
  $\isExpr[\Delta]{\Gamma_1}{e_2}U$. We conclude with context typing, now in the
  reverse direction.
  % OLD (STRANGE) PROOF
  % We invert
  % ({\rulenameprocExp}) to obtain
  % \begin{gather}
  %   \isExpr[\Empty]{\Gamma}{E[e]} T
  % \end{gather}
  % By \cref{lem:context-typing}, we can pick some
  % $S$ and $z: S$ and $\splitctx[\Empty]\Gamma{ \Gamma_1}{ \Gamma_2}$ such that
  % \begin{gather}
  %   \isExpr[\Empty]{\Gamma_1, z:S}{E[z]} T \\
  %   \isExpr[\Empty]{\Gamma_2}{ e} {S}
  % \end{gather}
  % By \cref{thm:preservation-exps}
  % \begin{gather}
  %   \isExpr[\Empty]{\Gamma_2}{ e'} {S}
  % \end{gather}
  % Applying substitution yields
  % \begin{gather}
  %   \isExpr[\Empty]{\Gamma}{E[e']} T
  % \end{gather}
  % and we conclude by applying \rulenameprocredExp.
\end{proof}

\begin{theorem}[Type preservation for process reduction]
  \label{thm:preservation-procs}
  If $\isRed pq$ and $\isproc\Gamma p$, then $\isproc\Gamma q$.
\end{theorem}
\begin{proof}
  We proceed by rule induction on $\isRed pq$.

  Case \rulenameprocredExp. We invert \rulenameprocExp to obtain
  $\isExpr[\Empty]{\Gamma}{e} T$. Preservation for expression reduction
  (\cref{thm:preservation-exps}) gives
  $\isExpr[\Empty]{\Gamma}{e'} T$. We conclude with rule \rulenameprocExp.
  
  Case \rulenameprocredFork. We invert \rulenameprocExp to obtain
  \begin{gather}
    \isExpr[\Empty]{\Gamma}{E[\appe{\tappe\forkk\_} e]} T
  \end{gather}
  By \cref{lem:context-typing}, we can pick some $z:
  \unitt$ and $\splitctx[\Empty]\Gamma{ \Gamma_1}{ \Gamma_2}$ such that
  \begin{gather}
    \isExpr[\Empty]{\Gamma_1, z:\unitt}{E[z]} T \\
    \isExpr[\Empty]{\Gamma_2}{\appe{\tappe\forkk\_} e} {\unitt}
  \end{gather}
  By substitution and by inversion (for some $\istype \Empty {T'}\kindtu$)
  \begin{gather}
    \isExpr[\Empty]{\Gamma_1}{E[\unite]} T  \\
    \isExpr[\Empty]{\Gamma_2}{ e} {T'}
  \end{gather}
  By rules {\rulenameprocExp} and {\rulenameprocPar}
  \begin{gather}
    \isproc\Gamma{\PROC {E[\unite]} \mid \PROC e}
  \end{gather}
  
  Case \rulenameprocredNew. We invert \rulenameprocExp to obtain
  \begin{gather}
    \isExpr[\Empty]{\Gamma}{E[\newe S]} T
  \end{gather}
  By \cref{lem:context-typing}, we can pick some $z:
  \pairt S {\dual S}$ such that
  \begin{gather}
    \isExpr[\Empty]{\Gamma, z: \pairt S{\dual S}}{E[z]} T \\
    \isExpr[\Empty]{\Empty}{\newe S} {\pairt S {\dual S}}
  \end{gather}
  We create a new typed value from two channel ends by
  \begin{gather}
    \isExpr[\Empty]{x:S, y:\dual S}{(x, y)}{\pairt S {\dual S}}
  \end{gather}
  By substitution of this value for $z$, we obtain
  \begin{gather}
    \isExpr[\Empty]{\Gamma, x: S, y:\dual S}{E[(x, y)]} T
  \end{gather}
  By {\rulenameprocExp} and {\rulenameprocNew}, we find:
  \begin{gather}
    \isproc\Gamma {\PROC {E[(x, y)]}}(\nu xy)
  \end{gather}
  % We discuss one illustrative case for rule {\rulenameprocNew}
  % concluding with ${\isproc{\Gamma}{\NU xyp}}$.
  
  Case \rulenameprocredMsg. We invert
  {\rulenameprocNew} to obtain
  \begin{gather}
    \label{eq:1}
    \isproc{\Gamma, x\colon T, y \colon \dual
      T}{(\PROC{E_1[\appe{\tappe{\appe{\tappe\sendk\_}v}\_}x]}  \PAR  \PROC{E_2[\appe{\tappe{\tappe\receivek\_}\_} y]})}\\
    \label{eq:9}
    \istype{}{T}{\kindsl}
  \end{gather}
  We invert \eqref{eq:1} using {\rulenameprocPar} to obtain
  \begin{gather}
    \Gamma = \Gamma_1 \circ \Gamma_2 \\
    \label{eq:2}
    \isExpr[\Empty]{\Gamma_1, x\colon T}{E_1[\appe{\tappe{\appe{\tappe\sendk\_}v}\_}x]}\unitt \\
    \label{eq:3}
    \isExpr[\Empty]{\Gamma_2, y\colon \dual T}{E_2[\appe{\tappe{\tappe\receivek\_}\_} y]}\unitt
  \end{gather}
  Applying \cref{lem:context-typing} to judgement \eqref{eq:2}
  yields some $\Gamma_1'', \Gamma_2'', S_1$ such that, for some type
  $R$ and all $z_1$ not in $\Gamma_2''$,
  \begin{gather}
    \Gamma_1, x:T = \Gamma_1'' \circ (\Gamma_2'', x:T) \\
    \isExpr[\Empty]{\Gamma_2'', x:T}{\appe{\tappe{\appe{\tappe\sendk\_}v}\_}x}{S_1} \\
    \isExpr[\Empty]{\Gamma_1'', z_1: S_1}{E_1[z_1]}\unitt
  \end{gather}
  Inversion of the sending application yields
  \begin{gather}\label{eq:10}
    T = \sout R; S_1 \\
    \label{eq:4}
    \isExpr[\Empty]{\Gamma_2''}v R
  \end{gather}
  Applying \cref{lem:context-typing} to \eqref{eq:3}
  yields some $\Gamma_1''', \Gamma_2''', S_2$ such that, for all $z_3$
  not in $\Gamma_2'''$,
  \begin{gather}
    \Gamma_2, y:\dual T = \Gamma_1''' \circ (\Gamma_2''', y:\dual T) \\
    \dual T = \sint R; \dual{S_1} \text{ and }
    S_2 = \pairt R {\dual {S_1}}\\
    \isExpr[\Empty]{\Gamma_2''', y:\dual T}{\appe{\tappe{\tappe\receivek\_}\_} y}{S_2} \\
    \label{eq:5}
    \isExpr[\Empty]{\Gamma_1''', z_2:S_2}{E_2[z_2]}\unitt
  \end{gather}
  To conclude, take $z_1 = x$ to obtain
  \begin{gather}
    \label{eq:7}
    \isExpr[\Empty]{\Gamma_1'', x: S_1}{E_1[x]}\unitt
  \end{gather}
  From \eqref{eq:4}, we obtain
  \begin{gather}
    \label{eq:6}
    \isExpr[\Empty]{\Gamma_2'', y:\dual {S_1}}{(v, y)}{\pairt R {\dual{S_1}}}
  \end{gather}
  Apply substitution (\cref{lem:value-substitution}) for $z_2$ to
  \eqref{eq:5} and \eqref{eq:6} yielding
  \begin{gather}
    \label{eq:8}
    \isExpr[\Empty]{\Gamma_1''' \circ \Gamma_2'', y : \dual{S_1}}{E_2[(v,y)]}\unitt
  \end{gather}
  Apply {\rulenameprocPar} to \eqref{eq:7} and \eqref{eq:8},
  exploiting that $\Gamma = \Gamma_1''\circ\Gamma_2'' \circ
  \Gamma_1''' \circ \Gamma_2'''$
  \begin{gather}\label{eq:12}
    \isproc{\Gamma, x:S_1, y:\dual{S_1}}{{(\PROC{E_1[x]} \PAR  \PROC{E_2[(v, y)]})}}
  \end{gather}
  By \eqref{eq:9} and \eqref{eq:10} and inversion of
  {\rulenamekindSeq}, we know that
  \begin{gather}\label{eq:11}
    \istype{\Empty}{S_1}{\kindsl}
  \end{gather}
  Conclude by applying {\rulenameprocNew} to \eqref{eq:11} and \eqref{eq:12}.
  
  Case \rulenameprocredCh. We proceed analogously.

  Case \rulenameprocredPar. Conclude by induction.
  
  Case \rulenameprocredBind. Conclude by induction.

  Case \rulenameprocredCong. Apply \cref{lemma:congruence-preserves} and
  conclude by induction.
\end{proof}

\paragraph{Absence of Run-time Errors}

% RUNTIME ERRORS
We start with the definition of runtime errors. The \emph{subject} of an
expression~$e$, denoted by $\subj(e)$, is~$z$ in the following cases and
undefined in all other cases. 
\begin{equation*}
  \sendk[T]v [U] z \qquad
  \receivek[T][U] z \qquad
  \selecte \ell z \qquad
  \matche z{\recordpa \ell e}
\end{equation*}

Two expressions~$e_1$ and $e_2$ \emph{agree} on channel
$xy$, notation $\agree^{xy}(e_1,e_2)$, in the following four cases.
\begin{itemize}
\item $\agree^{xy}(\sendk[T]v [U] x, \receivek[T][\dual U] y)$;
\item $\agree^{xy}(\receivek[T][U] x,\sendk[T]v [\dual U] y)$;
\item $\agree^{xy}(\selecte kx, \matche y{\recordp \ell e L})$ and $k\in L$;
\item $\agree^{xy}(\matche x{\recordp \ell e L, \selecte ky})$ and $k\in L$.
\end{itemize}

A closed process is an \emph{error} if it is structurally congruent to some
process that contains a subprocess of one of the following forms.
\begin{enumerate}
\item $\PROC{E[\appe v e]}$ where $v$ is
  $x, c, \tabse a\kind {v'}, \recorde \ell v L, \injecte \ell {v'},
  \sendk[T]v'$ or $\receivek[T]$;
% not $\lambda$, $\forkk[T]$, $\sendk[T]$,
%   $\sendk[T]v'[U]$, $\receivek[T][U]$ or $\selecte\ell{}$;
\item $\PROC{E[\tappe vT]}$ where $v$ is
  $x$, $ \abse xTe$, $ \recorde \ell v L$, $ \injecte \ell {v'}$, $ \sendk[T]$, $
  \sendk[T]v'[U]$ or $\receivek[T][U]$. % or $\forkk[T]$;
% not $\Lambda$$, $ $\forkk$, $\sendk$, $\sendk[T]v'$,
%   $\receivek$, $\receivek[T]$;
\item $\PROC{E[\lete \ell{x}L v e]}$ and  $v \neq \recorde\ell{v}L$;
  % , for all
  % values $v_\ell$;
  % $ $v \ne v'$ (a record with fields matching $L$);
\item $\PROC{E[\matche v{\recordp \ell e L}]}$ and $v \neq \injecte k{v'}$, for
  all $k\in L$;
\item\label{item:1} $\PROC{E_1[e_1]} \PAR \PROC{ E_2[e_2]}$ and
  $\subj(e_1) = \subj(e_2)$;
\item\label{item:2} $\NU xy(\PROC{E_1[e_1]} \PAR  \PROC{E_2[e_2]} \PAR p)$ and $\subj(e_1)=x$
  and $\subj(e_2)=y$ and $\neg\agree^{xy}(e_1,e_2)$.
% \item $  (\new a,b)(E_1[\selectk\,l_j\,a] \PAR
%   E_2[\casek\,b\,\ofk\,\{l_i\rightarrow e_i\}_{i\in I}] \PAR p) $
%   where $E_1$ does not bind $a$ and $E_2$ does not bind $b$ and
%   $j\notin I$ (or the same with $a$ and $b$ exchanged).
\end{enumerate}

The first four cases are typical of polymorphic functional
languages with records, variants, and functional constants.
\Cref{item:1} guarantees that no two threads hold references to the same channel
end. If $\PROC{E_1[e_1]} \PAR \PROC{ E_2[e_2]}$ is closed for non channel-end
variables (and given that the evaluation contexts $E_1$ and $E_2$ do not bind
variables), then the subjects of $e_1$ and $e_2$ are channel ends.
\Cref{item:2} says that channel ends agree at all times: if one
thread is ready for sending, then the other is ready for receiving,
and similarly for selection and branching.

\begin{theorem}[Absence of run-time errors]
        If $\isproc{}p$, then $p$ is not an error.
\end{theorem}
\begin{proof}
  Suppose that $\isproc{}p$ and that $p$ is an error. Then we know that for each
  subprocess $q$ of $p$ there exists $\Gamma$ such that $\isproc{\Gamma}{q}$
  with $\isCtx[]{\Gamma}{\kinds^m}$, since all entries in $\Gamma$ are
  introduced by rule \rulenameprocredNew.
  % The premise to rule \rulenameprocExp reads $\isExpr[\Empty]{\Gamma}{e}{T}$
  %
  % Furthermore, if $q$ is of the form
  % $\PROC{e}$, then by \cref{lem:context-typing} we also know that
  % $\isExpr[\Empty]{\Gamma_1}{e}{T}$ with
  % $\splitctx[\Empty]\Gamma{\Gamma_1}{\Gamma_2}$.
  By \cref{lemma:congruence-preserves}, we may assume that $q$ is one of the
  offending subprocesses. We analise in turn each of the six cases in the
  definition of runtime errors.

%  \begin{enumerate}
 1. Suppose $\isExpr[\Empty]{\Gamma_1}{\appe ve}T$ where $v$ is of one of the
    forms in the definition of run-time errors.
    By inversion (\cref{lem:inversion}), we have
    $\splitctx[\empty]{\Gamma_1}{\Gamma'_1}{\Gamma''_1}$ and
    $\isExpr[\Empty]{\Gamma_1}v{\funt UV}$. We distinguish the various cases for
    $v$, extracting a contradiction in each case.
    If $v$ is a variable, we find a contradiction because $\Gamma_1$, and hence
    $\Gamma'_1$, contain no arrow types.
    If $v$ is a constant, there is a contradiction because no constant is of an
    arrow type (see \cref{fig:types-constants}). % $\unitt \not\sim \funt TU$.
    If $v$ is a type abstraction, a record, an injection, $\sendk[T]v'$ or
    $\receivek[T]$, a contradiction arises because, by inversion, their types
    are not arrow types.

  % \item Suppose $\isExpr[\Empty]{\Gamma'}{\tappe vU}T$ where $v$ is of one of
  %   the forms in the definition of run-time errors. By inversion,
  %   $\isExpr[\Empty]{\Gamma'} v{\forallt a \kind V}$. The case follows by
  %   distinguishing the various cases for $v$, extracting a contradiction in each
  %   case, as in the item above.

    % If $v$ is a constant, then contradiction because $\Gamma$, and hence
    % $\Gamma_1$, contain no universal types.
    % % 
    % $v =\unite$ and $ v= \forkk$ yield a contradiction. If $v$ is a lambda or
    % $v= \sendk[T]$ or $v= \sendk[T]v'[S]$ or $v=\receivek[T][S]$, then we have a
    % contradiction, because $\funt TU \not\sim \forallt a \kindt V$. If $v$ is a
    % record, injection, or $\selectk\ell$, we have contradiction.
  % \item Analogous to previous cases.
  % \item Analogous to previous cases.

    2, 3 and 4. Analogous to the previous case.
    
    5. Suppose that $\isproc\Gamma{\PROC{E_1[e_1]} \PAR \PROC{E_2[e_2]}}$ where
    $\subj(e_1) = \subj(e_2) = z$. By inversion of \rulenameprocPar we have
    $\isproc{\Gamma_1}{\PROC{E_1[e_1]}}$ and
    $\isproc{\Gamma_2}{\PROC{E_2[e_2]}}$ with
    $\splitctx[\Empty]\Gamma{\Gamma_1}{\Gamma_2}$. By
    \cref{lem:context-typing} we know that
    $\isExpr[\Empty]{\Gamma'_1}{e_1}{T_1}$ with $\Gamma'_1\subseteq\Gamma_1$,
    % $\splitctx[\Empty]{\Gamma_1}{\Gamma'_1}{\Gamma''_1}$,
    and similarly for expression $e_2$. Analysing inversions
    (\cref{lem:inversion}) for the various expressions $e_1$ such that
    $\subj(e_1)$ is defined, we conclude that $\istype{\Empty}{T_1}{\kindsl}$
    and $\isExpr[\Empty]{\Gamma''_1}{z}{T_1}$
    % and $\isExpr[\Empty]{\Gamma'_2}{z}{T_2}$
    with $\Gamma''_1\subseteq\Gamma_1$, and similarly for $e_2$. But this
    contradicts the existence of the splitting of $\Gamma$.
    
    6. Suppose that
    $\isproc\Gamma{\NU xy(\PROC{E_1[e_1]} \PAR \PROC{E_2[e_2]} \PAR q)}$ with
    $\subj(e_1)=x$ and $\subj(e_2)=y$ and $\neg\agree^{xy}(e_1,e_2)$. By
    inversion of {\rulenameprocNew} we obtain
    % \begin{gather*}
      $\istype{\empty} T \kindsl$ and 
      $\isproc{\Gamma, x\colon T, y \colon {\dual T}}{
        (\PROC{E_1[e_1]} \PAR \PROC{E_2[e_2]} \PAR q)
      }$.
    % \end{gather*}
    By congruence and inversion of {\rulenameprocPar} we obtain
    % \begin{gather*}
      $\isproc{\Gamma', x\colon T, y \colon {\dual T}}{
        \PROC{E_1[e_1]} \PAR \PROC{ E_2[e_2]}
      }$.
    % \end{gather*}
    By further inversion of {\rulenameprocPar} we get
    % \begin{gather*}
      $\splitctx[\Empty]{\Gamma', x\colon T, y \colon {\dual
          T}}{(\Gamma_1', x\colon T)}{(\Gamma_2', y\colon \dual T)}$ and 
      $\isproc{\Gamma_1', x\colon T}{\PROC{E_1[e_1]}}$ and
      $\isproc{\Gamma_2', y\colon \dual T}{\PROC{E_2[e_2]}}$.
    % \end{gather*}
    Now we have consider the cases where $\neg\agree^{xy}(e_1,e_2)$,
    all of which contradict the definition of duality. We show
    two representative examples. 
    \begin{itemize}
    \item $e_1$ has the form $\sendk[U_1]v_1[U_2] x$ and $e_2$ has the form
      $\sendk[V_1]v_2[V_2] y$. In this case, $T = \semit{\sout U_1}{U_2}$ and
      $\dual T = \semit{\sout V_1}{V_2}$, which contradicts the definition of
      duality.
    \item $e_1$ has the form  $\sendk[U_1]e [U_2] x$ and $e_2$ has the
      form $\selecte \ell y$ . In this case, $T = \semit{\sout
        U_1}{U_2}$ and $\dual T = \ichoicet{\recordt \ell T L}$, which also
      contradicts the definition of duality.\qedhere
    \end{itemize}
  % \end{enumerate}
\end{proof}

\subsection{Progress for Single-threaded Processes}

Progress for arbitrary processes cannot be guaranteed by typing, for threads may
share more than one channel and interact on these in the ``wrong'' order. Yet,
single-threaded processes enjoy the pleasing result of progress.

Before we embark on the proof for progress, we need to tell which values inhabit
a few select types.

\begin{lemma}[Canonical forms]
  \label{lem:canonical-forms}
  Let $\isExpr{\Gamma}{v}{T}$ with $\isCtx\Gamma\kindsl$.
  % and $\isequiv[\Empty]\Delta T {T'} \kind$.
  \begin{enumerate}
  \item\label{item:3} If $T = \funt UV$, then $v$ is $\abse x{W}e$,
    $\selectk\, k$, $\sendk[W]$, $\sendk[W]v'[X]$, or $\receivek [W][X]$.
%    \vv{Not relevant: for some $\isequiv[\Empty]{\Delta}{U_0}{U}\kind$).}
  \item \sloppy{If $T = \forallt a\kind U$, then
    $v = \tabse a\kind{v'}$, $\sendk$, $\sendk[W]v'[X]$, $\receivek$ or
    $\receivek[W]$.}
  \item If $T = \unitt$, then $v = \unite$.
  \item If $T = \recordt lTL$, then $v = \recorde lvL$.
  \item If $T = \variantt lTL$, then $v = k\,v'$ and $k\in L$.
  \item\label{it:canonical-choice} If $T = \choicet{\recordt \ell T L}$ or
    $T = \msgt U;W$, then $v = x$.
    % \vv{Not relevant: for some $x \in \dom (\Gamma)$.}
  \end{enumerate}
\end{lemma}
\begin{proof}
  Given that $\isCtx\Gamma\kindsl$, the only variables that can be read from
  $\Gamma$ in \cref{it:canonical-choice} are those of a session type.
  We start with an observation, which follows by inspection of the
  typing rules:
  For all syntactic values $v$, there is a derivation
  $\isExpr{\Gamma}{v}{T_0}$ such that $T_0$ has one of the
  forms \ref{item:3}--\ref{it:canonical-choice}. Specifically, $T_0$
  is not a recursive type.
  Second, we exploit \cref{lem:normalized-type-derivation} to
  obtain that the derivation of $\isExpr{\Gamma}{v}{T}$ ends with a
  single use of the rule {\rulenametypeEq}.
  By inversion of that rule, we have $\isExpr{\Gamma}{v}{T_0}$ and
  $\isequiv[\Empty]{\Delta}{T_0}{T}{\kind}$, where $T_0$ has  one of the
  forms~\ref{item:3}--\ref{it:canonical-choice}.
  The last step in the derivation of
  $\isequiv[\Empty]{\Delta}{T_0}{T}{\kind}$ cannot be \rulenameeqFix
  (because $\Theta = \Empty$), nor \rulenameeqRecL (because $T_0$ is not
  a $\mu$ type), nor \rulenameeqRecR (because $T$ is not a $\mu$
  type).

  %%%%%%%%%%%%%%%%%%%%%%%%%%%%%%%%%%%%%%%%%%%%%%%%%%%%%%%%%%%%%%%%%%%%%%%%%%%%%%%%
  Case $T = \funt UV$. By rule \rulenameeqArrow, $T_0 =
  \funt{U_0}{V_0}$ with $\isequiv[\Empty]{\Delta}{U}{U_0}{\kind'}$ and
  $\isequiv[\Empty]{\Delta}{V}{V_0}{\kind'}$. As rule \rulenametypeEq
  is not applicable to construct $\isExpr{\Gamma}{v}{T_0}$, we find
  that this derivation may end with \rulenametypeConst,
  \rulenametypeVar, \rulenametypeAbs, \rulenametypeTApp, or \rulenametypeSel.

  All constants have polymorphic types or type $\unitt$, which rules
  out \rulenametypeConst.

  All variables have kind $\kindsl$, which rules out \rulenametypeVar.

  Rule \rulenametypeAbs is applicable and yields $v = \abse x{U_0}e$.

  Rule \rulenametypeSel is applicable and yields $v = \selectk\, k$.

  Rule \rulenametypeTApp is applicable and yields $v = \sendk[W]$, $v
  = \sendk[W]v'[X]$, or $v = \receivek [W][X]$.

  %%%%%%%%%%%%%%%%%%%%%%%%%%%%%%%%%%%%%%%%%%%%%%%%%%%%%%%%%%%%%%%%%%%%%%%%%%%%%%%%
  Case $T = \forallt a\kind U$. By rule \rulenameeqPoly, $T_0 = \forallt a\kind{U_0}$
  with $\isequiv[\Empty]{\Delta, a:\kind}{U}{U_0}{\kind'}$. As rule \rulenametypeEq
  is not applicable to construct $\isExpr{\Gamma}{v}{T_0}$, we find
  that this derivation may end with \rulenametypeConst,
  \rulenametypeVar, \rulenametypePoly, or \rulenametypeTApp.

  From rule \rulenametypeConst, we obtain $v = \sendk$ or $v =
  \receivek$.

  All variables have kind $\kindsl$, which rules out \rulenametypeVar.

  Rule \rulenametypePoly yields $v = \tabse a\kind{v'}$.

  Rule \rulenametypeTApp is applicable and yields $v = \sendk[W]v'[X]$
  or $v = \receivek[W]$.

  %%%%%%%%%%%%%%%%%%%%%%%%%%%%%%%%%%%%%%%%%%%%%%%%%%%%%%%%%%%%%%%%%%%%%%%%%%%%%%%%
  Case $T = \unitt$. By rule \rulenameeqUnit, $T_0 =
  \unitt$. As rule \rulenametypeEq
  is not applicable to construct $\isExpr{\Gamma}{v}{T_0}$, we find
  that this derivation may end with \rulenametypeConst,
  \rulenametypeVar,  or \rulenametypeTApp.

  From rule \rulenametypeConst, we obtain $v = \unite[m]$.

  All variables have kind $\kindsl$, which rules out \rulenametypeVar.

  There is no polymorphic constant abstracting over unit, which rules
  out \rulenametypeTApp.

  %%%%%%%%%%%%%%%%%%%%%%%%%%%%%%%%%%%%%%%%%%%%%%%%%%%%%%%%%%%%%%%%%%%%%%%%%%%%%%%%
  Case $T = \recordt \ell TL$. By rule \rulenameeqRcd, $T_0 =
  \recordt \ell {T_0}L$ where
  $\isequiv[\Empty]{\Delta}{T_\ell}{T_{0\ell}}{\kind'}$, for all
  $\ell\in L$.

  As rule \rulenametypeEq
  is not applicable to construct $\isExpr{\Gamma}{v}{T_0}$, we find
  that this derivation may end with \rulenametypeConst,
  \rulenametypeVar,  \rulenametypeTApp, or \rulenametypeRcd.

  There is no applicable constant, which rules out \rulenametypeConst.

  All variables have kind $\kindsl$, which rules out \rulenametypeVar.

  There is no polymorphic constant abstracting over records, which rules
  out \rulenametypeTApp.
  
  Rule \rulenametypeRcd yields $v = \recorde lvL$.

  %%%%%%%%%%%%%%%%%%%%%%%%%%%%%%%%%%%%%%%%%%%%%%%%%%%%%%%%%%%%%%%%%%%%%%%%%%%%%%%%
  Case $T = \variantt lTL$. Similar to the case for records.

  %%%%%%%%%%%%%%%%%%%%%%%%%%%%%%%%%%%%%%%%%%%%%%%%%%%%%%%%%%%%%%%%%%%%%%%%%%%%%%%%
  Case $T = \choicet{\recordt \ell T L}$ or $T = \msgt U;W$. 
  In this case, $\kind = \kindsl$ and the equivalence derivation ends with
  rule \rulenameeqST. By \cref{lem:agreement-type-equiv},
  $\istype\Delta{T_0}\kindsl$. 

  As rule \rulenametypeEq
  is not applicable to construct $\isExpr{\Gamma}{v}{T_0}$, we find
  that this derivation may end with \rulenametypeConst,
  \rulenametypeVar, or  \rulenametypeTApp.

  There is no constant typed with kind $\kindsl$, which rules out rule
  \rulenametypeConst.

  Rule \rulenametypeVar is applicable and yields $v = x$.

  Rule \rulenametypeTApp is not applicable as there are no polymorphic
  constants with kind $\kindsl$.
\end{proof}

\begin{theorem}[Progress for the functional sub-language]
  \label{thm:progress-functional-sublanguage}
  Suppose that $\isExpr\Gamma{e}T$ with $\isCtx\Gamma\kindsl$. Then either
  % \vv{Was: Suppose that $\Gamma = {\vec x}\colon{\vec T}$ where
  % $\istype{\Delta}{T_i}{\kindsl}$ and
  % $\isExpr\Gamma{e}T$.}
  \begin{enumerate}
  \item $e$ is a value,
  \item $e$ reduces, or % $\isRed e{e'}$, or
  \item $e$ is of the form $E[e']$ with $e' = \newe U$, $\sendk[U]v[V]x$,
    $\receivek[U][V]x$, $\selecte \ell x$ or $\matche x {\recordp \ell e L}$,
    in which case $e$ is stuck.
  \end{enumerate}
\end{theorem}
\begin{proof}
  By rule induction on hypothesis $\isExpr \Gamma{e}T$.

  Cases \rulenametypeVar, \rulenametypeConst, \rulenametypeAbs and
  \rulenametypePoly. Variables, constants, term and type abstractions are all
  values.

  Case \rulenametypeApp. Then $e = \appe{e_1}{e_2}$. The premises to the rule
  are $\splitctx \Gamma{\Gamma_1}{\Gamma_2}$ and
  $\isExpr {\Gamma_1}{e_1}{\funt U T}$ and $\isExpr {\Gamma_2}{e_2}U$. By
  \cref{lemma:context-split-preserves-kinding}, $\isCtx{\Gamma_1}\kindsl$ and
  $\isCtx{\Gamma_2}\kindsl$. By induction on $e_1$, three cases may happen; we
  analise each separately. Case $e_1$ is a value $v_1$. By induction, this time
  on $e_2$ we are faced with three subcases. Subcase $e_2$ is a value $v_2$:
  canonical forms (\cref{lem:canonical-forms}) for $v_1$ yield one of the
  following cases:
  \begin{itemize}
  \item $v_1 = \abse xU{e'_1}$, in which case $e$ reduces by rule
    \rulenameexpredApp.
  \item $v_1 = \selectk\,k$. In this case inversion (\cref{lem:inversion}) for
    $v_1$ gives
    $\isequiv[\Empty]\Delta {\funt{\ichoicet{\recordt \ell T L}}{T_k}} {\funt
      UV} \kindtm$, hence
    $\isequiv[\Empty]\Delta {\ichoicet{\recordt \ell T L}} U \kind$, for some
    kind $\kind$. Given that $\TypeEquiv$ is commutative, by rule \rulenametypeEq, we
    have $\istype{\Gamma}{v_2}{\ichoicet{\recordt \ell T L}}$ and by canonical
    forms, $v_2 = x$, hence $e$ is stuck at the empty context.
  \item $v_1 = \sendk[W]$, in which case $e = \sendk[W]v_2$ is a value.
  \item $v_1 = \sendk[W]v_1'[X]$. Similarly to $\selectk\,k$, by inversion of
    $v_1$ and canonical forms for $v_2$ we have $v_2 = x$, hence $e$ is stuck at
    the empty context.
  \item $v_1 = \receivek[W][X]$. As above.
  \end{itemize}

  Subcase $e_2$ reduces: take $E=v_1[]$ so that $e$ reduces by rule
  \rulenameexpredCtx. Subcase $e_2$ is stuck at $E[e'_2]$: take $E' = v_1E$ and
  $e$ is stuck at $E'$.

  Case \rulenametypeTApp.
  % Then $T = \subs UaV$ and $e = e_1[U]$. The premises to
  % the rule are $\isExpr{\Gamma}{e_1}{\tabst{a}{\kind}{V}}$ and
  % $\istype{\Delta}{U}{\kind}$.
  Follows the lines of rule \rulenametypeApp.

  Case \rulenametypeRcd. Then $e = \recorde \ell eL$ and $T= \recordt \ell TL$.
  The premises to the rule read $\Delta\vdash \Gamma = \circ\Gamma_\ell$ and
  $\isExpr{\Gamma_\ell}{e_\ell}{T_\ell}$. By
  \cref{lemma:context-split-preserves-kinding}, $\isCtx{\Gamma_\ell}{\kindsl}$.
  By induction three cases may happen, for each $\ell$ in $L$.
  If all $e_\ell$ are values, then $e$ is a value. Otherwise let $e_k$
  ($k\in L$) be the first non-value. Expression $e_k$ may reduce or be stuck. If
  it reduces, take for $E$ the context
  $\{\ell_1=v_1,\dots,\ell_k=[],\dots \ell_n=e_n\}$ and $e$ reduces by rule
  \rulenameexpredCtx. If $e_k$ is stuck and is of the form $E[e'_k]$, then $e$
  is stuck under context $\{\ell_1=v_1,\dots,\ell_k=E,\dots \ell_n=e_n\}$.

  Case \rulenametypeLet. Then $e = \lete \ell xL{e_1}{e_2}$. The premises to the
  rule read $\splitctx{\Gamma}{\Gamma_1}{\Gamma_2}$ and
  $\isExpr{\Gamma_1}{e_1}{\recordt \ell TL}$ and
  $\isExpr{\Gamma_2, (x_\ell\colon T_\ell)_{\ell\in L}}{e_2} T§$. By
  \cref{lemma:context-split-preserves-kinding}, $\isCtx{\Gamma_1}{\kindsl}$. By
  induction on $e_1$ three cases may happen. If $e_1$ is a value, then
  canonical forms give $e_1 = \recorde \ell eL$ and $e_1$ reduces by rule
  \rulenameexpredLet. If $e_1$ reduces, then $e$ reduces by rule
  \rulenameexpredCtx under context $\lete \ell xL{[]}{e_2}$. If $e_1$ is stuck
  under context $E$, then so is $e$ under context $\lete \ell xL{E}{e_2}$.

  Cases \rulenametypeUnitElim, \rulenametypeVariant, \rulenametypeCase and
  \rulenametypeMatch. Similar.

  Case \rulenametypeSel. Expression $\selecte k {}$ is a value.

  Case \rulenametypeNew. Take the empty context for $E$. Then $e$ is stuck under
  context $E[\newe T]$.

  Case \rulenametypeEq. By induction.
\end{proof}

We say that \emph{variable $x$ is an active channel end in expression $e$ under
  contexts $\Delta\mid\Gamma$} if $x$ is free in $e$ and $\isExpr \Gamma x T$
with $\istype{\Delta}{T}{\kindsl}$ but \emph{not} $\istype{\Delta}{T}{\kindsu}$.
Intuitively, this means that $T$ is not (equivalent to) $\skipk$.

\begin{lemma}\label{lemma:closed-values}
  Suppose that $\isExpr\Gamma v T$ with $\isCtx\Gamma\kindsl$ and
  $\istype{\Delta} T\kindtu$. Then $v$ contains no active channel end under
  $\Delta\mid\Gamma$.
\end{lemma}
\begin{proof}
  By structural induction on $v$.
  
  Case $c$. Constants are closed.

  Case $x$. By inversion (\cref{lem:inversion}), $\isExpr\Gamma{x} U$ and
  $\isequiv[\Empty] \Delta UT\kind$. Because type equivalence is symmetric
  (\cref{lem:type-equivalence}), by agreement (\cref{lem:agreement-type-equiv})
  we have $\istype{\Delta}{T}{\kind}$. Hence $\kind=\kindt^\un$. Then, again by
  agreement, $\istype{\Delta}{U}{\kindtu}$. If $x$ is active, then it must be
  that $\istype\Delta U \kindsl$. Contradiction.
  
  Case $\abse xUe$. By inversion and agreement, $\istype{\Delta}{T}{\kindt^m}$
  and $\isCtx{\Gamma}{\kindt^m}$. Hence $m=\un$. Then
  $\isCtx{\Gamma}{\kindtu}$ cannot contain active channel ends.

  Case $\tabse{a}{\kind}{v}$. By inversion,
  $\isExpr[\Delta, a\colon\kind]\Gamma vU$ and
  $\isequiv[\Empty] \Delta T {\forallt a \kind U} \kindtu$. By agreement
  $\istype{\Delta}{T}{\kindt^m}$, hence $m=\un$. Because type equivalence is
  symmetric, by agreement $\istype{\Delta}{\forallt a \kindtu U}{\kindtu}$.
  Inverting type formation (using rules \rulenamekindSub and \rulenamekindPoly),
  we have $\istype{\Delta, a \colon \kindtu}{U}{\kindtu}$. Weakening (\cref{lem:ctx-weakening})
  gives $\isCtx[\Delta,a\colon\kindtu]\Gamma\kindsl$. The result follows by
  induction.
  
  Case $\recorde \ell vL$. Inversion gives $\splitctx{\Gamma}{}{\Gamma_\ell}$,
  $\isExpr{\Gamma_\ell}{v_\ell}{T_\ell}$ and
  $\isequiv[\Empty]\Delta T {\recordt \ell TL}\kind$. Agreement gives
  $\istype{\Delta}{T}{\kind}$, hence $\kind=\kindtu$. From $\isCtx\Gamma\kindsl$
  and $\splitctx{\Gamma}{}{\Gamma_\ell}$ we can show that
  $\isCtx{\Gamma_\ell}\kindsl$. Inverting type formation (using rules
  \rulenamekindSub and \rulenamekindRcd), we have
  $\istype{\Delta}{T_\ell}{\kindtu}$. The result follows by induction.

  Case $\injecte{k}{v}$. Inversion gives $\isExpr\Gamma v {T_k}$, $k\in L$ and
  $\isequiv[\Empty]\Delta T{\variantt \ell{T}L}\kind$. We proceed as above to
  establish that $\istype{\Delta}{T_k}{\kindtu}$ and the result follows by
  induction.
    
  Cases $\selectk\,\ell$, $\sendk[U]$, $\receivek[U]$ and
  $\receivek[U][V]$. All these values are closed.
  
  Case $\sendk[U]v$. Inversion gives $\splitctx{\Gamma}{\Gamma_1}{\Gamma_2}$,
  $\isExpr{\Gamma_1}{\sendk[U]}{\funt VW}$, $\isExpr{\Gamma_2}{v}V$ and
  $\isequiv[\Empty] \Delta T {V} \kind$. From $\isCtx\Gamma\kindsl$ and
  $\splitctx{\Gamma}{\Gamma_1}{\Gamma_2}$ by~\cref{lem:agreement-context-split} we know that
  $\isCtx{\Gamma_2}\kindsl$. Because type equivalence is symmetric, by agreement
  $\istype{\Delta}{\funt[\un] VW}{\kindtu}$. Inverting type formation (using
  rules \rulenamekindSub and \rulenamekindArrow), we have
  $\istype{\Delta}{V}{\kindtu}$. Because $\sendk[U]$ is closed, the result
  follows by induction.

  Case $\sendk[U]v[V]$. Similar to the above case.
\end{proof}

\begin{corollary}[Progress for single-threaded processes]
  Suppose that $\isproc\Gamma{\PROC e}$ with $\isCtx[\empty]\Gamma\kindsl$. Then
  either
  % \par\vv{Was: Suppose that $\Gamma = {\vec x}:{\vec T}$ where
  %   $\istype{\Empty}{T_i}{\kindsl}$ and $\isproc\Gamma{\PROC e}$.}
  \begin{enumerate}
  \item\label{it:progval} $e$ is a value that contains no active channel end
    under $\Empty\mid\Gamma$,
  \item\label{it:progred} $\PROC e$ reduces, or
  \item\label{it:progstuck} $\PROC e$ is of the from $\PROC{E[e']}$ with
    $e' = \sendk[U]v[V]x$, $\receivek[U][V]x$, $\selecte \ell x$ or
    $\matche x {\recordp \ell e L}$, in which case process $\PROC e$ is stuck.
% stuck, in which case $e$ is
%     $E[\sendk[T]v [U] x]$, $E[\receivek[T][U] x]$, $E[\selecte \ell x]$ or
%     $E[\matche x {\recordp \ell e L}]$, for some variable $x$ in $\Gamma$.}
  \end{enumerate}
\end{corollary}
\begin{proof}
  The premises of rule \rulenameprocExp read $\isExpr[\Empty]\Gamma e T$
  and % \\\label{eq:14}
  $\istype{\Empty}T\kindtu$. The statement mostly follows from
  \cref{thm:progress-functional-sublanguage}. For \cref{it:progval}, we further
  use \cref{lemma:closed-values}. For \cref{it:progred}, we further use rule
  \rulenameprocredExp. For \cref{it:progstuck}, we note that while the
  expression $e= E[\newe T]$ is stuck with respect to expression reduction, this
  expression reduces by rule \rulenameprocredNew when wrapped in a process:
  $\isRed{\PROC{E[\newe T]}}{\NU
    xy\PROC{E[(x,y)]}}$. % It remains to examine the
  % cases where $e$ is a value. Here the result follows from \eqref{eq:14} and
  % \cref{lemma:closed-values}.
\end{proof}

%%% Local Variables:
%%% mode: latex
%%% TeX-master: "main"
%%% End:

% LocalWords:  redex

\section{Type Equivalence is Decidable}
\label{sec:type-equiv-decidable}

% We say that a $\mu$-type is \emph{proper} if it
% is not terminated. 
% %
% For the extent of this section we assume that types contain only proper
% $\mu$-subterms.
% %
% Notice that we can always assume that a $\mu$-type is non-terminated, otherwise
% the $\alpha\mu\checkmark$-identification convention allows picking, say,
% $\skipt$ in its stead.

This section focuses on the decidability of relation $\isequivst TU$
under the following assumptions:
\begin{itemize}
\item $\istype{\Delta} T \kinds^m$ and $\istype{\Delta} U \kinds^m$, as per the
  premises of rule \rulenameeqST, \cref{fig:type-equivalence},
\item In all subterms of $T$ and $U$ of the form $\rect{a}{\kind}{V}$, type
  variable $a$ occurs free in $V$, and
\item Both types are $\alpha$-renamed in such a way that they do not share bound
  variables.
\end{itemize}

Note that if $a$ turns out not to occur free in $V$, then
$\rect{a}{\kind}{V}$ may be replaced by $V$, as per \cref{lemma:laws},
\cref{it:mu1}.

Decidability of type equivalence for context-free session types is inherited
from the decidability of bisimilarity of basic process algebras
(BPA)~\cite{DBLP:journals/iandc/ChristensenHS95}. For this purpose, we translate
context-free session types into BPA processes and we prove that the translation
converts equivalent types to bisimilar processes.

%
%The proof builds on the decidability of bisimulation equivalence for the Basic
%Process Algebra (BPA)~\cite{DBLP:journals/iandc/ChristensenHS95}. For technical
%reasons, the syntax of BPA does not include the terminated process
%$\varepsilon$. But $\varepsilon$ is the natural candidate for the image of type
%$\skipt$ or any other terminated type (any type $T$ such that
%$\isDone[\empty] T$). We thus treat terminated types separately, resorting only
%to the below construction when both types are non terminated. 
%Clearly the rules
%in \cref{fig:terminated} provide for an algorithm to decide whether a type is
%terminated or not.

\paragraph{Basic Process Algebra}
BPA processes are defined by the following grammar:
\begin{equation*}
  p \grmeq
  \alpha
  \grmor X
  \grmor p+p
  \grmor p \cdot p
  \grmor \varepsilon
\end{equation*}
where $\alpha$ ranges over a set of atomic actions, $X$ is a BPA process variable, $+$
represents non-deterministic choice, and $\cdot$ represents sequential
composition. The \emph{terminated process} $\varepsilon$ is the neutral element of sequential composition.
We use $\sum _{i\in I} p_i$ when referring to an arbitrary number of non-deterministic choices.

Recursive BPA processes are defined by means of process equations
$\geqs = \{ X_i \triangleq p_i \}_{i\in I}$, where $X_i$ are distinct process
variables. The empty process $\varepsilon$ cannot occur in a process definition;
we included it as part of the syntax because the operational semantics takes
advantage of it~\cite{DBLP:journals/iandc/ChristensenHS95} and, in the scope of
our translation, $\varepsilon$ is the natural candidate for the image of type
$\skipt$ and other terminated types.
A process equation is \emph{guarded} if any
variable occurring in $p_i$ is in the scope of an atomic action. We are
interested in guarded equations only.

% \todo{This notion is not crystal clear. Do we need an inductive def? $\alpha$
% guarded; $p+q$ guarded if both $p$ and $q$
% guarded; $p\cdot q$ guarded in $p$ guarded. \textcolor{blue}{ I don't think it
% is our task.} }
% The \emph{empty process} $\varepsilon$ is the (left and
% right) neutral element of sequential composition.

The labelled transition relation on a set $\geqs$ of guarded
equations is defined by the following rules.
  \begin{gather*}
  	\axiom{\rulenameltsbpaVar}{\alpha \xrightarrow{\alpha} \varepsilon}
  	\quad
    \infrule{\rulenameltsbpaChOne}
    {p \xrightarrow{\alpha} p'}
    {p+q \xrightarrow{\alpha} p'}
     \quad
    \infrule{\rulenameltsbpaChTwo}
    {q \xrightarrow{\alpha} q'}
    {p+q \xrightarrow{\alpha} q'}
    \smallskip\\
    \infrule{\rulenameltsbpaSeqOne}
    {p \xrightarrow{\alpha} p' \\ p'\neq\varepsilon}
    {p\cdot q \xrightarrow{\alpha} p'\cdot q}
    \quad
    \infrule{\rulenameltsbpaSeqTwo}
    {p \xrightarrow{\alpha} \varepsilon}
    {p\cdot q \xrightarrow{\alpha} q}
    \quad
  	\infrule{\rulenameltsbpaDef}
  	{p \xrightarrow{\alpha} p' \\ X \triangleq p \in \geqs}
    {X \xrightarrow{\alpha} p'}
  \end{gather*}
The bisimulation associated to this labelled transition system is denoted 
by $\BPAEquiv[\geqs]$. We omit the subscript when it is clear from context. 
% Christensen \etal proved that bisimulation  is decidable for
% all guarded BPA processes~\cite{DBLP:journals/iandc/ChristensenHS95}.

\paragraph{Translating Context-free Session Types to BPA Processes}

Types are translated in a three-step procedure: we start by identifying the
recursive subterms of a given type, then we translate types to BPA processes
and, finally, we introduce equations for BPA process variables.

The first step identifies
the $\mu$-subterms of a given type $T$.
Assume that
$\subterms(T)=\{\rect {a_1}{\kind_1}T_1,\dots,\rect {a_n}{\kind_n}T_n\}$ is the
set of all $\mu$-subterms in~$T$.

The second step translates type $T$ to a BPA process. To each recursive type $\rect{a_i}{\kind_i}{T_i}$ in $\subterms(T)$ we associate
a distinct process variable $X_i$. The set of atomic actions of BPA processes is
instantiated with $\msgt T$, $\choicet \ell$ and $b$, where $b$ denotes a
polymorphic variable.

To ensure that the processes obtained are well-defined, we need to remove any
occurrences of the empty process $\varepsilon$ from process definitions. For the
purpose, we introduce a \emph{special} sequential composition operator $\odot$
that filters occurrences of $\varepsilon$.
\begin{gather*}
  p\odot \varepsilon = p
  \qquad
    \varepsilon \odot q = q
  \qquad
  p\odot q = p \cdot q
\end{gather*}
Then, the translation is as follows.
\begin{gather*}
  \label{eq:cfst_to_procs}
  \transprocs{\skipt} = \varepsilon 
  \qquad
  \transprocs{\msgt T} = \msgt T
  \qquad
  \transprocs{\choicet{\recordt \ell T L}} = \sum _{\ell\in L} \choicet \ell \odot\transprocs{T_\ell}
  \\
  \transprocs{T;U} =
  \transprocs{T} \odot \transprocs{U}
  \qquad
  \transprocs{b} = b
  \qquad
  \transprocs{a_i} = X_i
  \qquad
  \transprocs{\rect{a_i}{\kind_i}{T_i}} = X_i
  % , \text{ for each } \rect{a_i}{\kind_i}{T_i}\in \subterms(T)\\
\end{gather*}
The terminated type $\skipt$ is translated to the terminated process
$\varepsilon$, message exchanges and polymorphic variables are translated to the
corresponding atomic actions, choices are translated to process choices,
sequential compositions are converted into sequential compositions through
$\odot$ and each recursive type is translated to a corresponding process
variable $X_i$.

The third step translates each of the bodies of the $n$ $\mu$-subterms in the
type. To ensure that the equations are guarded, we perform a preliminary
\emph{unravel} of each type. Intuitively, the unravel function $\unravel \cdot$
consists of unfolding recursive types until exposing a non-recursive type
constructor. In the process we eliminate occurrences of type $\skipt$.
  \begin{gather*}
    %\rulenameunrSeqTwo
    \infrule{}{\unravelSubs[\sigma] {T} = \skipt}{\unravelSubs[\sigma]{\semit TU} = {\unravelSubs[\sigma] {U}}} 
    \qquad
  %\rulenameunrSeqOne
    \infrule{}{\unravelSubs[\sigma] {T} \neq \skipt}{\unravelSubs[\sigma]{\semit TU} =\semit {\unravelSubs[\sigma] T} U}
    \\
    %\rulenameunrRec
    \axiom{}
    %{\sigma' = {\sigma \circ [\rect a \kind T/a]}}
    {\unravelSubs[\sigma]{\rect a \kind T} = \unravelSubs [{\sigma \circ [\rect a \kind T/a]}]{T}}
    \qquad
    %\rulenameunrOther
    \infrule{}{T\neq (\semit{U}{V}), \rect a \kind U}{\unravelSubs[\sigma] {T} = T}
  \end{gather*}

The equations for the $\mu$-subterms of $T$ are obtained by translating the
unravelled $\mu$-subterms: % (under the identity substitution):
\begin{equation*}
  \geqs_T =  
  \{ X_1 \triangleq \transprocs{\unravelSubs[\mathsf{id}] {T_1}}, \dots, X_n \triangleq \transprocs{\unravelSubs[\mathsf{id}] {T_n}} \}
\end{equation*}

As an example of the translation procedure, consider type
\begin{equation*}
  T = \sout\chart ; \rect{a_2}{\kindsl}{T_2}
\end{equation*}
with $T_2 = \skipt ; \rect{a_1}{\kindsl}{T_1}$ and
$T_1 = \sout\intt ; a_1 ; \sint\boolt ; a_2$. The $\mu$-subterms of $T$
identified in step 1 are
$\subterms(T)=\{\rect{a_1}{\kindsl}{T_1},\rect{a_2}{\kindsl}{T_2}\}$. % With the
% purpose of specifying the definitions of the translated processes, 

In step 2 we translate type $T$:
\begin{equation*}
    \transprocs{\sout\chart ; \rect{a_2}{\kindsl}{T_2}} = \sout\chart \cdot X_2
\end{equation*}

In step 3, we translate unravelled versions of bodies of the $\mu$-types in
$\subterms(T)$:
\begin{align*}
  % \transprocs{\unravelSubs[\mathsf{id}]{\sout\chart ; \rect{a_2}{\kindsl}{T_2}}} &= \sout\chart \cdot X_2\\
  \transprocs{\unravelSubs[\mathsf{id}]{\sout\intt ; a_1 ; \sint\boolt ; a_2}} &= \sout\intt \cdot X_1 \cdot \sint\boolt \cdot X_2
  \\
  \transprocs{\unravelSubs[\mathsf{id}]{ \skipt ; \rect{a_1}{\kindsl}{T_1}}} &= \sout\intt \cdot X_1 \cdot \sint\boolt \cdot X_2
\end{align*}
%
% In step 3, we conclude that the BPA process corresponding to type $T$ is $\sout\chart \cdot X_2$ under equations:
to obtain
\begin{equation*}
  \geqs_T =  \{
  X_1 \triangleq\sout\intt \cdot X_1 \cdot \sint\boolt \cdot X_2, \enspace
  X_2 \triangleq\sout\intt \cdot X_1 \cdot \sint\boolt \cdot X_2\}
\end{equation*}

\paragraph*{Type Equivalence Is Decidable}

The decidability of the type equivalence problem for context-free session types
builds on the decidability of the bisimulation of BPA processes.
We denote by $\geqs_{T\cdot U}$ the system of BPA process equations resulting
from the translation of types $T$ and $U$. The $\alpha$-identification
convention allows choosing distinct recursion variables and hence distinct 
process variables.

% \vv{Not quite sure where (or whether) this is relevant.} We have that
% $\istype{\empty}{T}{\kindsu}$ implies $\isDone[\empty] T$. The same does not
% hold for an arbitrary kind context, $\istype{\Delta}{T}{\kindsu}$, just think of
% $\istype{a\colon\kindsu}{a}{\kindsu}$.

\begin{lemma}
	If $\istype{\Delta} U\kinds^m$
  % and $\sigma$ is a substitution on $\dom(\Delta)$, 
	then $\unravelSubs[\sigma] U$ terminates.
\end{lemma}
\begin{proof}
  By rule induction on the hypothesis. 
%	If 
%	$U = \rect a \kindsl V$, rule \rulenameunrRec 
%  is applicable until the type constructor in the top level is either 
%  a message $\msgt T$ or a choice $\choicet{\recordt \ell TL}$. 
%  Since $V$ is contractive on $a$,
%  rules from~\autoref{fig:contractivity} ensure that these type
%  constructors are reachable, hence the unravel function terminates.
\end{proof}

\begin{lemma}
  \label{lem:top_type_constr}
  \sloppy{If $\istype{\Delta} U\kinds^m$ then $\unravelSubs[\sigma] U$ is $\skipk$ or a
  type of the form $(\dots(G;T_1) ; \dots ; T_n)$ with $n\ge0$ and}
  \begin{align*}
    % \unravel T \grmeq& \skipk \grmor G;U\\
    G \grmeq& \msgt T \grmor \choicet{\recordt \ell T L} \grmor a
  \end{align*}
\end{lemma}
\begin{proof}
  By rule induction on the hypothesis. 
  
  Case \rulenamekindMsg. In this case, $U=\msgt{T}$ and the last 
  unravelling rule applies. Cases \rulenamekindSkip, \rulenamekindCh 
  and \rulenamekindVar are similar and constitute the base cases 
  for induction.
  
  Case \rulenamekindSeq for $U = T;V$  has two subcases: 
  if $\unravelSubs[\sigma] T = \skipt$ the proof follows by induction 
  hypothesis on $V$, otherwise we apply induction 
  hypothesis on $T$. 
  
  Case \rulenamekindRec. 
  Since $\istype{\Delta} {\rect a \kindsl T} \kindsl$,
  by \rulenamekindRec we have $\istype{\Delta, a\colon\kindsl} T \kindsl$ and we 
  proceed by induction hypothesis.  
  
  Case \rulenamekindSub. Immediate by induction.
\end{proof}

\begin{lemma}
  \label{lem:terminated}
  Let $\istype{\empty}{T}{\kinds^m}$. Then $\transprocs T = \varepsilon$ if and
  only if $\isDone[\empty]{T}$.
\end{lemma}
\begin{proof}
  The forward implication follows by rule induction on the formation of type $T$
  such that $\transprocs T = \varepsilon$, noting that
  $\transprocs T = \varepsilon$ if $T=\skipt$ or $T=T_1;T_2$ with
  $\transprocs {T_1} = \transprocs {T_2} = \varepsilon$. For the reverse
  direction, assume that $\transprocs T \neq \varepsilon$; we show that
  $\isNotDone T$. By definition of the translation, since
  $\transprocs T \neq \varepsilon$, $T$ takes one of the following forms:
  $\msgt{U}$, $ \choicet{\recordt \ell T L}$, $\rect{a_i}{\kind_i}{T_i}$ or
  $\semit{U}{W}$. (Notice that the cases for variables are not applicable
  because $\istype{\empty}{T}{\kinds^m}$.) According to the rules in
  \cref{fig:terminated}, cases $\msgt{U}$ and $\choicet{\recordt \ell T L}$
  immediately imply $\isNotDone T$. If $T= \rect{a_i}{\kind_i}{T_i}$, since
  $a_i$ occurs free in $T_i$ (as assumed at the beginning of this section) and
  $T_i$ is contractive on $a_i$, we also have $\isNotDone T$. The conclusion
  that $\isNotDone {\semit{U}{W}}$ follows from the previous observations, using
  \rulenametermSeq.
  % cases $T = \skipt$ and $T = \semit{U}{V}$, 
	% are immediate by definition of $\transprocs{\cdot}$. If 
	% $T = \rect a {\kinds^m} U$, then $a$ occurs free in $U$ 
	% (as assumed at the beginning of this section). Hence, by rules
	% of \cref{fig:terminated}, $\isNotDone U$ and, thus, 
	% we would have $\isNotDone T$.
\end{proof}
% Notice that if $T$ is a $\mu$-type, then $\rect {a_n}{\kind_n}{T_n}$ 
% is $T$ itself and the last equation is redundant.

\begin{lemma}
	If $\isNotDone[]T$ and $\isNotDone[]U$, then
	the BPA process equations in $\geqs_{T\cdot U}$ are guarded.
    %\todo{see \cref{lem:terminated}}
\end{lemma}
\begin{proof}
	Immediate from the definition of $\geqs_{T\cdot U}$, using
	\cref{lem:terminated,lem:top_type_constr}.
\end{proof}

The following result helps in showing that the translation preserves
bisimulation.

\begin{lemma}
\label{lem:translated_transitions}
\changed{
If $\isType{\tvarset}{\Delta} T \kinds^m$
%$T \neq a$ for $a\not\in\tvarset$, 
and 
$\transeqs{T}\xrightarrow{\alpha} p$ then 
$ p = \transeqs{T'}$ for $T'$ such that $T\LTSderives[\alpha] T'$.}
\end{lemma}

\begin{proof}
  By case analysis on the definition of the translation. %~\eqref{eq:cfst_to_eqs}.
  If $T = \rect{a_i}{\kind}{U}$, then $\transeqs{T}=X_i$ and
  $(X_i \triangleq \transprocs{\unravelSubs[\mathsf{id}] {U}}) \in \geqs_T$.
  Note that, by rule \rulenameltsbpaDef, the labelled transitions of
  $\transeqs{T}$ are the transitions of $\transprocs{\unravelSubs {U}}$.
  By \cref{lem:top_type_constr}, we have four subcases to analyse.
  
  Subcase $\unravelSubs{U}= \skipt$. Does not have any transition.
  
  Subcase $\unravelSubs{U} = \msgt T; T_1; \cdots ; T_n$. In this case we have
  $\transprocs{\unravelSubs{U}} = \transprocs{\msgt{T}}\odot \transprocs{T_1}\cdots \transprocs{T_n}
  \xrightarrow{\msgt{T}} \transprocs{T_1}\odot\cdots\odot \transprocs{T_n} = \transprocs{T_1;\ldots; T_n}$
  and $T'= T_1;\ldots; T_n$.
  
  Subcase $\unravelSubs{U}= \choicet{\recordt \ell U L}; T_1; \ldots; T_n$. In
  this case we have
  $\transprocs{\unravelSubs{U}} = (\sum _{\ell\in L} \choicet \ell
  \odot\transprocs{U_\ell})\odot \transprocs{T_1}\cdots \transprocs{T_n}
  \xrightarrow{\star \ell} \transprocs{U_\ell} \odot \transprocs{T_1}\cdots
  \transprocs{T_n} = \transprocs{U_\ell; T_1; \ldots; T_n}$ and we have
  $T'=U_\ell; T_1; \ldots; T_n$.
  
  Subcase $\unravelSubs{U}= \changed{a_j}; T_1; \ldots; T_n$, \changed{ with
    $j\neq i$}. In this case we have
  $\transprocs{\unravelSubs{U}} = \changed{X_j}\odot \transprocs{T_1}\cdots
  \transprocs{T_n}$, with
  \changed{$(X_j \triangleq \transprocs{\unravelSubs[\mathsf{id}] {U_j}}) \in
    \geqs_T$. Since types are contractive, we have}
  $\transprocs{\unravelSubs{U}} = \changed{X_j}\odot \transprocs{T_1}\cdots
  \transprocs{T_n} \xrightarrow{\changed{\alpha}}
  \changed{\transprocs{U_j'}\odot}\transprocs{T_1}\cdots \transprocs{T_n} =
  \transprocs{\changed{U_j';}T_1;\ldots; T_n}$ and
  $T' = \changed{U_j';}T_1;\ldots; T_n$.
  
  In all these cases, $\unravelSubs{U}\LTSderives[\alpha] T'$ and, thus, $T\LTSderives[\alpha] \subs{T}{a_i}{T'}$.
  Noting that $\transprocs{T'} = \transprocs{\subs{T}{a_i}{T'}}$, we conclude that 
  $\transeqs{T}\xrightarrow{\alpha} \transeqs{\subs{T}{a_i}{T'}}$.
  If $T\neq \rect{a}{\kind}{U}$, we proceed similarly.
\end{proof}

%\begin{lemma}
%	The BPA process defined by the equations in $\geqs_T$ with root
%	$a_{n+1}$ is deterministic.\todo{Why relevant?}
%\end{lemma}
%%
%\begin{proof}
%  \todo{}
%\end{proof}

\begin{lemma}
  \label{lem:full-abstraction}
  The translation $\transeqs{\cdot}$ is a fully abstract encoding, \ie
  $\isequivst{T}{U}$ if and only if
  $\isequivbpa[\geqs^*]{\transeqs{T}}{\transeqs{U}}$, where $\geqs^*$ is a
  system of process equations such that $\geqs_{T\cdot U}\subseteq \geqs^*$.
\end{lemma}
\begin{proof}
  For the forward direction, let $\mathcal{S}$ be a bisimulation for $T$ and
  $U$. Consider the relation
  $\mathcal{R} = \{(\transeqs{T_0}, \transeqs{U_0}) \mid (T_0,U_0)\in
  \mathcal{S}\}$ and the set of process equations
  $\geqs^* = \bigcup_{(T_0,U_0)\in \mathcal{S}} \geqs_{T_0\cdot U_0}$. We have
  $(\transeqs{T}, \transeqs{U})\in \mathcal{R}$. To prove that $\mathcal{R}$
  is a bisimulation for $\transeqs{T}$ and $\transeqs{U}$, we just need to prove
  that given $(\transeqs{T_0}, \transeqs{U_0})\in \mathcal{R}$, if
  $\transeqs{T_0}\xrightarrow{\alpha} p$, then
  $\transeqs{U_0}\xrightarrow{\alpha} p'$ with $(p,p')\in \mathcal{R}$, and
  vice-versa. To prove that $U_0$ mimics $T_0$, we use the definition of
  $\transeqs{T_0}$, \cref{lem:translated_transitions}, and the fact that
  $\mathcal{S}$ is a bisimulation.
  %~\eqref{eq:cfst_to_eqs}, 
%  the structure of the unravel presented in \cref{lem:top_type_constr}, and 
%  \cref{lem:translated_transitions}.\\
%  In the first case of the definition in~\eqref{eq:cfst_to_eqs}:
  For instance, if
  $T_0 = \rect{a_i}{\kind}{T_0'}$, we have $\transeqs{T_0}=X_i$ and 
  $(X_i\triangleq \transprocs{\unravelSubs{T_0'}} )\in \geqs^*$. The proof
  proceeds by case analysis on the structure of $\unravelSubs[\sigma]{T_0'}$, using
  \cref{lem:top_type_constr}. We detail the subcase
  $\unravelSubs{T_0'} = \msgt S; S_1; \ldots; S_k$. In this case,
  $\alpha = \msgt{S}$ and $p = \transprocs{\skipt;S_1; \ldots; S_k}$
  and $T_0 \LTSderives[\msgt{S}] \subs{S_0}{a}{(\skipt;S_1; \ldots; S_k)}$.
  Since $\mathcal{S}$ is a bisimulation then 
  $U_0 \LTSderives[\msgt{S}] U_0'$ and $(\subs{T_0}{a}{(\skipt;S_1; \ldots; S_k)}, U_0')\in \mathcal{R}.$
  Hence, $\unravelSubs{U_0}$ is of the form $\msgt{S}; U_0'$ and 
  $\transprocs{\unravelSubs{U_0}}\xrightarrow{\msgt{S}} \transprocs{U_0'}$.
  Furthermore, $(\transprocs{\subs{T_0}{a}{(\skipt;S_1; \ldots; S_k)}}, \transprocs{U'} )\in \mathcal{R}$.
  Since, by definition of the translation, $\transprocs{T_0} = \transprocs{a_i}$, we have 
  $\transprocs{\subs{T_0}{a}{(\skipt;S_1; \ldots; S_k)}} = \transprocs{\skipt;S_1; \ldots; S_k}$
  and we conclude that $(\transprocs{{\skipt;S_1; \ldots; S_k}}, \transprocs{U'} )\in \mathcal{R}$.
  The other subcases for $\unravelSubs{T_0'}$ and the case where 
  $T_0 \neq \rect{a}{\kind}{T_0'}$ are similar, as well
  as the analysis for $U_0$.
  
  Reciprocally, let $\mathcal{R}$ be a bisimulation 
  for $\transeqs{T}$ and $\transeqs{U}$ and let us fix
  $\subterms(T)=\{\rect {a_1}{\kind_1}T_1,\dots,\rect {a_n}{\kind_n}T_n\}$
  and
  $\subterms(U)=\{\rect {b_1}{\kind_1'}U_1,\dots,\rect {b_m}{\kind_m'}U_m\}$.
  Consider
  the relation 
  $\mathcal{S} = \{(T_0,U_0) \mid (\transeqs{T_0}, \transeqs{U_0})\in \mathcal{R},
  T_0 \neq a_i, U_0 \neq b_j \}$.
  Obviously, $(T,U)\in \mathcal{S}$. Hence, to prove that 
  $\mathcal{S}$ is a bisimulation, we just need to prove that
  any simulation step of $T_0$, $T_0\LTSderives[\alpha] T_0'$,
  can be mimicked by $U_0$, $U_0 \LTSderives[\alpha] U_0'$, with 
  $(T_0',U_0')\in \mathcal{R}$, and vice-versa.
  The proof is done by case analysis on $T_0$ and $U_0$, using
  the rules of \cref{fig:lts}, \cref{lem:translated_transitions}, 
  and the definition of $\transprocs{\cdot}$.
\end{proof}

\begin{theorem}
  \label{thm:type-equiv-decidable}
  The type equivalence problem for context-free session types is decidable.
\end{theorem}
\begin{proof}
  By \cref{lem:full-abstraction}, the full abstraction $\transeqs{\cdot}$
  reduces the type equivalence problem of context-free session types to the
  bisimulation of BPA processes. The latter was proved to be decidable by
  Christensen et al.~\cite{DBLP:journals/iandc/ChristensenHS95}. Hence, type
  equivalence for context-free session types is decidable.
\end{proof}

%%% Local Variables:
%%% mode: latex
%%% TeX-master: "main"
%%% End:

\section{Algorithmic Type Checking}
\label{sec:algorithmic-typing}

This section presents a bidirectional
procedure~\cite{DBLP:journals/toplas/PierceT00} to decide expression formation
as in \cref{fig:typing-exps}.

% type checking and algorithmic type
% formation. The main result of this section, \emph{correctness}, can be broken
% down in two standard components: \emph{soundness} and \emph{completeness} of the
% algorithmic type checking regarding the declarative system presented in
% \cref{fig:typing-exps}. Our algorithms for type formation and typechecking are
% based on bidirectional typing \cite{DBLP:journals/toplas/PierceT00}.

\subsection{Algorithmic Kinding, Algorithmic Equivalence and Type Normalisation}

\paragraph{Kinding}

% We start by introducing the machinery needed for proving the main result.
The rules for algorithmic type formation are in \cref{fig:alg-kinding}. Sequents
are of the form $\algKindOut[\Delta_\In]{T_\In}{\kind_\Out}$ for kind synthesis
and $\algKindIn[\Delta_\In]{T_\In}{\kind_\In}$ for the check against relation.
Marks $\In$ and $\Out$ describe the input/output mode of parameters.
Algorithmic kinding relies on contractivity (\cref{fig:contractivity}), which in
turn relies on the terminated relation for session (\cref{fig:terminated}),
both of which are algorithmic.

Rules
mostly follow their declarative counterparts in \cref{fig:kinding}.
For example, the rule for sequential composition, \rulenamealgkindSeq, requires
both types to be session types, and returns kind $\kinds^{m\join n}$, where
$m\join n$ is the least upper bound of the multiplicities of the synthesised
kinds.
Rule \rulenamealgkindMsg checks that the type is of message kind and returns kind
$\kindsl$.
Rule \rulenamealgkindCh checks that all the elements of a choice type are
themselves session types, by checking against kind $\kindsl$ and returns kind
$\kindsl$.

As with the declarative rules, recursive types $\rect a{\kind} T$
are both session and functional. Rule \rulenamealgkindRec checks whether type
$T$ is contractive on type variable~$a$. Then, it ensures that $T$ has the kind
of $a$ by checking $T$ against $\kind$. Rule \rulenamealgkindPoly checks
whether type $T$ is well formed under contexts $\tvarset, a$ and
$\Delta, a\colon\kind$.
%
% Rules \rulenamealgkindRcd and \rulenamealgkindData both check that all the
% variant elements are well-formed and then results in kind $\kindt$ whose
% multiplicity is the least upper bound between those synthesised.
%
Rules \rulenamealgkindRcd and \rulenamealgkindData both check that all
components are well-formed. The result is kind $\kindt^{\join m_\ell}$,
where $\join m_\ell$ is the least upper bound of the multiplicities of the
synthesised kinds.
Rule \rulenamealgkindVar looks for an entry $a\colon\kind$ in $\Delta$ and
returns $\kind$. Rule \rulenamealgkindArrow requires both types $T$ and $U$ to
be well-formed and return $\kindt^{m}$, where $m$ stands for the function
multiplicity. $\skipt$ has kind $\kindsu$ and $\unitt$ has kind $\kindm^m$.

\begin{figure}[t!]
  \decltworel{Kind synthesis and check-against}
    {$\algKindOut[\Delta_\In]{T_\In}{\kind_\Out}$}
    {$\algKindIn[\Delta_\In]{T_\In}{\kind_\In}$}
  \begin{gather*}
    \infrule{\rulenamealgkindType}
    {\algkindout [\dom(\Delta)] \Delta T \kind }
    {\algKindOut T \kind}
    \qquad
    \infrule{\rulenamealgkindTypeAgainst}
    {\algkindin [\dom(\Delta)] \Delta T \kind }
    {\algKindIn T \kind}    
  \end{gather*}
  
  \declrel{Kind synthesis (inner)}{$\algkindout[\tvarset_\In]{\Delta_\In}{T_\In}{\kind_\Out}$}
  \begin{gather*}
    % Session types
    \axiom{\rulenamealgkindSkip}{\algkindout{\Delta} \skipt \kindsu}
    \qquad
    \infrule{\rulenamealgkindMsg}{\algkindin{\Delta}{T}{\kindml}}{\algkindout{\Delta}{\msgt T}{\kindsl}}
    \qquad
    \infrule{\rulenamealgkindCh}{
      \algkindin{\Delta}{T_\ell}{\kindsl}
    }{
      \algkindout{\Delta}{\choicet{\recordt \ell T L}}{\kindsl}
    }
    \\    
    \infrule{\rulenamealgkindSeq}{
      % \algkindin{\Delta}{T}{\kindsl}
      % \\
      % \algkindin{\Delta}{U}{\kindsl}
      \algkindout{\Delta}{T}{\kinds^m}
      \\
      \algkindout{\Delta}{U}{\kinds^n}
    }{
%      \algkindout{\Delta}{(\semit{T}{U})}{\kindsl}
      \algkindout{\Delta}{(\semit{T}{U})}{\kinds^{m\join n}}      
    }
    \quad
    \axiom{\rulenamealgkindUnit}{\algkindout{\Delta}{\unitk}{\changed{\kindm^m}}}
    \\
    \infrule{\rulenamealgkindArrow}{
      \algkindout{\Delta}{T}{\kind_1}
      \\
      \algkindout{\Delta}{U}{\kind_2}
    }{
      \algkindout{\Delta}{(\funt{T}{U})}{\kindt^m}
    }
    \\
    \infrule{\rulenamealgkindRcd}{
      \algkindout{\Delta}{T_\ell}{v^{m_\ell}}\\
      \changed{(\forall\ell\in L)}
    }{
      \algkindout{\Delta}{\recordt \ell T L}{\kindt^{\join m_\ell}}%\kindt
    }
    \qquad
    \infrule{\rulenamealgkindData}{
      \algkindout{\Delta}{T_\ell}{v^{m_\ell}}\\
      \changed{(\forall\ell\in L)}
    }{
      \algkindout{\Delta}{\variantt \ell T L}{\kindt^{\join m_\ell}}%\kindt
    }  
    \\
    % Functional and session
    % Forall
    \infrule{\rulenamealgkindPoly}
    	{\algkindout[\tvarset,a]{\Delta,a\colon\kind} T {v^m}}
    	{\algkindout{\Delta}{\forallt a \kind T}{\kindt^m}}
    \;\;\,
    \infrule{\rulenamealgkindRec}{
      \isContr aT
      \\
      \algkindin{\Delta, a\colon \kind} {T} \kind
    }{
      \algkindout{\Delta}{\rect a{\kind} T}{\kind}
    }
    \;\;\,
    \infrule{\rulenamealgkindVar}{a \colon \kind \in \Delta}{\algkindout{\Delta} a \kind}
  \end{gather*}
  \declrel{Kind check-against (inner)}{$\algkindin[\tvarset_\In]{\Delta_\In}{T_\In}{\kind_\In}$}
  \begin{equation*}
    \infrule{\rulenamealgcheckAgainst}{
      \algkindout \Delta T {\kind_1}
      \\
      \isSubkind{\kind_1}{\kind_2}
    }{
      \algkindin \Delta T {\kind_2}
    }
  \end{equation*}
  \caption{Algorithmic type formation (kinding)}
  \label{fig:alg-kinding}
\end{figure}

%%% Local Variables:
%%% mode: latex
%%% TeX-master: "main"
%%% End:

We now look at the correctness of algorithmic type formation.

\begin{lemma}[Kinding correctness]\
  \begin{enumerate}
  \item (Soundness, synthesis) If $\algKindOut{T}{\kind}$, then
    $\istype{\Delta}{T}{\kind}$.
  \item (Soundness, check against) If $\algKindIn{T}{\kind}$, then
    $\istype{\Delta}{T}{\kind}$.
  \item (Completeness, synthesis) If $\istype{\Delta}{T}{\kind}$, then
    $\algKindOut{T}{\kind'}$ with $\kind' \subk \kind$.
  \item (Completeness, check against) If $\istype{\Delta}{T}{\kind}$, then
    $\algKindIn{T}{\kind}$.
  \end{enumerate}
  \label{lem:kinding-correctness}
\end{lemma}
\begin{proof}
  Straightforward mutual rule induction on each hypothesis.
\end{proof}

\paragraph{Equivalence}

The proper integration of context-free session types in a compiler requires the
design and implementation of a type equivalence algorithm. The rules in
  \cref{fig:type-equivalence} are algorithmic
% an algorithm for type equivalence may be extracted
% using a general decision procedure for coinductive
% definitions~\cite{DBLP:books/daglib/0005958},
provided an algorithm to decide $\isequivst TU$ is available.
A practical algorithm to check type equivalence for context-free session types
can be found in Almeida \etal~\cite{DBLP:conf/tacas/AlmeidaMV20}.
%
%In order to check the equivalence of types $T$ and $U$ under kinding context
%$\Delta$, we write $\algequivin[\Delta]{T}{U}$.

% \begin{theorem}[Correctness of algorithmic type
%   equivalence~\cite{DBLP:conf/tacas/AlmeidaMV20}]
%   \label{thm:type-equivalence}
%   $\algequivin TU$ if and only if $\isequiv{\Delta}{T}{U}{\kind}$.
%   % \begin{description}
%   % \item[Soundness] If $\algequivin TU$, then $\isequiv{\Delta}{T}{U}{\kind}$.
%   % \item[Completeness] If $\isequiv{\Delta}{T}{U}{\kind}$, then $\algequivin TU$.
%   % \end{description}
% \end{theorem}

% \vv{What is this algorithm $\isequiv{\Delta}{T}{U}{\kind}$? for session types is
%   in paper [4]; for the remaining papers is a ``general decision procedure''.
%   But if want to detail the proof, then the backward direction places us in the
%   same situation and in \cref{lem:agreement-type-equiv}: show an inductive
%   definition from a coinductive one.}

\paragraph{Context Operations}

% We introduce basic properties of context split as defined in
% \cref{fig:context-split}. Let % $\dom{(\Gamma)}$ denote the set of variables $x$
% % such that $x \colon T \in \Gamma$ and let 
% $\mathcal L(\Gamma)$ and
% $\mathcal U(\Gamma)$ to refer to the linear and unrestricted portions of
% $\Gamma$, respectively~\cite{walker:substructural-type-systems} (see \cref{lemma:context-split-preserves-kinding}).

% \begin{lemma}[Properties of context split]\label{lem:cSplitProps}
%   Let $\splitctx \Gamma {\Gamma_1}{\Gamma_2}$
%   \begin{enumerate}
%   \item $\mathcal U(\Gamma) = \mathcal U{(\Gamma_1)} = \mathcal U{(\Gamma_2)}$
%   \item If $x\colon T\in\Gamma$ and $\istype{\Delta}{T}{\kindtl}$, then either
%     $x \in \dom{(\Gamma_1)}$ and $x \not\in\dom{(\Gamma_2)}$, or 
%     $x \not\in\dom{(\Gamma_1)}$ and $x \in\dom{(\Gamma_2)}$.
%   \end{enumerate}
% \end{lemma}
% \begin{proof}
%   By rule induction on the context split hypothesis.
% \end{proof}
  
% CONTEXT DIFFERENCE

Context difference, $\ctxdiff[\Delta_\In]{\Gamma_\In}{x_\In}{\Gamma_\Out}$
in~\cref{fig:context-difference}, checks that $x$ is not linear and removes
it from the incoming context.
% in~\cref{fig:context-difference}, checks that linear variables do not appear in
% the context and removes unrestricted entries from the context. 
\begin{figure}[t!]  
  \declrel{Context difference}{$\ctxdiff[\Delta_\In]{\Gamma_\In}{x_\In}{\Gamma_\Out}$}
  \begin{equation*}
    \infrule{}{
      x\colon T\notin\Gamma
    }{
      \ctxdiff \Gamma x \Gamma
    }
    \qquad
    \infrule{}{
      \algkindin \Delta T \kindtu
    }{
      \ctxdiff {(\Gamma_1,x\colon T,\Gamma_2)} x {\Gamma_1,\Gamma_2}
    }
  \end{equation*}
  \caption{Context difference}
  \label{fig:context-difference}
\end{figure}

%%% Local Variables:
%%% mode: latex
%%% TeX-master: "main"
%%% End:

\begin{lemma}[Properties of context difference]\
  Let $\ctxdiff{\Gamma_1}x\Gamma_2$.
  \begin{enumerate}
  \item $\Gamma_2 = \Gamma_1 \backslash \{x \colon T\} $
  \item $\mathcal L(\Gamma_1) = \mathcal L(\Gamma_2)$
  \item If $x\colon T \in \Gamma_1$, then $\istype \Delta T {\kindtu}$ and $x \not\in \dom(\Gamma_2)$
  \end{enumerate}
  \label{lem:ctx-diff}
\end{lemma}
\begin{proof}
  By rule induction on the context difference hypothesis.
\end{proof}

\paragraph{Type Normalisation}

Type normalisation, $\normalisation{T_\In}{U_\Out}$ in \cref{fig:norm}, is
used in the elimination rules for algorithmic expression formation. Type
normalisation exposes the top-level type constructor by unfolding recursive
types, removing terminated types and adjusting associativity
on  sequential composition.

\begin{figure}[t]
  \declrel{Type concatenation}{$\append {T_\In} {T_\In} = T_\Out$}
  \begin{align*}
    \append{T}{\skipt} = T
    \qquad
    \infrule{}{V \neq \skipt}{\append {(\semit TU)} {V} = \semit T {(\append U V)}}
    \qquad
    \infrule{}{T \neq \semit VW \and U \neq \skipt}{\append {T}{U} = \semit TU}
  \end{align*}
  \emph{Type normalisation} \hfill\fbox{$\normalisation {T_\In} {T_\Out}$}
  \begin{gather*}
    \infrule{}{\isDone[\empty] T \\ \normalisation{U}{U'}}
    {\normalisation{\semit TU}{U'}}
    \qquad
    \infrule{}{\isNotDone[\empty] T\\\normalisation{T}{T'}}
    {\normalisation{\semit TU}{\append {T'} U}}
    \\
    \infrule{}{\normalisation{T}{T'}}
    {\normalisation{\rect a\kind T}{\subs{\rect a\kind T'}{a}{T'}}}
    \qquad
    \infrule{}{T\neq (\semit UV),(\rect x\kind U)}{\normalisation{T}{T}}
  \end{gather*}
  \caption{Type normalisation}
  \label{fig:norm}
\end{figure}

%%% Local Variables:
%%% mode: latex
%%% TeX-master: "main"
%%% End:

\begin{lemma}[Type normalisation]
  \label{lem:normProps} 
  % \begin{enumerate}
  % \item
  If $\istype{\Delta} {T} \kind$, then $\normalisation T U$ and
  $\isequiv[\Empty]\Delta{T}{U}{\kind}$. Furthermore, $U$ does not start with $\mu$ and
  if $U=V;W$ then $V$ is a message or a choice.
  % \item If $\isequiv\Delta{T}{U}{\kind}$, then $\normalisation T{T'}$,
  % $\normalisation U{U'}$, and $\isequiv\Delta {T'}{U'}\kind$.
  % \end{enumerate}
\end{lemma}
\begin{proof}
  By rule induction on the hypothesis. Consider the case where the derivation ends
  with rule \rulenamekindRec. We know that $\isContr aT$ and
  $\isType{\tvarset}{\Delta, a\colon\kind} T \kind$. By induction hypothesis we
  obtain $\normalisation TT'$ and $\isequiv[\Empty]{\Delta, a\colon\kind}{T}{T'}\kind$
  where $T'$ does not start with $\mu$ and if $T' = V;W$, then $V$ is
  either a message or a choice. By applying the third rule of type normalisation
  (\cref{fig:norm}) we get
  $\normalisation{\rect a\kind T}{\subs{\rect a\kind T'}{a}{T'}}$. We conclude
  that
  $\isequiv[\Empty]{\Delta, a\colon\kind}{\rect a\kind T}{\subs{\rect a\kind
      T'}{a}{T'}}{\kind}$ via the properties of type bisimilarity.
  
  When the derivation ends with rule \rulenamekindSeq, we consider two cases for 
  $T = T_1;T_2$,
  regarding whether $T_1$ is terminated or not. Take the latter as an example. We
  have that $\isType{\tvarset}{\Delta} {T_1} \kindsl$ and
  $\isType{\tvarset}{\Delta} {T_2} \kindsl$. Induction hypothesis on the first
  premise yields $\normalisation {T_1}{T_1'}$ and $\isequiv[\Empty]{\Delta}{T_1}{T_1'}{\kind}$.
  Since $\isNotDone[\empty] {T_1}$, we apply the second rule of type normalisation
  to obtain $\normalisation{\semit {T_1}{T_2}}{\append {T_1'} {T_2}}$. Now we have to prove
  three different cases depending on the nature of the type concatenation
  operation. The first two follow by induction hypothesis, using the properties 
  of type bisimilarity
  (\cref{lemma:laws}) and the last follows by reflexivity. 
  
  All the other cases follow
  directly by the last rule of type normalisation.
\end{proof}

% DUALITY, TYPE AND CONTEXT EQUIVALENCE

% The dual function on types, $\dual T$, can be readily extracted from the
% declarative duality rules in \cref{fig:duality} and is omitted.
% %
% Soundness: if $\isdual{T}{U}{\kinds^m}$, then $\dual T = U$.
% %
% Completeness: if $\isType{\tvarset}{\Delta}{T}{\kinds^m}$, then $\dual T = U$ and
% $\isdual{T}{U}{\kinds^m}$.
% %
% \todo{Check}

% The rules \cref{fig:duality} define a decidable function for duality.

% TYPING

\subsection{Algorithmic Type Checking}

\begin{figure}[t!]
  \declrel{Type synthesis}{$\algtypeout[\Delta_\In]{\Gamma_\In}{e_\In}{T_\Out}{\Gamma_\Out}$}
  \begin{gather*}   
    % Unit
    \axiom{\rulenamealgtypeConst}
    {\algtypeout \Gamma c {\typeof(c)} \Gamma}
    \\
    % Variables
    \infrule{\rulenamealgtypeVarLin}{
      \algKindOut{T}{\prekind^\lin}
    }{
      \algtypeout{\Gamma,x\colon T} x T \Gamma
    }
    \qquad
    \infrule{\rulenamealgtypeVarUn}{
      \algKindOut{T}{\prekind^\un}
    }{
      \algtypeout{\Gamma,x\colon T}  x T {\Gamma,x\colon T}
    }
    \\
    % Linear Abstraction
    \infrule{\rulenamealgtypeAbsLin}{
      \algKindOut{T_1}{\_}
      \\
      \algtypeout{\Gamma_1,x\colon T_1}{e}{T_2}{\Gamma_2}
    }{
      \algtypeout{\Gamma_1}{(\abse[\lin] x{T_1}e)}{(\funt[\lin]{T_1}{T_2})}{\Gamma_2 \div x}
    }
    \\
    % Unrestricted Abstraction
    \infrule{\rulenamealgtypeAbsUn}{
      \algKindOut{T_1}{\_}
      \\
      \algtypeout{\Gamma_1,x\colon T_1}{e}{T_2}{\Gamma_2}
      \\
      \Gamma_1 = \Gamma_2 \div x
    }{
      \algtypeout{\Gamma_1}{(\abse[\un] x{T_1}e)}{(\funt[\un]{T_1}{T_2})}{\Gamma_1}
    }
    \\
    % Application
    \infrule{\rulenamealgtypeApp}{
      \algtypeoutnorm{\Gamma_1}{e_1}{\funt TU}{\Gamma_2}
      % \\
      % \normalisation T {(\funt UV)}
      \\
      \algtypein{\Gamma_2}{e_2}{T}{\Gamma_3}
    }{
      \algtypeout{\Gamma_1}{\appe{e_1}{e_2}}{U}{\Gamma_3}
    }
    \\
    % Type Abstraction
    \infrule{\rulenamealgtypeTAbs}{
      \algtypeout[\Delta,a\colon\kind]{\Gamma_1}{e}{T}{\Gamma_2}
      \\
      a\notin\free(\Gamma_2)
    }{
      \algtypeout{\Gamma_1}{\tabse a\kind e}{\forallt a \kind T}{\Gamma_2}
    }
    \\
    % Type application
    \infrule{\rulenamealgtypeTApp}{
      \algtypeoutnorm{\Gamma_1}{e}{\forallt a \kind U}{\Gamma_2}
      % \\
      % \normalisation {U} {\forallt a \kind V}
      \\
      \algKindIn T \kind
    }{
      \algtypeout {\Gamma_1}{\tappe eT}{\subs T a U} \Gamma_2
    } 
 \end{gather*}
   \caption{Algorithmic type checking (synthesis)}
   \label{fig:alg-typing}
\end{figure}

\begin{figure}[t!]  
  \declrel{Type synthesis}{$\algtypeout[\Delta_\In]{\Gamma_\In}{e_\In}{T_\Out}{\Gamma_\Out}$}
  \begin{gather*}
    % Record constructor
    \infrule{\rulenamealgtypeRcd}{
      \algtypeout{\Gamma_1}{e_1}{T_1}{\Gamma_2}
      \\
      \cdots
      \\
      \algtypeout{\Gamma_n}{e_n}{T_n}{\Gamma_{n+1}}
    }{
      \algtypeout{\Gamma_1}{\recorde \ell e{\{1,..,n\}}}{\recordt \ell T{\{1,..,n\}}}{\Gamma_{n+1}}
    }
    \\
    % Record destructor
    \infrule{\rulenamealgtypeLet}{
      \algtypeoutnorm{\Gamma_1}{e_1}{\recordt{\ell}{T}{L}}{\Gamma_2}
      % \\
      % \normalisation{T}{\recordt{\ell}{T}{L}}
      \\
      \algtypeout{\Gamma_2, (x_\ell\colon T_\ell)_{\ell\in L}}{e_2}{U}{\Gamma_3}
    }{
      \algtypeout{\Gamma_1}{\lete \ell xL{e_1}{e_2}}{U}{\Gamma_3} \div (x_\ell)_{\ell\in L}
    }
    \\
    % Unit lin destructor
    \changed{\infrule{\rulenamealgtypeUnitElim}{
      \algtypein{\Gamma_1}{e_1}{\unitk}{\Gamma_2}
      \\
      \algtypeout{\Gamma_2}{e_2}{T}{\Gamma_3}
    }{
      \algtypeout{\Gamma_1}{\unlete {\unite}{e_1}{e_2}}{T}{\Gamma_3} 
    }}
    \\ % Sum contructor
    \infrule{\rulenamealgtypeVariant}{
      \algKindIn {\variantt \ell T L} \kindtl
      \\
      \algtypein{\Gamma_1}{e}{T_k}{\Gamma_2}
      \\
      k \in L
      % \\
      % \changed{(\forall\ell\in L)} % l is bound only
    }
    {\algtypeout{\Gamma_1}{\injecte {k^{\variantt \ell T L}} e}{\variantt \ell T L}{\Gamma_2}}
    \\ % Sum destructor (case)
    \infrule{\rulenamealgtypeCase}{
      \algtypeoutnorm{\Gamma_1}{e}{\variantt \ell T L}{\Gamma_2}
      % \\
      % \normalisation T {\variantt \ell T L}
      \\
      \algtypeoutnorm{\Gamma_2}{e_\ell}{\funt[\changed{\lin}]{T_\ell}{U_\ell}}{\Gamma_\ell}
      \\
 %     \typeequiv{U_k}{U_\ell}    
      \algequivin{U_k}{U_\ell}    
      \\
      \algequivin{\Gamma_k}{\Gamma_\ell} 
%      \ctxequiv{\Gamma_k}{\Gamma_\ell} 
      \\
      k\in L
      \\
      (\forall\ell\in L)
    }{
      \algtypeout{\Gamma_1}{\casee{e}{\{\ell\rightarrow e_\ell\}_{\ell\in L}}}{U_k}{\Gamma_k}
    }
    \\ % Channel creation
    \infrule{\rulenamealgtypeNew}
     {\algKindIn[\Empty] T \kindsl}{
      \algtypeout \Gamma {\newe{T}} {\pairt T {\dual T}} \Gamma
    }
    \\ % select
    \infrule{\rulenamealgtypeSel}
    {
      \algKindIn{\ichoicet{\recordt \ell T L}} \kindsl
      \\
      \changed{(k\in L)}
    }
    {\algtypeout \Gamma{\selecte {k^{\ichoicet{\recordt \ell T L}}} {}}{(\funt{\ichoicet{\recordt \ell T L}}{T_k})}
      \Gamma}
    \\
    % External Choice destructor (Match)
    \infrule{\rulenamealgtypeMatch}{
      \algtypeoutnorm{\Gamma_1} e {\echoicet {\recordt \ell T L}} {\Gamma_2}
      % \\
      % \normalisation T {\echoicet {\recordt \ell T L}}
      \\
      \algtypeoutnorm{\Gamma_2}{e_\ell}{\funt[\changed{\lin}]{T_\ell}{U_\ell}}{\Gamma_\ell}
      \\
%      \typeequiv{U_k}{U_\ell}
      \algequivin{U_k}{U_\ell}
      \\
%      \ctxequiv{\Gamma_k}{\Gamma_\ell}
      \algequivin{\Gamma_k}{\Gamma_\ell}
      \\
      k\in L
      \\
      (\forall\ell\in L)
    }{
      \algtypeout{\Gamma_1}{\matche e {\{\ell
          \rightarrow e_\ell\}_{\ell\in L}}}{U_k}{\Gamma_k}
    }    
  \end{gather*}
  \declrel{Type check-against}{$\algtypein[\Delta_\In]{\Gamma_\In}{e_\In}{T_\In}{\Gamma_\Out}$}
  \begin{gather*}
    \infrule{\rulenameAlgtypeEq}{
      \algtypeout {\Gamma_1} e U {\Gamma_2}
      \\
      \algequivin T U
    }{
      \algtypein {\Gamma_1} e T {\Gamma_2}
    }
  \end{gather*}
   \caption{Algorithmic type checking (continued, synthesis and check against)}
  \label{fig:alg-typing-two}
\end{figure}

%%% Local Variables:
%%% mode: latex
%%% TeX-master: "main"
%%% End:

The typing rules in \cref{fig:typing-exps} provide for a declarative description
of well-typed programs. They cannot however be directly implemented due to two
difficulties:
\begin{itemize}
\item The non-deterministic nature of context split, and
\item The necessity of guessing the type for the injection expression operators,
  rules \rulenametypeVariant and \rulenametypeSel in \cref{fig:typing-exps}.
\end{itemize}

We address the first difficulty by restructuring the typing rules so that we do
not have to guess the correct split.
%  leads to a practical algorithm. The key idea is
% that
Rather than guessing how to split the context for each subexpression, we use
the entire context when checking the first subexpression, obtaining as a
result the unused
part, which is then passed to the next subexpression~\cite{DBLP:journals/jfp/Mackie94,DBLP:conf/sfm/Vasconcelos09,walker:substructural-type-systems}.
To avoid guessing the type for the two injection expressions, we assume these
are explicitly typed, as in $\selecte {k^{\ichoicet{\recordt \ell T L}}} e$ and
$\injecte {k^{\variantt \ell T L}} e$.

The type synthesis and check-against rules are in
\cref{fig:alg-typing,fig:alg-typing-two}. The judgements are now of form
$\algtypeout[\Delta_\In]{\Gamma_\In}{e_\In}{T_\Out}{\Gamma_\Out}$ for the
synthesis relation, and
$\algtypein[\Delta_\In]{\Gamma_\In}{e_\In}{T_\In}{\Gamma_\Out}$ for the check
against relation.
\changed{We abbreviate type synthesis $\algtypeout{\Gamma_1}{e}{T}{\Gamma_2}$
  followed by normalisation $\normalisation TU$ as
  $\algtypeoutnorm{\Gamma_1}{e}{U}{\Gamma_2}$.}
Most of the rules are simple adaptations of the declarative typing rules to a
bidirectional context. We describe the most relevant.
Rule \rulenamealgtypeAbsLin synthesises the kind of $T_1$. Then, it synthesises
a type for the body of the abstraction under the extended context
$\Gamma_1, x \colon T_1$. The result is a type and a context $\Gamma_2$, from
which the variable $x$ should be removed, represented by $\Gamma_2 \div x$. The
rule \rulenamealgtypeAbsUn is similar but it checks that no linear resources in
the original context $\Gamma_1$ are used, a constraint ensured by
$\Gamma_1 = \Gamma_2 \div x$.

The rule \rulenamealgtypeApp first synthesises type $T$ from expression $e_1$.
Type normalisation exposes the top-level type constructor from type $T$ which is
expected to be a function type $\funt{T}{U}$. Then, it checks function $e_2$
against argument $T$; the result is type $U$.
Rule TA-Case starts by synthesising a type T for expression $e$ that normalises
to a type of the form $\variantt \ell T L$. The rule then synthesises a type for
each branch; they should all be functions $\funt[\lin]{T_\ell}{U_\ell}$. To
ensure that all branches are equivalent, the rule selects a label $k$ in $L$ and
checks that $\algequivin{U_k}{U_\ell}$ for all $\ell$ in $L$. To ensure that all
branches consume the same linear resources, the rule also checks that the final
contexts are equivalent, $\algequivin{\Gamma_k}{\Gamma_\ell}$.

%  rules \rulenamealgtypeVariant and
% \rulenametypeSel in \cref{fig:typing-exps}---we assume that each label comes
% equipped with its type: sum-labels are explicitly annotated with a sum type,
% $k^{\variantt \ell T L}$, and choice-labels are explicitly annotated with an
% internal choice type, $k^{\ichoicet{\recordt \ell T L}}$. In either case we
% assume $k\in L$.

\subsection{Correctness of Algorithmic Type Checking}

\paragraph{Soundness}

% The proof of the main result can be broken into two standard parts:
% \emph{soundness} and \emph{completeness} of the algorithmic system regarding the
% declarative system. We first define and prove algorithmic monotonicity,
% algorithmic linear strengthening and then, prove the first half of the result:
% \emph{typing soundness}. Then we prove algorithmic weakening and the second half
% of the result: \emph{typing completeness}.

Algorithmic monotonicity relates the output against the input of the algorithmic
type checking, and it is broadly used in the remaining proofs.

\begin{lemma}[Algorithmic monotonicity]\
  \begin{enumerate}
  \item\label{itm:synthesis-mono} If $\algtypeout{\Gamma_1}{e}{T}{\Gamma_2}$, then
    $\mathcal U(\Gamma_1)=\mathcal U(\Gamma_2)$ and
    $\mathcal L(\Gamma_2) \subseteq \mathcal L(\Gamma_1)$.
  \item\label{itm:checkagainst-mono} If $\algtypein{\Gamma_1}{e}{T}{\Gamma_2}$,
    then $\mathcal U(\Gamma_1)=\mathcal U(\Gamma_2)$ and
    $\mathcal L(\Gamma_2) \subseteq \mathcal L(\Gamma_1)$.
  \end{enumerate}
  \label{lem:monotonicity}
\end{lemma}
\begin{proof}
  The proof follows by mutual rule induction on the hypothesis, using the
  properties of context difference (\cref{lem:ctx-diff}). For
  \cref{itm:synthesis-mono}, we show the case when the derivation ends with
  \rulenamealgtypeAbsLin. To show that
  $\mathcal U(\Gamma_1) = \mathcal U(\Gamma_2 \div x)$ and
  $\mathcal L(\Gamma_2 \div x) \subseteq \mathcal L(\Gamma_1)$, we start by
  applying induction hypothesis on the second premise and get
  $\mathcal U(\Gamma_1, x\colon T) = \mathcal U(\Gamma_2)$ and
  $\mathcal L(\Gamma_2) \subseteq \mathcal L(\Gamma_1, x\colon T)$. Then, we
  have two cases regarding the un/lin nature of $T$. Consider the case in which
  we have $\algKindOut{T}{\kindtu}$. We know that
  $\mathcal U(\Gamma_1, x\colon T) \div x = \mathcal U(\Gamma_2) \div x$ and so,
  $\mathcal U(\Gamma_1) = \mathcal U(\Gamma_2 \div x)$. Since we have that
  $\mathcal L(\Gamma_2) = \mathcal L(\Gamma_2 \div x) \subseteq \mathcal
  L(\Gamma_1, x\colon T) = \mathcal L(\Gamma_1)$, we conclude
  $\mathcal L(\Gamma_2 \div x) \subseteq \mathcal L(\Gamma_1)$. The remaining
  cases are similar.
\end{proof}

Algorithmic linear strengthening allows removing unused linear assumptions from
both the input context and the output contexts. % In the proof of soundness we
% show that the first subexpression is typable under a context that does not
% contain the linear entries in $\Gamma_\Out$, which are then used to type the
% second subexpression.

\begin{lemma}[Algorithmic linear strengthening]\
  Suppose that $\istype \Delta {U} {\kindtl}$.
  \begin{enumerate}
  \item If $\algtypeout{\Gamma_1, x \colon U}{e}{T}{\Gamma_2, x \colon U}$,
    then $\algtypeout{\Gamma_1}{e}{T}{\Gamma_2}$.
  \item If $\algtypein{\Gamma_1, x \colon U}{e}{T}{\Gamma_2, x \colon U}$,
    then $\algtypein{\Gamma_1}{e}{T}{\Gamma_2}$.
  \end{enumerate}
  \label{lem:alg-strengthening}
\end{lemma}
\begin{proof}
  \sloppy
  The proof is by mutual rule induction on both hypotheses and uses the
  properties of context difference (\cref{lem:ctx-diff}) and monotonicity
  (\cref{lem:monotonicity}). Let us show one example. When the derivation ends
  with rule \rulenamealgtypeApp, suppose that
  $\algtypeout{\Gamma_1, x\colon U}{\appe{e_1}{e_2}}{T_2}{\Gamma_3, x\colon U}$.
  The premises to the rule are
  \begin{enumerate*}[label=\emph{(\alph*)}]
  \item$\algtypeoutnorm{\Gamma_1,x\colon U}{e_1}{\funt {T_1}{T_2}}{\Gamma_2}$ and
  \item\label{itm:pre2}$\algtypein{\Gamma_2}{e_2}{T_1}{\Gamma_3, x\colon U}$
  \end{enumerate*}. The first premise is an
  abbreviation for
  \begin{enumerate*}[label=\emph{(\alph*)}, resume]
    \item\label{itm:pre1}$\algtypeout{\Gamma_1,x\colon U}{e_1}{T}{\Gamma_2}$ and
    \item\label{itm:res}$\normalisation T{\funt {T_1}{T_2}}$
   \end{enumerate*}. Monotonicity gives
  $\mathcal L (\Gamma_3,x\colon U) \subseteq \mathcal L (\Gamma_2)$ and thus
  $x\colon {T_1} \in \Gamma_2$. Let $\Gamma_2 = \Gamma_2',x\colon U$.
  Induction on~\cref{itm:pre1,itm:pre2} yields
  \begin{enumerate*}[label=\emph{(\alph*)}, resume]
  \item\label{itm:res2}$\algtypein{\Gamma_2'}{e_2}{T_1}{\Gamma_3}$ and
  \item\label{itm:res1}$\algtypeout{\Gamma_1}{e_1}{T}{\Gamma_2'}$
  \end{enumerate*}, respectively. Conclude with rule
  \rulenamealgtypeApp and~\cref{itm:res,itm:res1,itm:res2}.
  %
  % The premises to the rule are
  % $\algtypeout{\Gamma_1, x\colon T'}{e_1}{T}{\Gamma_2}$ and
  % $\algtypein{\Gamma_2}{e_2}{U}{\Gamma_3, x\colon T'}$. Monotonicity gives
  % $\mathcal L (\Gamma_3,x\colon T') \subseteq \mathcal L (\Gamma_2)$ and thus
  % $x\colon T' \in \Gamma_2$. Let $\Gamma_2 = \Gamma_2',x\colon T'$.
  % Induction on the first and third premises
  % yield $\algtypeout{\Gamma_1}{e_1}{T}{\Gamma_2'}$ and
  % $\algtypein{\Gamma_2'}{e_2}{U}{\Gamma_3}$, respectively. Conclude with rule
  % \rulenamealgtypeApp.
\end{proof}
%% \algtypeoutnorm{\Gamma_1}{e_1}{\funt TU}{\Gamma_2}

% We are now in conditions to prove the first half of the result: \emph{typing
%   soundness}. % The proof below is by mutual rule induction, even if we present
% % the results separately, for the sake of simplicity.

\begin{theorem}[Algorithmic soundness]\ Suppose that $\isCUn{\Gamma_2}$.
  \begin{enumerate}
  \item %[Synthesis]
    \label{it:algsoundsynt}
    If $\algtypeout{\Gamma_1}{e}{T}{\Gamma_2}$,
    then $\isExpr{\Gamma_1}{e}{T}$.
  \item %[Check against]
    \label{it:algsoundcheck}
    If $\algtypein{\Gamma_1}{e}{T}{\Gamma_2}$,
    then $\isExpr{\Gamma_1}{e}{T}$.
  \end{enumerate}
\end{theorem}
\begin{proof}
  By mutual rule induction on the hypotheses.
  
  \Cref{it:algsoundsynt}. When the derivation
  ends with \rulenamealgtypeVarLin, \rulenamealgtypeVarUn, \rulenamealgtypeNew,
  or \rulenamealgtypeSel we use Algorithmic Kinding Soundness
  (\cref{lem:kinding-correctness}) and the corresponding expression formation rules. Rules \rulenamealgtypeApp,
  \rulenamealgtypeRcd, \rulenamealgtypeLet, \rulenamealgtypeCase, and
  \rulenamealgtypeMatch require Monotonicity (\cref{lem:monotonicity}) and
  Algorithmic Linear Strengthening (\cref{lem:alg-strengthening}) to push
  through the result.
  
  We detail the case for rule \rulenamealgtypeCase. We strengthen the first
  premise to obtain
  $\algtypeout{\Gamma_1-\mathcal L({\Gamma_2})}{e}{T}{\Gamma_2-\mathcal
    L({\Gamma_2})}$, then by induction hypothesis we get
  $\isExpr{\Gamma_1-\mathcal L({\Gamma_2})}{e}{T}$. Properties of type
  normalisation (\cref{lem:normProps}) yield
  $\isequiv[\Empty]\Delta{T}{\variantt \ell T L}{\kindt^m}$. Then, by applying
  rule \rulenametypeEq, we get
  $\isExpr{\Gamma_1-\mathcal L({\Gamma_2})}{e}{\variantt \ell T L}$. Using
  agreement for type equivalence lifted to typing contexts we get
  $\isCUn {\Gamma_\ell}$, given that $\isCUn {\Gamma_k}$ and
  $\isequiv[\Empty] {\Gamma_k}{U_k}{U_\ell}\kindtu$. Then, we can apply induction
  hypothesis to the second premise, obtaining
  $\isExpr{\Gamma_2}{e_\ell}{\funt[\lin]{T_\ell}{U_\ell}}$. Since
  $\isequiv[\Empty]\Gamma{U_k}{U_\ell}\kindtu$, we get
  $\isExpr{\Gamma_2}{e_\ell}{\funt[\lin]{T_\ell}{U_k}}$ for $k \in L$. Before
  concluding with rule \rulenametypeCase we need to establish that
  $\Gamma_1 = (\Gamma_1-\mathcal L({\Gamma_2})) \circ \Gamma_2$ is in the
  context split relation. For the unrestricted portion monotonicity yields
  $\mathcal U(\Gamma_1 - \mathcal L (\Gamma_2)) = \mathcal U (\Gamma_2 -
  \mathcal L (\Gamma_2))$ which gives
  $\mathcal U(\Gamma_1) = \mathcal U (\Gamma_2)$. Regarding the linear entries
  we know that
  $\mathcal L (\Gamma_1 - \mathcal L (\Gamma_2)) \cap \mathcal L (\Gamma_2) =
  \emptyset$. Hence,
  $\mathcal L (\Gamma_1) = \mathcal L(\Gamma_1 - \mathcal L (\Gamma_2)) \cup
  \mathcal L (\Gamma_2)$.

  \Cref{it:algsoundcheck}. From the hypothesis, using \rulenameAlgtypeEq, we know that
  $\algtypeout {\Gamma_1} e U {\Gamma_2}$ and $\algequivin T U$. By~\cref{it:algsoundsynt},
  we get $\isExpr{\Gamma_1}{e}{U}$.
  % By type equivalence
  % soundness (\cref{thm:type-equivalence}) we obtain $\isequiv\Delta TU\_$.
  Conclude
  with \rulenametypeEq.
\end{proof}

\paragraph{Completeness}

Weakening allows adding new entries to both the input and output contexts.

% The following lemma is required on the \emph{completeness} proof. As we can
% strengthen contexts by removing variables, we can also weaken them by adding
% variables, regardless of their multiplicities. We assume that all the added
% variables do not belong to the context.

\begin{lemma}[Algorithmic weakening]\
  \begin{enumerate}
  \item
    \label{it:alg-weak-it1}
    If $\algtypeout{\Gamma_1}{e}{T}{\Gamma_2}$, then
    $\algtypeout{\Gamma_1, x \colon U}{e}{T}{\Gamma_2, x \colon U}$.
  \item If $\algtypein{\Gamma_1}{e}{T}{\Gamma_2}$, then
    $\algtypein{\Gamma_1, x \colon U}{e}{T}{\Gamma_2, x \colon U}$.
  \end{enumerate}
  \begin{proof}
    By mutual rule induction on the hypotheses. We show the case
    \rulenamealgtypeAbsLin\ for ~\cref{it:alg-weak-it1}. By induction hypothesis 
    we know that
    $\algtypeout{\Gamma_1,x\colon U,y\colon V}{e}{T_2}{\Gamma_2,y\colon V}$.
    From the properties of context difference
    (\cref{fig:context-difference}), we have
    $(\Gamma_2 \div x), y\colon V = (\Gamma_2, y\colon V) \div x$. Conclude
    applying the rule \rulenamealgtypeAbsLin.
  \end{proof}
\end{lemma}

% We are now able to prove \emph{completeness} which is the second part of the
% result and the reverse of \emph{soundness}.

\begin{theorem}[Algorithmic completeness]\
  \begin{enumerate}
  \item %[Synthesis]
    \label{it:algcompsynth}
    If $\isExpr{\Gamma_1}{e}{T}$, then $\algtypeout{\Gamma_1}{e}{U}{\Gamma_2}$
    and $\isCUn{\Gamma_2}$ and $\isequiv[\Empty]{\Delta}{T}{U}{\kind}$.
  \item %[Check against]
    \label{it:algcompcheck}
    If $\isExpr{\Gamma_1}{e}{T}$, then $\algtypein{\Gamma_1}{e}{T}{\Gamma_2}$
    and $\isCUn{\Gamma_2}$.
  \end{enumerate}
\end{theorem}
\begin{proof}  
  By mutual rule induction on the hypotheses.

  \sloppy\Cref{it:algcompsynth}. We detail the case for \rulenametypeCase. Induction on
  the first premise gives
  \begin{enumerate*}[label=\emph{(\alph*)}]
  \item\label{itm:synth1}$\algtypeout{\Gamma_1}{e}{U}{\Gamma_2'}$,
  \item\label{itm:ctxUn1}$\isCUn{\Gamma_2'}$, and
  \item\label{itm:equivHyp1}$\isequiv[\Empty]{\Delta}{U}{\variantt \ell T L}{\kind}$
  \end{enumerate*}. Rules of \cref{fig:type-equivalence,fig:norm}
  with~\ref{itm:synth1} and \ref{itm:equivHyp1} yield
  $\normalisation U {\variantt \ell T L}$ and thus
  \begin{enumerate*}[label=\emph{(\alph*)}, resume]
  \item\label{itm:synthNorm1}$\algtypeoutnorm{\Gamma_1}{e}{\variantt \ell TL}{\Gamma_2'}$
  \end{enumerate*}.
  Induction on the second premisse of \rulenametypeCase yields 
  \begin{enumerate*}[label=\emph{(\alph*)}, resume]
  \item\label{itm:synth2}$\algtypeout{\Gamma_2}{e}{V}{\Gamma_\ell}$,
  \item\label{itm:ctxUn2}$\isCUn{\Gamma_\ell}$, and
  \item\label{itm:equivHyp2}$\isequiv[\Empty]{\Delta}{V}{\funt[\lin]{T_\ell}T}{\kind}$
  \end{enumerate*}. Again, rules of~\cref{fig:type-equivalence,fig:norm}
  with~\ref{itm:synth2} and~\ref{itm:equivHyp2} yield
  $\normalisation V {\funt[\lin]{T_\ell}T}$ and thus
  \begin{enumerate*}[label=\emph{(\alph*)}, resume]
  \item\label{itm:synthNorm2}$\algtypeoutnorm{\Gamma_2}{e_\ell}{\funt[\lin]{T_\ell}T}{\Gamma_\ell}$
  \end{enumerate*}.
  We have to prove that $\Gamma_2' = \Gamma_2$. Since
  $\Gamma_1 \cup \mathcal L({\Gamma_2}) = \Gamma_1 \circ \Gamma_2$, we
  weaken~\ref{itm:synthNorm1} to obtain
  $\algtypeoutnorm{\Gamma_1 \circ \Gamma_2}{e}{\variantt \ell TL}{\Gamma_2' \cup \allowbreak{\mathcal L({\Gamma_2})}}$.
  We have to show that
  $\Gamma_2' \cup \mathcal L({\Gamma_2}) = \Gamma_2$. We know that
  $\mathcal L(\Gamma_2' \cup \mathcal L({\Gamma_2})) = \mathcal L({\Gamma_2})$
  since $\mathcal L({\Gamma_2'}) = \emptyset$. From monotonicity and properties
  of context split (lemmas~\ref{lem:monotonicity} and~\ref{lem:cSplitProps}), we
  know that
  $\mathcal U({\Gamma_2'} \cup \mathcal L({\Gamma_2})) = \mathcal U({\Gamma_2'})
  = \mathcal U({\Gamma_1\circ\Gamma_2}) = \mathcal U({\Gamma_2})$.
  From~\ref{itm:synthNorm2} and reflexivity of type equivalence
  (\cref{lem:type-equivalence}) we obtain
  $\isequiv{\Delta}{T}{T}{\kind}$.
  Before concluding with rule \rulenamealgtypeCase we still have to show that
  $\algequivin{\Gamma_k}{\Gamma_\ell}$. From monotonicity
  and~\ref{itm:synthNorm2} we get
  $\mathcal U(\Gamma_2) = \mathcal U(\Gamma_\ell)$. Since $\isCUn{\Gamma_\ell}$
  we can pick a $k\in L$ such that
  $\isequiv[\Empty]{\Delta}{\Gamma_k}{\Gamma_\ell}{\kindtu}$. By subkinding
  we get $\algequivin{\Gamma_k}{\Gamma_\ell}$.

  \Cref{it:algcompcheck}. From~\cref{it:algcompcheck} we know that
  $\algtypeout{\Gamma_1}{e}{U}{\Gamma_2}$ and $\isCUn{\Gamma_2}$ and $\algequivin
  TU$. Conclude with rule \rulenameAlgtypeEq.
\end{proof}

\paragraph*{Pragmatics and \freest}

% The language developed in \cref{sec:processes} introduces $\sendk$ and
% $\receivek$ as polymorphic functions (\cref{fig:types-constants}). Having such
% polymorphic functions requires type applications in each $\sendk$ and
% $\receivek$ operations. For example, to send the integer value 5 on a channel
% $c$ that terminates thereafter, programmers would have to write
% $\tappe{\tappe{\sendk}{\intk}}{\skipk}\ 5\ c$; rather cumbersome. Instead, in
% the absence of type applications, \freest infers both types of the channel and
% of the transmitted value. If one provides the type applications, then type
% checking continues accordingly with the rule \rulenamealgtypeTApp. For the
% remaining polymorphic functions, type applications still have to be provided, a
% situation that may be amended with local type
% inference~\cite{DBLP:books/daglib/0005958,DBLP:journals/toplas/PierceT00}.

In the injection rules, \rulenamealgtypeVariant and \rulenamealgtypeSel
expressions are type-annotated. \freest handles this by requiring variants to be
declared in advance so that each datatype constructor becomes a function symbol
in the symbol table, as usual in functional programming languages where variants
and records are treated nominally.

The check against relation, defined in \cref{fig:alg-typing-two}, may include
more rules to produce better error messages. The following two rules are
derivable, and we use them in the compiler.
\begin{gather*}
  % Pair destructor
  \infrule{}{
    \algtypeout{\Gamma_1}{e_1}{T_1}{\Gamma_2}
    \\
    \normalisation {T_1}{\pairt{U_1}{U_2}}
    \\
    \algtypein{\Gamma_2, x\colon U_1, y\colon U_2}{e_2}{T_2}{\Gamma_3}
  }{
    \algtypein{\Gamma_1}{\binlete{x}{y}{e_1}{e_2}}{T_2}{\Gamma_3} \div x \div y
  }
  \\
  % Conditional
  \infrule{}{
    \algtypein{\Gamma_1}{e_1}{\boolt}{\Gamma_2}
    \\
    \algtypein{\Gamma_2}{e_2}{T}{\Gamma_3}
    \\
    \algtypein{\Gamma_2}{e_3}{T}{\Gamma_4}
    \\
    \ctxequiv{\Gamma_3}{\Gamma_4}
  }{
    \algtypein{\Gamma_1}{\conde{e_1}{e_2}{e_3}}{T}{\Gamma_3}
  }  
\end{gather*}

Similarly to the translation to BPA (\cref{sec:type-equiv-decidable}), the
algorithm to decide type equivalence~\cite{DBLP:conf/tacas/AlmeidaMV20} requires
$\mu$-introduced type variables to occur free in the body of the type, hence
\freest replaces types of the form $\rect{a}{\kind}{T}$ by $T$ when
$a\notin\free(T)$.
%
% A related problem is keeping programs renamed.
Furthermore, programs must be kept renamed so that the types inspected by type
equivalence do not share bound variables.
Substitution preserves neither of these two properties. The current version of
\freest ensures that the properties remain invariant throughout type checking by
renaming and replacing $\mu$-types at parsing and whenever it performs a
substitution.

\section{Related Work}
\label{sec:related}

The related work on session types is vast. The BETTY book
\cite{gay17:_behav_types} provides a recent overview. Here we concentrate on
direct influences to our design, systems that support polymorphism, and systems
that go beyond tail recursive types for communication channels.

\paragraph{Binary Session Types}

Session types were introduced by
Honda \etal~\cite{DBLP:conf/concur/Honda93,DBLP:conf/esop/HondaVK98,DBLP:conf/parle/TakeuchiHK94}
on variants of the $\pi$-calculus. These early ideas were later incorporated in
functional languages by Gay
\etal~\cite{DBLP:journals/jfp/GayV10,DBLP:journals/tcs/VasconcelosGR06,DBLP:conf/concur/VasconcelosRG04},
who also enunciate the importance of linearity and linear types for
session types.
Our language is particularly close to that of Gay and
Vasconcelos~\cite{DBLP:journals/jfp/GayV10}, which owes much of its
elegance and generality to its reliance on linear types. At the level of the operational
semantics we use a synchronous semantics for communication in place of
a buffered asynchronous one. This
choice simplifies the technical treatment, even if we believe that a buffered
semantics can be reintroduced without compromising the
metatheory of the language. At the level of types we incorporate record
and variant  types, general recursive types (not restricted to session
types), a sequential composition operator for session types, and
impredicative polymorphism.
We omit subtyping as it requires a significant treatment on its own.
% The linear treatment of session types is similar.
%
Vasconcelos~\cite{DBLP:journals/iandc/Vasconcelos12,DBLP:conf/sfm/Vasconcelos09}
introduces the syntactic distinction of the two ends of a channel, related by a
single $\nu$-binding, a technique we follow in this paper.
%\todo{Review \cite{DBLP:conf/esop/CairesPPT13}}

% The Sill language described by Toninho, Caires, and
% Pfenning~\cite{DBLP:conf/esop/ToninhoCP13}
Languages with conventional (that is, regular) recursive session
types~\cite{DBLP:conf/esop/CairesPPT13,DBLP:journals/jfp/GayV10,DBLP:conf/concur/Honda93,DBLP:conf/esop/HondaVK98,DBLP:conf/parle/TakeuchiHK94,DBLP:journals/eatcs/Vasconcelos11,DBLP:journals/iandc/Vasconcelos12,DBLP:conf/sfm/Vasconcelos09,DBLP:journals/tcs/VasconcelosGR06,DBLP:conf/concur/VasconcelosRG04}
can still describe some of the protocols discussed in this paper, but they
require higher-order session types (the ability to send channels on channels).
For example, the session type \lstinline|TreeChannel| introduced in
\cref{sec:tree-struct-transm} can be emulated by the type
\begin{equation*}
    \mu a.\oplus\!\{\leafl\colon \Endk, \nodel\colon !\intk.!a.!a\}
\end{equation*}
% \begin{lstlisting}
% TreeC = oplus{Leaf: end, Node: !int.!TreeC.!TreeC}
% \end{lstlisting}
where two new channels must be created to transmit the two
subtrees at every node. In comparison, our calculus is intentionally closer to a low level
language: it only supports the transmission of non-session typed values.
Furthermore its implementation is simpler and more
efficient: only one channel is created and used for the transmission of the
whole tree; thus, avoiding the overhead of multiple channel creation.

\paragraph{Context-free Session Types}

% Inference, Padovani
%
As an alternative to checking type equivalence for arbitrary context-free types,
Padovani~\cite{DBLP:conf/esop/Padovani17,DBLP:journals/toplas/Padovani19}
proposes a language that requires explicit annotations in the source code.
These annotations result in the structural alignment between code and types, thus
simplifying the problem of type checking.
Our system requires no annotations and relies directly on the
algorithm for type equivalence described in \cref{sec:type-equiv-decidable}.
% As a consequence, there are basic equivalences on types that Padovani's compiler is not able to identify.
%
% CFST for applied pi
%
Aagaard \etal~\cite{DBLP:journals/corr/abs-1808-08648} adapt the concept of
context-free session types from our previous work
\cite{DBLP:conf/icfp/ThiemannV16} to the applied pi calculus. They
also account for unrestricted (that is, arbitrarily shared) channels in the style of
Vasconcelos~\cite{DBLP:journals/iandc/Vasconcelos12},
a feature that we decided not to include in this paper. They prove
session fidelity by mapping their calculus to the psi calculus and
defining types up to equivalence with respect to session type
bisimulation.
Gay \etal recently showed that the equivalence for higher-order context-free
session types is decidable~\cite{gay.etal:shades}.

\paragraph{Context-free Protocols}

Early attempts to describe non-regular protocols as types include proposals by
Ravara and Vasconcelos~\cite{DBLP:conf/europar/RavaraV97} using directed labelled
graphs as types, by Puntigam~\cite{DBLP:conf/europar/Puntigam99} on non-regular
process types to describe component interfaces, and by
S\"udholt~\cite{DBLP:conf/soco/Sudholt05} who investigates some
use cases with an approach based on Puntigam's types. Ravara and
Vasconcelos define a concurrent object calculus and establish some
metatheoretical results culminating in subject reduction for their type system, but do not
address type checking. Puntigam studies type equivalence and
subtyping, claims decidability for both, and gives a trace-based
characterization. No details are given on the underlying
calculus.
These proposals appear in disparate formalisms based on process
calculi and none of them supports polymorphism, so they are quite
distinct from the system in this paper, which is based on System F.
Context-free session types were
introduced by Thiemann and Vasconcelos~\cite{DBLP:conf/icfp/ThiemannV16}; we
follow their formulation and extend it rigorously to System F.
%Context-free session types capture communication
%patterns represented by deterministic context-free languages recognized by
%single-state deterministic pushdown automata that accept by empty stack.
Das \etal~\cite{DBLP:conf/esop/DasDMP21,das2021subtyping} propose the theory of
nested session types, featuring recursive types with polymorphic type
constructors and nested types. Nested session types enable expressing protocols
described by context-free languages recognized by multi-state deterministic
pushdown automata, whereas context-free session types capture communication
patterns recognized by single-state deterministic pushdown
automata~\cite{gay.etal:shades}. Although covering a more restrictive class of
context-free languages, context-free session types enjoy a practical, sound and
complete, type equivalence algorithm~\cite{DBLP:conf/tacas/AlmeidaMV20}. In
turn, type equivalence for nested session types is decidable, the authors
present a sound equivalence algorithm, but a complete and practical type
equivalence algorithm is not yet known~\cite{DBLP:conf/esop/DasDMP21}.

\paragraph{Linear Logic Propositions as Sessions}

Caires and
Pfenning~\cite{DBLP:conf/concur/CairesP10,DBLP:journals/mscs/CairesPT16}
pioneered the interpretation of intuitionistic linear logic propositions as
sessions. This interpretation, indicated by $[\_]$, maps receiving and
sending session types for base types $T$ as payload to
tensor and lollipop as follows:
% tensor, lollipop, and the unit type as follows:
\begin{align*}
  [!T.U] & = T \otimes [U] & [?T.U] & = T \multimap [U]
  % & [\Endk] & = \mathbf{1}
\end{align*}
Base types $T$ need not be interpreted.
The semantics
of these systems, given directly by the cut elimination rules,
ensures deadlock freedom.
An analogous interpretation applies to our session type language, but
to compensate for the presence of the
sequencing operator, the interpretation must be formulated in CPS and
invoked as in $[S] (\mathbf{1})$:
\begin{align*}
  [!T] (U) &= T \otimes U & [?T] (U) & = T \multimap U
  & [T_1;T_2] (U) & = [T_1] ([T_2] (U))
\end{align*}

The interpretation of propositions as sessions is extend to dependent types
\cite{DBLP:conf/ppdp/ToninhoCP11}, encompassing polymorphism since the resulting
session type discipline includes explicit operators to send and receive types.

Wadler~\cite{DBLP:conf/icfp/Wadler12,DBLP:journals/jfp/Wadler14} proposes a
typing preserving translation of a variant of the aforementioned language by Gay and
Vasconcelos~\cite{DBLP:journals/jfp/GayV10} into a process calculus
that translates Caires and Pfenning's
work~\cite{DBLP:conf/concur/CairesP10,DBLP:journals/mscs/CairesPT16}
to classical linear logic, thus ensuring deadlock freedom. Even though
our system ensures progress for the
functional part of the language, the unrestricted interleaving of channel
read/write on multiple channels may lead to deadlocked situations. That is the
price to pay for the flexibility our language offers compared to the
accounts of session types based on linear logic.

% The logic-based approaches to session types do not support recursive
% types. There is no reason to discuss regular vs.\ context-free
% types without recursion.

\paragraph{Polymorphism for Session Types}

An early version of this work~\cite{DBLP:conf/icfp/ThiemannV16} uses predicative
(Damas-Milner~\cite{DBLP:conf/popl/DamasM82}) polymorphism, in a form closely
related to that offered by Bono \etal~\cite{DBLP:conf/forte/BonoPT13} for
functional session types. The current paper extends the predicative system to
System F, with the associated extra flexibility.
% Bono \etal associate qualifiers to polymorphic type variables,
% yet for a purpose different from ours. The extra complexity of context-free
% types lead us to a more elaborate kinding system, allowing to distinguish
% message from session and from functional types as well as linear from
% unrestricted types.
%

A different form of polymorphism---bounded polymorphism on the values
transmitted on channels---was introduced by
Gay~\cite{DBLP:journals/mscs/Gay08} in the realm of session types for
the $\pi$-calculus. The intention of this work was to gain flexibility
over subtyping: a protocol like $?\bytek; !\bytek$ is not a subtype of
$?\intk; !\intk$ because sending behaves contravariantly whereas receiving
is covariant.
Bono and Padovani~\cite{DBLP:conf/esop/BonoMP11,DBLP:journals/corr/abs-1108-0466}
present a calculus to model process
interactions based on copyless message passing. Their type system
includes subtyping and recursion. An extension
\cite{DBLP:journals/corr/abs-1202-2086} features polymorphic endpoint
types using bounded polymorphism in a similar way as
Gay~\cite{DBLP:journals/mscs/Gay08}. Dardha
\etal~\cite{DBLP:journals/iandc/DardhaGS17}
extend their encoding of session types into $\pi$-types with parametric
and bounded polymorphism, but recursive types are not considered;
Dardha~\cite{DBLP:journals/corr/Dardha14} extends the encoding to recursive types, but
polymorphism is only referred to as future work.
%

%%% DBLP:journals/mscs/GotoJJPR16
Goto \etal \cite{DBLP:journals/mscs/GotoJJPR16} study a type system
that provides polymorphism over sessions in the context of the $\pi$
calculus. Their calculus supports match processes which are guarded on
a check for token (in)equality and serve to implement the usual
branch session types.
The goal of their system is to enable the typing of
forwarders and session transducers using polymorphism and (session)
type functions provided by the programmer as part of the definition of
the type system. Their system is quite
different from ours (and from most other polymorphic type systems) in
that there is no syntax for quantification in their type
language. They rely on polymorphism in
their metalanguage (Coq), instead.

Wadler~\cite{DBLP:journals/jfp/Wadler14} supports polymorphism on session types,
introducing dual quantifiers, $\forall$ and $\exists$, interpreted as sending
and receiving types, similar to the polymorphic $\pi$-calculus by
Turner~\cite{DBLP:phd/ethos/Turner96}. Explicit higher-rank quantifiers on
session types were pursued by Caires \etal~\cite{DBLP:conf/esop/CairesPPT13,DBLP:journals/iandc/PerezCPT14}.
Caires \etal\ provide a logically motivated theory of parametric polymorphism
with higher-rank quantifiers but without recursive types.
Griffith explores the inclusion of parametric polymorphism without
nested types in the language SILL~\cite{griffith2016polarized}.
This paper considers recursive types but restricts polymorphism
to the functional layer (see rule \rulenamekindPoly in \cref{fig:kinding}),
leaving polymorphism via type passing to future work.

In recent work, Das \etal\ explore recursive session type definitions
with type parameters and nested type
instantiation~\cite{DBLP:conf/esop/DasDMP21}. Type quantification
and subtyping were added to this theory separately~\cite{das2021subtyping}.
In this paper, we only consider nested types in the functional layer;
incorporating explicit polymorphism and nested types on sessions
constitutes future work.

\paragraph{\freest and the Decidability of Type Equivalence}

Almeida \etal~\cite{DBLP:journals/corr/abs-1904-01284} describe the
implementation of the predicative version of \freest (\freest 1). At the core
of all \freest versions lies an algorithm to decide type equivalence.
Unfortunately the proof of decidability of type equivalence (in
\cref{sec:type-equivalence}, based on the reduction to the decidability of
bisimilarity of basic process algebras~\cite{DBLP:journals/iandc/ChristensenHS95}) does not directly lead to a
practical algorithm. Type equivalence in \freest (1 and 2) is based on the algorithm
developed by Almeida \etal~\cite{DBLP:conf/tacas/AlmeidaMV20}.
% based
% on ideas by Caucal
% \etal~\cite{DBLP:conf/stacs/Caucal86,DBLP:journals/iandc/ChristensenHS95,DBLP:journals/entcs/Hirshfeld96,DBLP:conf/concur/JancarM99}.

\paragraph{Linear and Recursive System F}

Bierman \etal~\cite{DBLP:journals/entcs/BiermanPR00} propose a polymorphic
linear lambda calculus with recursive types, supporting both call by value and
call by name semantics, and featuring a $!$ constructor to account for
intuitionistic terms.
%
% Zhao \etal~\cite{DBLP:conf/aplas/ZhaoZZ10} introduce a polymorphic linear
% lambda-calculus
Mazurak \etal~\cite{DBLP:conf/tldi/MazurakZZ10} use kinds to qualify linearity
of types in a variant of System~F with linear types. In their system, type
qualifiers are fixed as in our system. This approach was taken further by
Lindley and Morris~\cite{lindley17:_light_funct_session_types} in a
sophisticated kind structure with linear and unrestricted kinds as opposed to
type qualifiers (as in Walker~\cite{walker:substructural-type-systems} and in
Vasconcelos~\cite{DBLP:journals/iandc/Vasconcelos12}), the $\un \subk \lin$ kind
subsumption rule, distinct $\lin$/$\un$ arrow types (this distinction is already
present in Gay and Vasconcelos~\cite{DBLP:journals/jfp/GayV10} in the context of
functional session types, but their work does not support polymorphism). We
follow the latter approach path quite closely: we have an additional message
kind besides session and functional, whereas Lindley and Morris model records,
variants, and branching using row types. Moreover, Lindley and Morris support
polymorphism over kinds and rows, which we do not.

Alms \cite{DBLP:conf/popl/TovP11} is a polymorphic functional programming language
that supports affine types and modules. The mode of a type (affine or
unrestricted) is indicated with a kind that may depend on a type,
which serves the same purpose as Lindley and Morris' kind polymorphism
(which we do not support). Alms
comes with recursive algebraic datatypes, but does not offer equirecursive
types. A library implementation of (regular) session types without
recursion is among the examples provided with the Alms implementation.

% Cai \etal~\cite{DBLP:conf/popl/CaiGO16} introduce $F^\mu_\omega$, an extension of the higher-order
% polymorphic lambda calculus $F_\omega$ with records, variants, and equirecursive
% types.
% %
% With Cai \etal's proposal we share recursive
% types in equirecursive format, as opposed to other proposals that use weak type
% equivalence such as Bruce \etal~\cite{DBLP:journals/iandc/BruceCP99}
% (equirecursion is the norm with session
% types~\cite{DBLP:journals/acta/GayH05,DBLP:conf/esop/HondaVK98,DBLP:journals/iandc/Vasconcelos12}.)
% %
% We leave the study of higher-order
% polymorphism for future work. The absence of type level abstraction allows us to
% constrain $\mu$-types to be contractive via kinding and contractivity relations
% even in the presence of polymorphism (\cref{fig:kinding,fig:contractivity}).

\paragraph{API Protocols}

Ferles \etal~\cite{DBLP:journals/pacmpl/FerlesSD21} verify the correct use of
context-free API protocols by abstracting the program as a context-free grammar
(representing feasible sequences of API calls) and checking whether the thus
obtained grammar is included in that of a specification. Their tool is
restricted to LL(k) grammars, given the undecidability of inclusion checking for
arbitrary context-free languages. Our approach, instead, directly checks
programs against protocol specifications in the form of context-free types.

%%% Local Variables:
%%% mode: latex
%%% TeX-master: "main"
%%% End:

\section{Conclusion}
\label{sec:conclusion}

This paper proposes \fmusession, a functional language with support for
polymorphic equi-recursive context-free session types extended with
multi-threading and inter-process communication.
We prove that type equivalence is decidable for context-free session types, and
we present a bidirectional type checking algorithm that we prove correct with
respect to the declarative system.
The proposed algorithms are incorporated in a typechecker and interpreter for
\freest~\cite{freest}.

\freest has a past and a future. \freest 1 is based on a predicative type
system; essentially the proposal by Thiemann and
Vasconcelos~\cite{DBLP:conf/icfp/ThiemannV16}. \freest 2, the current release,
is a language based on this paper.
% In either version, the $\dualofk$
% operator is defined for recursive type variables but not for polymorphic type
% variables. \freest 3 is $F^\mu_\omega$ provided we manage to overcome all the
% hurdles. \freest 4 include sending and receiving types, even if we do not yet
% have compelling examples of its usefulness.
%
For future releases we plan to extend the \freest with higher-order kinds
(providing for a convenient $\dualofk$ type operator), higher-order channels
(allowing to pass channels on channels), and polymorphism over session types (in
addition to polymorphism over functional types). We also plan to explore type
inference at type application points and the incorporation of shared
channels~\cite{DBLP:journals/pacmpl/BalzerP17,DBLP:journals/pacmpl/RochaC21,DBLP:journals/iandc/Vasconcelos12}.

\changed{Type equivalence for higher-order channels poses interesting challenges, for
both the bisimulation-based declarative
definition~\cite{DBLP:conf/icfp/ThiemannV16} and its grammar-based algorithmic
version~\cite{DBLP:conf/tacas/AlmeidaMV20} rely on first-order messages. A
recent work by Costa \etal~\cite{costa.etal:polymorphic-hocfst} extends both the
bisimulation and the grammar construction to higher-order session types.
Furthermore, we do not expect difficulties in type duality for higher-order
types, given that the results in this paper rely on building a type dual to a
given type, a simple operation when contrasted with checking the duality of
session types.}

%Future: shared types as ``manifest sharing'' or as fundamentals (see applied pi calculus~\cite{DBLP:journals/corr/abs-1808-08648})

%%% Local Variables:
%%% mode: latex
%%% TeX-master: "main"
%%% End:

\paragraph{Acknowledgements}
We thank the anonymous reviewers for their detailed comments that greatly
contributed to improve the paper. This work was supported by FCT through project
SafeSessions, ref.\ PTDC/CCI-COM/6453/2020, and the LASIGE Research Unit, ref.\
UIDB/00408/2020 and ref.\ UIDP/00408/2020.

\bibliography{biblio}
\bibliographystyle{plain}

\end{document}

%%% Local Variables:
%%% mode: latex
%%% TeX-master: t
%%% End: